\documentclass[journal,comsoc]{IEEEtran}

\usepackage[T1]{fontenc}% optional T1 font encoding

\ifCLASSINFOpdf
\else
\fi
\usepackage{amsmath}
\usepackage{times}
\usepackage{algorithm}
\usepackage[noend]{algorithmic}
\usepackage[dvips]{graphicx}
\usepackage{caption}
\usepackage{subcaption}
\usepackage{verbatim}
\usepackage{amssymb}
\usepackage{cite}
\usepackage{epsfig}
\usepackage{color}
\usepackage[hyphens]{url}
\usepackage{amsthm}

\newtheorem{lemma}{Lemma}
\newtheorem{theorem}{Theorem}

\newcommand\txtblue[1]{{\color{black}#1}}
\newcommand\txtred[1]{{\color{black}#1}}
\newcommand\txtmag[1]{{\color{black}#1}}
\newcommand\txtredd[1]{{\color{black}#1}}
\newcommand\txtgreen[1]{{\color{black}#1}}
\usepackage[algo2e,lined,boxed,commentsnumbered]{algorithm2e}
\usepackage{url}
\usepackage{blindtext}

\begin{document}
%
% paper title
% Titles are generally capitalized except for words such as a, an, and, as,
% at, but, by, for, in, nor, of, on, or, the, to and up, which are usually
% not capitalized unless they are the first or last word of the title.
% Linebreaks \\ can be used within to get better formatting as desired.
% Do not put math or special symbols in the title.
\title{A Hybrid RF-VLC System for Energy Efficient Wireless Access}
%
%
% author names and IEEE memberships
% note positions of commas and nonbreaking spaces ( ~ ) LaTeX will not break
% a structure at a ~ so this keeps an author's name from being broken across
% two lines.
% use \thanks{} to gain access to the first footnote area
% a separate \thanks must be used for each paragraph as LaTeX2e's \thanks
% was not built to handle multiple paragraphs
%

\author{Abdallah~Khreishah,~\IEEEmembership{Senior,~IEEE,}Sihua~Shao,~\IEEEmembership{Student Member,~IEEE,}~Ammar~Gharaibeh,~\IEEEmembership{Member,~IEEE,}~        Moussa~Ayyash,~\IEEEmembership{Senior,~IEEE,}~Hany~Elgala,~\IEEEmembership{Member,~IEEE,}~and~Nirwan~Ansari,~\IEEEmembership{Fellow,~IEEE}% <-this % stops a space
\IEEEcompsocitemizethanks{\IEEEcompsocthanksitem Abdallah Khreishah, Sihua Shao, and Nirwan Ansari are with the Department of Electrical \& Computer Engineering, New Jersey Institute of Technology, Newark, NJ 07012 USA, Ammar Gharaibeh is with the Department of Computer Engineering, German Jordanian University, Amman, Jordan 11180, Moussa Ayyash is with the Department of Information Studies, Chicago State University, Chicago, IL 60628, and Hany Elgala is with the Department of Computer Engineering, State University of New York at Albany, Albany, NY 12222.\protect\\
% note need leading \protect in front of \\ to get a newline within \thanks as
% \\ is fragile and will error, could use \hfil\break instead.
E-mail: \{abdallah, ss2536, nirwan.ansari\}@njit.edu, ammar.gharaibeh@gju.edu.jo, mayyash@csu.edu, and
helgala@albany.edu.
}% <-this % stops an unwanted space
}

\maketitle \vspace{-0.9in}
% As a general rule, do not put math, special symbols or citations
% in the abstract or keywords.
\begin{abstract}
In this paper, we propose a new paradigm in designing and realizing energy efficient wireless indoor access
networks, namely, a hybrid system enabled by traditional RF access, such as WiFi, as well as the emerging
visible light communication (VLC). VLC facilitates the great advantage of being able to jointly perform
illumination and communications, and little extra power beyond illumination is required to empower
communications, thus rendering wireless access with almost zero power consumption. On the other hand, when
illumination is not required from the light source, the energy consumed by VLC could be more than that consumed
by the RF. By capitalizing on the above properties, the proposed hybrid RF-VLC system is more energy efficient
and more adaptive to the illumination conditions than the individual VLC or RF systems. To demonstrate the
viability of the proposed system, we first formulate the problem of minimizing the \txtgreen{power} consumption
of the hybrid RF-VLC system while satisfying the users requests and maintaining acceptable level of
illumination, which is NP-complete. Therefore, we divide the problem into two subproblems. In the first
subproblems, we determine the set of VLC access points (AP) that needs to be turned on to satisfy the
illumination requirements. Given this set, we turn our attention to satisfying the users' requests for real-time
communications, and we propose a randomized online algorithm that, against an oblivious adversary, achieves a
competitive ratio of $\log(N)\log(M)$ with probability of success $(1-\frac{1}{N})$, where $N$ is the number of
users and $M$ is the number of VLC and RF APs. We also show that the best online algorithm to solve this problem
can achieve a competitive ratio of $\log(M)$. Simulation results further demonstrate the advantages of the
hybrid system.
\end{abstract}

% Note that keywords are not normally used for peerreview papers.

% For peer review papers, you can put extra information on the cover
% page as needed:
% \ifCLASSOPTIONpeerreview
% \begin{center} \bfseries EDICS Category: 3-BBND \end{center}
% \fi
%
% For peerreview papers, this IEEEtran command inserts a page break and
% creates the second title. It will be ignored for other modes.
\IEEEpeerreviewmaketitle

\section{Introduction} \label{intro}
The access portion of the Internet is becoming predominantly wireless. Cellular networks and wireless indoor
networks, such as WiFi, are becoming the predominant choice of wireless access. Watching HD streaming videos and
accessing cloud-based services are the main user activities that rapidly consume the data capacity. They will
continue to be the major capacity consuming activities in the future. Moreover, activities linked with high data
rate wireless traffic are stationary and typically occur in fixed wireless access scenarios~\cite{Informa}.
Reported data indicate that a majority of IP-traffic usage occurs indoors (70\% indoors and only 30\% of the
traffic is served outdoors~\cite{GBI}). According to Ericsson report~\cite{ErR}, we spend 90\% of our time
indoors, and 80\% of the mobile Internet access traffic happens indoors~\cite{AlR,ciscoprovider}. This percentage
will only increase as 54\% of the cellular traffic is expected to be offloaded to WiFi by
2019~\cite{Webtrafficgrowth14}. It is expected that 87\% of the companies would switch cellular providers by 2019
for better indoor coverage~\cite{AlR}.

The rapid increase in wireless traffic is highly coupled with an increase in energy consumption. In 2011, energy
consumption of Internet in the US was larger than that of the whole automotive industry and about half that of the
chemical industry, as mentioned in a CNN report~\cite{CNNReport}. Furthermore, Internet traffic is estimated to
grow by a factor of 10 during 2013-2018~\cite{Webtrafficgrowth}, and thus incurs more energy consumption.
According to~\cite{NUMWiFi}, the number of only public WiFi routers in the world, such as routers in shopping
malls and coffee shops, will be 340 million in 2018. While being idle, the energy consumption of these public
routers can be estimated to be more than 17 billion kWh per year, which costs more than \$2
billion~\cite{WiFiusage}. Therefore, the total energy cost and the harm to the environment of all WiFi routers in
the world while being active will be very huge, if no innovative methods are applied. This trend of dirty energy
consumption of wireless access networks (WACNs)\footnote{WACN is adopted here to avoid confusion with WAN which
usually refers to Wide Area Network.} calls for new methods to reduce the carbon footprint.

Several approaches have been used to reduce the \txtgreen{power} consumption of WACNs.
In~\cite{shih2002wake,rozner2010napman,liu2008micro,manweiler2011avoiding}, a sleep-wake schedule to reduce the
power consumption of WiFi routers is proposed. In~\cite{ye2002energy}, the low power listening method is applied
to achieve the same objective. Correlated packet detection is applied in~\cite{sen2012csma,sen2010phy}. Recent
works~\cite{han2013greening,han2013optimizing,han2015traffic} have championed the concept of powering wireless
networks with renewable energy. While all of these proposals are geared to realizing green communications, they
only consider the RF band and none of them has exploited the potential of energy savings provided by operating
outside this band.

This work aims to reduce the \txtgreen{power} consumption of wireless indoor access networks. Different from all
of the previous works, we utilize the visible light communication (VLC), in which the information is modulated on
the wireless light signals. This is motivated by the fact that {\em whenever communication is needed, lighting is
also needed most of the time}. Energy consumption of lighting represents about 15\% of the world's total energy
consumption~\cite{majorelec}. By jointly performing lighting and Internet access, VLC can operate on a very small
energy budget. The emergence and the commercialization of the power line communication
(PLC)~\cite{tsuzuki2012feasibility,tonello2008challenges}, makes it feasible and attractive to add a driver
circuit to perform modulation between the light source and the power cables and utilize VLC with very small
one-time cost based on the available indoor infrastructure. This can also achieve a spectrum reuse range of 2-10
meters, which could potentially resolve the spectrum crunch problem~\cite{kavehrad2013optical}.

\txtredd{There are many earlier works \cite{giordani2016multi,asadi2013wifi,arribas2017multi} studied the
coexistence of different RF access technologies, such as LTE and WiFi.} Rather than relying only on VLC, this
paper integrates traditional RF access technologies with VLC
\txtredd{\cite{feng2016applying,basnayaka2015hybrid}}. The reason from an energy efficiency perspective is of
two-fold. First, the ambient lighting level and the lighting requirements change over the day/year. For example,
for a room with windows in a sunny morning, the ambient light could be enough for lighting, which incurs high
energy overhead for communications when using VLC. On the other hand, at night or during a cloudy day, the
required illumination level from the light source will be very high, thus incurring rather low energy consumption
for communications. Sometimes at night, illumination might not be needed, which makes VLC operate on large energy
overhead. Second, utilizing both VLC and RF reduces the interference effect among individual spectrum domains (VLC
or RF), and therefore significantly reduces the energy consumption. Based on this, our proposed system model
utilizes both RF and VLC access methods. \txtgreen{In reality, the hybrid system will rely on a central controller
to perform the feedback collection, optimization and resource allocation. A proof-of-concept hybrid WiFi-VLC
system has been implemented in \cite{shao2014indoor} and \cite{ShaoJOCN2015}.} \txtgreen {There are some earlier
works that have investigated energy/power consumption problems for hybrid RF/VLC systems. However, none of them
considered illumination. The work in \cite{chowdhury2013energy} studied power consumption of mobile terminals in
hybrid radio-optical wireless systems. Constant power consumption and data rate were modeled for each wireless
access technique. The authors in \cite{kashef2016energy} and \cite{kafafy2017power} investigated power efficiency
of hybrid RF/VLC wireless networks by maximizing the total data rate over the total power consumption. Both
\cite{kashef2016energy} and \cite{kafafy2017power} assumed the VLC AP is sending data at maximum transmission
power. The work in \cite{kashef2016energy} considered the case of only one VLC AP and the VLC AP was assumed to be
always turned on, while \cite{kafafy2017power} considered the cases of multiple VLC APs and turning on or off VLC
APs were studied in different cases. The authors in \cite{kashef2017transmit} investigated the total transmission
power of hybrid PLC/VLC/RF communication systems. In \cite{kashef2017transmit}, the authors considered single user
case and assumed a power additive model by setting fixed power for each subcarrier with equal bandwidth. In this
paper, we take illumination into account, assign integer variables to represent the ON or the OFF states of VLC/RF
APs and integrate the variables into the optimization problem.} The major challenge tackled in this paper is to
{\em minimize the \txtredd{power} consumption of the hybrid system while satisfying users' demands and maintaining
required illumination.} While~\cite{siddique2011joint,sugiyama2007brightness,lee2011modulations} address the joint
communication and illumination problem, the power consumption problem is not considered. Also, the developed
approaches are designed to address a specific given modulation scheme in the VLC system rather than the hybrid
system. Din and Kim in~\cite{din2014energy} addressed the joint illumination and power control problem in the VLC
system. It is still specific to the pulse position modulation (PPM) modulation scheme with only one VLC user in
the system.

The major contributions presented in this paper include:
\begin{itemize}

\item Proposing a new hybrid system design for WACNs that is environment friendly. This system contains both the
    traditional RF access points along with the light sources that can be used to provide lighting and
    communications jointly. We show that this system is more energy efficient and more adaptive to the changes
    in the ambient lighting conditions than the systems utilizing VLC or RF access methods separately.

\item Formulating the problem of minimizing the \txtgreen{power} consumption of the proposed system while
    satisfying the users' requests and maintaining acceptable illumination level.

\item Since the problem is NP-complete, we divide the problem into two subproblems. In the first subproblems, we
    determine the set of VLC access points (AP) that needs to be turned on to satisfy the illumination
    requirements. Taking this set as an input to the second subproblem, we develop a randomized online algorithm
    to satisfy the users' requests for real-time communications. We show that, against an oblivious adversary,
    our proposed algorithm achieves a competitive ratio of $\mathcal{O}(\log(N)\log(M))$ with probability of
    success $(1-\frac{1}{N})$. \txtgreen{Here, $N$ is the number of users and $M$ is the total number of access
    points.}

%\item Designing an online algorithm for the problem, to facilitate real time communications. We show that the
 %   best online algorithm for our problem can achieve a competitive ratio of $\log(M)$; we show that our
  %  proposed online algorithm achieves a competitive ratio of $\log(N)\log(M)$ with probability of success
   % $(1-\frac{1}{N})$ \txtblue{against an oblivious adversary}. Here, $M$ is the total number of access points
    %and $N$ is the number of users.

\item Performing extensive simulations to demonstrate the effectiveness of the proposed system model and the
    hybrid and online schemes.

\end{itemize}

\begin{comment}
The rest of the paper is organized as follows. We present our system model in Section~\ref{sysmodel}. The problem
formulation including the proof of NP-completeness is presented in Section~\ref{form} and our online algorithm is
introduced in Section~\ref{online}. We present our simulation results in Section~\ref{sim} and conclude the paper
in Section~\ref{conc}.
\end{comment}

%\section{Motivation}\label{motivate}

\section{System Model}\label{sysmodel}

\subsection{Settings} In this paper, we consider a WACN containing $M$ RF and VLC APs. Every VLC AP consists of
light-emitting diode (LED) based luminaries devices providing both lighting and communications. The communications
is performed over back-haul links to the Internet and to users through free-space wireless transmission over the
light signals. The WiFi or femtocell APs provide communications to the Internet through backhaul links and to
users through the RF signals. There are also $N$ user devices that are equipped with VLC and RF enabled
transceivers. We consider the problem of minimizing the total power consumption of the VLC-enabled luminaries (VLC
APs) and the WiFi APs while satisfying the users' requests and maintaining an acceptable illumination level. Here,
we focus on the downlink traffic as the majority of indoor traffic is downlink, i.e., asymmetric link. Therefore,
in the following, we describe the device-level power consumption model and after that we formulate and solve the
problem of minimizing the total downlink power consumption of the proposed system.

\subsection{Communications}

Optical modulation is performed by varying the forward current of the light source. The output optical flux
changes proportionally to the modulated forward current. The increase in \txtgreen{power} consumption is mainly
due to the switching loss in the driver circuitry at high speeds (AC current). Such behavior is observed in our
experimental results~\cite{ShaoJOCN2015,shao2014indoor,shao2015design} and the results in~\cite{DengPerf14}.

\subsection{Illumination}

The illumination level at a given location depends on the average optical power received. This can be generated by
both the DC and the AC currents. Typically, the optical signal generated by the AC current does not cause
flickering as it changes at a higher rate than what can be observed by the human eye. The DC component does not
require a current switching process. This switching process reduces the efficiency of the driver circuit and light
source by consuming more energy. Formally, let the $m$-th AP be a VLC AP and let $P_m^{op}$ represent its average
output optical power. The illumination provided by this VLC AP is given by a commonly used expression $\Phi_{m} =
683 \int_{0}^{\infty} V(\lambda) P(\lambda) d\lambda,$ where $V(\lambda)$ is the standard luminous
function~\cite{ghassemlooy2012optical} and $P(\lambda)$ is the spectral power distribution that depends on the
average transmitted power and the LED type. Note that $\int_{0}^{\infty} P(\lambda) d\lambda = P_m^{op}$.
Therefore, we can write $\Phi_{m}$~[lm]~$=G_{m}P_{m}^{op},$ where $G_m$ is a constant that depends on the LED. The
illumination provided by the different APs at a given location $w$ can be written as $\phi_{w}$~[lux]~$=
\sum_{m}\bigl[(\frac{g+1}{2\pi})\frac{\cos^{g}(\theta_m)\cos(\psi_m)}{r_{mw}^2}\Phi_m\bigr]$~\cite{ghassemlooy2012optical},
where $\theta_m$ and $\psi_m$ are the irradiance angle and the incidence angle of the $m$-th AP, respectively,
$r_{mw}$ is the distance between the $m$-th AP and location $w$, and $g$ is the Lambertian order that is related
to the semi-angle at half power $\varphi_{1/2}$, at which direction the intensity of luminous flux is reduced to
half of that of the central luminous flux. Formally, we have $g=-\frac{\ln 2}{\ln (cos(\varphi_{1/2}))}.$ Here,
$\phi_{w}$ is a linear function of $\Phi_{m},$ since typically all of the other parameters are constants. In the
illumination, a minimum horizontal illuminance level (e.g. 300 lux for typical office environment) constraint for
each $\phi_{w}$ is set. The ambient light intensity can be added as a new variable in the linear function for
$\phi_{w}$.

\subsection{VLC Device-Level Joint Illumination and Power Control}

Increasing the AC component of the signal and decreasing its DC component can manipulate the average output
optical  power to remain fixed. However, as noted by our empirical
results~\cite{ShaoJOCN2015,shao2014indoor,shao2015design},  this manipulation will increase the electrical power
consumption due to the current switching loss at the LED driver.  In order to achieve a given data rate and
maintain a given illumination level with the minimum \txtgreen{power} consumption,  it is necessary to send the
minimum possible AC component that achieves the given data rate, and supplement it with  a DC component to satisfy
illumination. The receiver can employ a high pass filter on the received signal to remove  the DC component. The
ambient lighting signals can be treated as DC optical signals. \txtgreen{Note that increasing the DC component
increases the \txtredd{background} noise level at the receiver. The background noise consists of the thermal noise
and the shot noise. If the modulation bandwidth is large (above 50 MHz)  and the optical power level is low (below
20 W) as is the case with most VLC APs and the VLC front-ends utilized in  our research~\cite{ShaoJOCN2015}, the
thermal noise would dominate the shot noise~\cite{grobe2013high,langer2013optoelectronics,komine2004fundamental}.
With fixed gain of the receiver, the thermal noise does not depend on the AC and DC signals while the shot noise
does. Therefore, a constant level of background noise can be assumed.} Our experimental
results~\cite{ShaoJOCN2015} validated this behavior even in outdoor settings. The AC signal bandwidth can be
divided into several channels and two AC signals from different APs operating on the same channel will interfere
with each other. Due to the illumination constraints, the maximum optical power at any location is limited; the
photodetector can be designed accordingly so that it will be able to remove any DC offset.

\begin{comment}
This framework can be extended to alternative dimming control methods for adjusting the average optical power
signal including pulse width modulation (PWM)~\cite{tanaka2003indoor,sugiyama2007brightness,lee2011modulations}.
Here, PWM is used instead of the DC level. Under this case, it is necessary to ensure that depending on the
frequency of the PWM, there is no inband interference from this PWM on the AC signal of other APs.
\end{comment}

\begin{figure}
\centering
\includegraphics[width=0.45\textwidth]{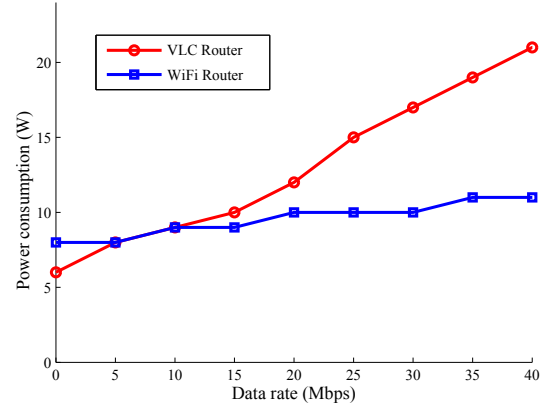}
\caption{Preliminary experimental results of power consumption vs. throughput for VLC and WiFi data transmissions}
\label{fig:Throughput_Power}
\end{figure}

If different users connecting to the same AP use different channels, the power additive model is applied. In this
model, except for the power to turn-on the AP, the additional power needed to serve a user at a given location, a
given field of view (FOV), and a given data rate by a given AP is fixed regardless of the number of users served
by the same AP. Even when assigning different users to one channel, our preliminary experimental results have
validated the suitability of the power additive model for VLC in an indoor environment. \txtgreen{As shown in
Fig.~\ref{fig:Throughput_Power}. The VLC and WiFi power consumption measurements are performed based on our VLC
front-ends and a NETGEAR Wireless Dual Band Gigabit Router WNDR4500, respectively. VLC and
WiFi transmitters are set to keep sending UDP packets at different data rates while a power meter is attached to
the transmitters to monitor the power consumption. The power consumptions for both VLC and WiFi are approximately
proportional to the data rate when the turn-on power is not considered, which validates the suitability of the
power additive model for both wireless access technologies.} Other experimental studies for WiFi also validate
this model~\cite{khan2013model,huang2012close,friedman2013power,sun2014modeling}.

\begin{comment}
Even when assigning different users to one channel, our preliminary experimental results~\cite{ShaoJOCN2015} have
validated the suitability of the power additive model for VLC in an indoor environment.
\end{comment}

Let the $m$-th AP be a VLC AP; we use $P_{m}^{on}$ to denote the electrical power consumed when the whole
transmitted optical power from the $m$-th AP is generated by DC current. If the $m$-th AP is only used to provide
illumination, $P_{m}^{on}$ will represent the total consumed power at that AP. We use $P_{mn}$ to denote the extra
electrical power (i.e. caused by AC current switching loss) consumed to support a data rate of $R_{n}$ for the
$n$-th user through the $m$-th AP. Based on the power additive model, knowing both $P_{m}^{on}$ and $P_{mn}$, the
power consumption of a VLC AP when performing joint communications and illumination can be fully characterized.
Note that $P_{m}^{on}$ depends on the required illumination level. Therefore, $P_{mn}$ needs to be fully
characterized when the $n$-th user initiates the request to be served. For stationary users, the VLC channel is
very stable as compared to the RF channel~\cite{zhangdancing}. Therefore, the channel state information (CSI) in
terms of signal-to-noise ratio (SNR) or bit error rate (BER) is able to provide a good estimation for these
values.

\subsection{RF Device-Level Model}
If the $m$-th AP is an RF AP, its total power consumption can be written as $P^{on}_{m}+\sum_{n\in
\mathcal{N}(m)}P_{mn}.$ Here, $P^{on}_{m}$ is the power needed to turn on the RF AP, $P_{mn}$ is the extra power
required to support a data rate of $R_n$ to the $n$-th user from the $m$-th AP and $\mathcal{N}(m)$ is the set of
users that are connected to the $m$-th AP. As explained in~\cite{mabell2013energy,zhang2012mili}, $P^{on}_{m}$
represents the major component of the power consumption of the WiFi or the femtocell AP. Here, the power additive
model is utilized, which is validated by preliminary experimental results and
by~\cite{khan2013model,huang2012close,friedman2013power,sun2014modeling}. $P_{m}^{max}$ is used to represent the
maximum $(\sum_{}P_{mn})$ that the AP can support. If the $m$-th AP is an RF one, both $P_{m}^{on}$ and
$P_{m}^{max}$ will be fixed. On the other hand, $P_{m}^{max}$ depends on $P_{m}^{on}$ for the VLC AP, which could
be fixed or could be controllable depending on the ambient illumination level that is changing over time.

%In some cases, $P_{m}^{max}$ can be replaced by $Y_{m}^{max}$ or both of them are needed, where $Y_{m}^{max}$ is the maximum number of users an AP can support.

When the ambient light level is very small, $P_{m}^{on}$ for the VLC APs will be utilized for lighting, thus
making the VLC AP {\em power-proportional}, i.e., the total power consumed by the VLC AP for communications is
proportional to the data rate. Designing a power proportional communications device is considered ideal rather
than realistic to be achieved~\cite{barroso2007case,valancius2009greening,gunaratne2005managing}. Here, it is
shown that \emph{jointly performing illumination and communications can potentially realize the power proportional
AP}. On the other hand, when the ambient light level is large, $P_{m}^{on}$ will either be very small so that it
cannot support high data rates or be large to support high data rates but not to be utilized for lighting. This
makes the VLC AP power non-proportional in this case.

\section{Problem Formulation}\label{form}

\subsection{Hybrid WiFi-VLC System Problem Formulation}\label{form2}

The problem of minimizing the total power consumption is formulated as the following Mixed Integer Linear Program
(MILP):

\begin{equation} \mathbb{P}1: \min \sum_{m=1}^{M}P_{m}^{on}X_{m} +  \sum_{m=1}^{M}\sum_{n=1}^{N}P_{mn}Y_{mn} \label{eqn:basic_formulation}
\end{equation}

subject to

\begin{align}
& Y_{mn} \leq X_{m}, \quad \forall n, m \label{const1}\\
& \sum_{m = 1}^{M} Y_{mn} \geq 1 \quad \forall n \label{const2}\\
& \sum_{n = 1}^{N} P_{mn}Y_{mn} \leq X_{m}P_{m}^{max} \quad \forall m \label{const3}\\
& \sum_{m\in VLC}\bigl[(\frac{g+1}{2\pi})\frac{\cos^{g}(\theta_m)\cos(\psi_m)}{r_{mw}^2}G_mP_{m}^{on}\eta_{m}^{DC}X_{m}\bigr]\geq \mathcal{I}_w\label{const4}
\end{align}

Here, $X_{m}  =\left\{ \begin{array}{rl}
1 &\mbox{if the $m$-th AP is on}\\
0 &\mbox{otherwise.}
\end{array}\right.$

$Y_{mn}  =\left\{ \begin{array}{rl}
1 &\mbox{if the $n$-th user is connected to the}\\
   &\mbox{$m$-th AP}\\
0 &\mbox{otherwise.}
\end{array}\right.$

\vspace{5pt} We assume that the output optical power of a VLC AP is fixed. Therefore, in the objective function of
the stated problem formulation, the first term refers to the total power consumption for turning on the APs while
the second term refers to the total power consumption for data transmission. The first set of constraints ensures
that a user cannot connect to an AP that is not turned on. The second set of constraints ensures that a user
should be connected to at least one AP. The third set of constraints ensures that the total power consumed by an
AP does not exceed its maximum power consumption limit. The last set of constraints is the illumination
constraints. Here, $\mathcal{I}_w$ represents the illumination level required at the $w$-th location.
\txtgreen{Note that the value of $\mathcal{I}_w$ can be zero if the ambient light is suffient. If illumination is
not required, a VLC AP might still be turned on to reduce the total power consumption. Regarding the issue about
turning on VLC APs while illumination is not required, we discuss it in three cases. 1) Illumination is not needed
but turning on VLC APs is acceptable. Case 1 is the situation studied in our simulation settings (Sec.~VI).
Turning on VLC APs at day especially for communication purpose is capable of reducing the total power consumption
and the additional illumination produced by VLC APs will not cause the illumination levels exceeding the
recommended range. 2) Illumination is not needed and turning on VLC APs will produce unexpected brightness. Case 2
considers scenarios when darkness is preferred. For case 2, maximum illumination constraints can be added to the
optimization problem and different maximum illumination levels will lead to different results. Case 2 will be
studied in the future. 3) Illumination is not needed but turning on VLC APs will not produce unexpected
brightness. In case 3, turning on VLC APs for communication may not produce any human-sensitive lighting.
\cite{tian2016darklight} presents the approach that enables a VLC AP for communication while keep the environment
being dark. The main idea is utilizing pulse width modulation for VLC and reducing the pulse width in each symbol
to perform the brightness control. The method proposed in \cite{tian2016darklight} provides a flexible solution to
case 2 where the maximum illumination constraint varies.} If the $m$-th AP is a VLC one and it is sending only DC
signal, we use $\eta_{m}^{DC}$ to represent its efficiency. Note that the output optical power of a VLC AP is
fixed. Therefore, $P_{m}^{op}=P_{m}^{on}\eta_{m}^{DC}$. \txtgreen{In the following, we show that our problem is
NP-complete.
\subsection{NP-Completeness Proof}\label{NP_comp_proof}
The following theorem shows that our problem is NP-complete.
\begin{theorem}\label{theorem:theorem1}
The ILP optimization problem is NP-complete.
\end{theorem}

\begin{proof}
Given any solution to our problem, we can easily check the solution's feasibility in polynomial time by checking
the constraints. Since the number of constraints is polynomial in terms of the number of users, APs, and
locations, our problem belongs to NP.

To prove that our problem is NP-hard, we reduce the capacitated facility location problem, which is known to be an
NP-complete problem, to an instance of our problem. The capacitated facility location problem is defined as
follows: Given a set of facilities $I$ and a set of users $D$, we wish to open a subset of facilites to serve the
users in $D$, where each facility $i \in I$ has an opening cost $q_i$, a capacity $c_i$, and the cost of serving a
user $d$ from a facility $i$ is $r_{id}$. The objective is to minimize the total cost.

The reduction is performed as follows: (1) The $i$-th facility in the facility location problem is mapped to the
$m$-th AP in our problem. (2) The $d$-th user in the facility location problem is mapped to the $n$-th user in our
problem. (3) $P_{m}^{on} = q_i$, (4) $P_{m}^{max} = c_i$. (5) $P_{mn} = r_{id}$. (6)$\mathcal{I}_w=0, \forall w$.

Now we prove that there exists a solution to the capacitated facility location problem with cost $\mathcal{A}$ if
and only if there exists a solution to our problem with cost $\mathcal{A}$.

If there exists a solution to the capacitated facility location problem with cost $\mathcal{A}$, then the set of
opened facilities represents the APs that are turned on, a facility $i$ serving user $d$ represents the $m$-th AP
serving the $n$-th user, and the total cost of our problem is $\mathcal{A}$. The proof of the other direction of
the if and only if statement is done in a similar way.
\end{proof}
}

\section{Online Algorithm}\label{online}
To solve the offline optimization problem $\mathbb{P}1$, one can use an optimizer such as CPLEX
\cite{cplex2014v12}. Solving $\mathbb{P}1$ requires the knowledge of all requests for connections from users
apriori. However, such knowledge is not available in real scenarios. Alternatively, one can recompute the solution
to $\mathbb{P}1$ whenever a new request (or batch of requests) for connection from a user (users) appears, but
this may lead to the reconfiguration of the system with the appearance of the request and may require advanced
techniques, such as soft handover, to avoid connection disruption for previous users when switched to a different
AP. Moreover, recomputing the solution to $\mathbb{P}1$ will lead to additional intolerable delays to the new user
since the complexity of resolving $\mathbb{P}1$ increases as the number of users increases. Therefore, we present
in this section a \txtblue{randomized} online algorithm for the hybrid RF-VLC system, to tackle the
above-mentioned issues.

In the online version of the problem addressed here, requests for connections from users appear one by one. The
online algorithm has to make a decision to connect the user to an AP, and the algorithm's decision has to be made
before the next user appears. Moreover, the decisions of the online algorithm cannot be reversed in the future.
\txtgreen{Due to the NP-completeness of the problem as shown in Theorem 1}, we remove the $P_{m}^{max}$ constraint
and assume that there is no limit on it. \txtgreen{In the simulation results, we show that the power consumption
of the online algorithm is not high (in order of 50 Watts for 100 users)}

To compare the performance of the online algorithm to that of the optimal offline algorithm, we use the concept of
competitive ratio. The competitive ratio is defined as the worst-case ratio of the performance achieved by the
online algorithm to the performance achieved by the optimal offline algorithm, i.e., if we denote the performance
of the online algorithm by $\mathcal{P}_{on}$ and that of the optimal offline algorithm by $\mathcal{P}_{off}$,
then the competitive ratio is:
\begin{displaymath}
\sup_{All\ requests} \frac{\mathcal{P}_{on}}{\mathcal{P}_{off}}.
\end{displaymath}
\txtgreen{To experimentally measure the competitive ratio, one needs to find the input sequence (among $2^M$
possible input sequences) that gives the worst possible performance, which becomes infeasible as $M$ grows large.}
As the ratio gets closer to 1, the online performance gets closer to the offline performance. In other words, the
smaller the competitive ratio is, the better the online algorithm's performance will be.

In the online algorithm, we may desire to satisfy the illumination requirements before the appearance of any
request for connection from users. Therefore, we decompose the optimization problem $\mathbb{P}1$ into two
subproblems. In the first subproblem, we solve an optimization problem to find the set $\mathcal{M}'$ of turned-on
VLC APs to satisfy the illumination requirements, and is formulated as follows:

\begin{equation}
\mathbb{P}2: \min \sum_{m \in VLC}P_{m}^{on}X_{m} \label{eqn:basic_formulation1}
\end{equation}

subject to

\begin{equation}
\sum_{m\in VLC}\bigl[(\frac{g+1}{2\pi})\frac{\cos^{g}(\theta_{m})\cos(\psi_{m})}{r_{mw}^2}G_{m}P_{m}^{on}\eta_{m}^{DC}X_{m}\bigr]\geq \mathcal{I}_w\label{const4}
\end{equation}

Given the set of turned-on VLC APs, the second subproblem is to minimize the power consumption by all APs while
satisfying the users' requests, and is formulated as follows:

\begin{equation}
\mathbb{P}3: \min \sum_{m=1}^{M}P_{m}^{on}X_{m} +  \sum_{m=1}^{M}\sum_{n=1}^{N}P_{mn}Y_{mn} \label{eqn:basic_formulation3}
\end{equation}

subject to

\begin{align}
& X_{m} = 1, \quad \forall m \in \mathcal{M}' \label{const1}\\
& Y_{mn} \leq X_m, \quad \forall n, m \label{const2}\\
& \sum_{m = 1}^{M} Y_{mn} \geq 1 \quad \forall n \label{const3}\\
& \sum_{n = 1}^{N} P_{mn}Y_{mn} \leq X_{m}P_{m}^{max} \quad \forall m \label{const4}
\end{align}

\txtgreen{Note that the two subproblems are also NP-complete, which can be proved by a reduction from the facility
location problem as illustrated in Section~\ref{NP_comp_proof}}. Our online algorithm proposed here \txtred{for
problem $\mathbb{P}3$} is inspired by the idea of the online algorithm presented in \cite{buchbinder2009design}.
We construct a graph $G = (V, E)$ consisting of $M+N+1$ vertices, where $M$ is the number of APs, $N$ is the
number of users, and the additional vertex represents a virtual source $S$\txtgreen{, which could represent the
gateway router and the centralized controller that all APs are connected to}. We connect the source $S$ to each of
the $M$ APs, and each AP to all of the users that the AP can serve. We associate two values for each edge $e$.
\txtgreen{The first value represents the flow on the edge $w_e$ and is unitless. The weights are dynamically
changing during the run of the algorithm and their values are in the range [0, 1]. The second value associated
with an edge is the edge's cost $c_e$, which represents the power consumption of that edge and has the unit of
Watts.} \txtblue{There are three types of edges when assigning these initial costs:

\begin{itemize}
\item The first type of edges is an edge connecting the virtual source $S$ to a VLC AP that is already turned on
    to satisfy illumination requirements (as determined from solving the optimization problem $\mathbb{P}2$).
    The set of turned-on APs due to solving $\mathbb{P}2$ is represented by $\mathcal{M}'$. The initial cost of
    an edge connecting the virtual source $S$ to an AP $m \in \mathcal{M}'$ is set to 0.

\item The second type of edges contains the edges connecting the virtual source to the remaining APs (i.e., the
    VLC APs that are not in the set $\mathcal{M}'$ and WiFi APs). The initial cost of an edge connecting the
    virtual source $S$ to an AP $m \notin \mathcal{M}'$ is set to the power consumption ($P_m^{on}$) required to
    turn on the $m$-th AP.

\item The third type of edges is an edge connecting a user to an AP. The initial cost of an edge connecting the
    $n$-th user to the $m$-th AP is set to the power consumption ($P_{mn}$) required if the $n$-th user is to be
    connected to the $m$-th AP.
\end{itemize}

For example, as shown in Fig.~\ref{fig:graph}, the cost of the edge connecting $AP1$ to user 1 in the graph on the
right side of the figure is set to $P_{11}$, which is the power consumed if user 1 is to be connected to $AP1$ in
our problem (the graph on the left side of Fig.~\ref{fig:graph}).}

After the graph is constructed, our problem becomes equivalent to the following \txtred{Minimum Cost Connectivity
(MCC)} problem on the newly constructed graph:
\begin{displaymath}
\txtblue{\mathbb{P}4:} \min \sum_{e}w_{e}c_e
\end{displaymath}
subject to
\begin{align}
max\_flow(S,n,w) \geq 1 \quad \forall n\label{const:MCMF}
\end{align}
where $c_e$ is the cost of edge $e$, $w_e$ is the weight of edge $e$, and $max\_flow(S,n,w)$ is the maximum flow
from the source $S$ to the $n$-th user when the capacity of each edge is 1 and the flow on each edge is
represented by the weight \cite{WiKiFlow}. \txtblue{The Minimum Cost Connectivity (MCC) problem is defined as
follows: given a graph $G' = (V', E')$ with cost function $c: E' \rightarrow \mathbb{R}^+$, and a demand $D_i =
(s_i, t_i)$, where $s_i$ is the source of $D_i$ and $t_i$ is the sink of $D_i$. The objective is to find an
assignment of weights $w \in \{0, 1\}$ to $E'$, such that there is a flow from $s_i$ to $t_i$ of value at least 1
with the minimum cost, where the cost is equal to $\sum_{e \in E'} w_e c_e$. Constraint~\ref{const:MCMF} states
that the flow from the source to each user must be at least 1 in order to satisfy the user's demand}. An example
of the reduction with the cost assigned to the edges is shown in Fig.~\ref{fig:graph}, \txtblue{which shows how
our problem can be mapped to the new problem $\mathbb{P}4$. Since the costs of the edges in the graph represent
the power consumption of turning on an AP or connecting a user to an AP, a solution that minimizes the total cost
in problem $\mathbb{P}4$ while having a flow of value of at least 1 from the virtual source $S$ to every
destination is a solution that minimizes the power consumption in our problem $\mathbb{P}3$ while satisfying the
users' demands.}

\begin{figure}
\centering
\includegraphics[scale = 0.6]{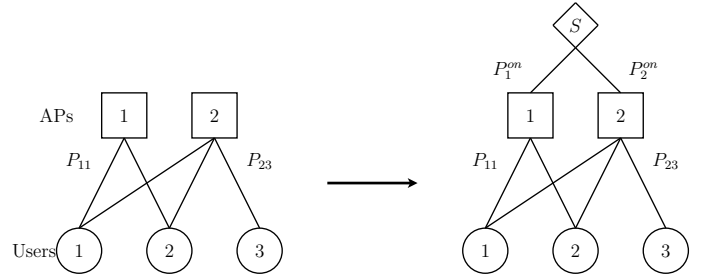}
\caption{Constructing the graph for the online algorithm.}
\label{fig:graph}
\end{figure}

The online algorithm is described in Algorithm \ref{alg:alg1} \txtblue{and the notations used in the algorithm are
presented in Table~\ref{table:notation}}. The algorithm operates in an online fashion upon each user's arrival.
The algorithm consists of two phases. In the first phase, the algorithm computes a fractional solution. The second
phase is a randomized rounding phase, where the algorithm rounds the fractional solution computed in the first
phase to an integral solution. \txtblue{
\begin{table}
\caption{Notations used in Algorithm 1}\label{table:notation}
\centering
\begin{tabular}{|c|c|}
\hline
Notation & Definition\\
\hline
$\alpha$ & Guess of the value of the optimal fractional solution\\
\hline
$c_e$ & Cost of edge $e$\\
\hline
$c_e^{'}$ & Normalized version of $c_e$\\
\hline
\txtred{$c_{tot}^{frac}$} & \txtred{Total fractional cost of the online algorithm}\\
\hline
$c_{tot}$ & Total \txtred{integral} cost of the solution of the online algorithm\\
\hline
$w_e$ & Weight of edge $e$\\
\hline
$w_e^{'}$ & Weight of edge $e$ updated by the online algorithm\\
\hline
$\gamma_e$ & Value of the threshold associated with edge $e$\\
\hline
\txtgreen{$\mathcal{C}$} & \txtgreen{Minimum weight cut} \\
\hline
\txtgreen{$\Gamma(e,j)$}& \txtgreen{Set of random variables associated with edge $e$}\\
\hline
\end{tabular}
\end{table}
}

\begin{algorithm}
\caption{Online Hybrid RF-VLC System} \label{alg:alg1}
\begin{algorithmic}[1]
\scriptsize \STATE{Solve the ILP Optimization Problem \eqref{eqn:basic_formulation1} \txtgreen{using CPLEX} to
find the set of turned-on VLC APs, $\mathcal{M}'$.}

\STATE{\txtred{$\forall e$ connecting $S$ to $m, m \in \mathcal{M}'$, $c_e = 0$}}

\STATE{\txtred{$\forall e$ connecting $S$ to $m, m \notin \mathcal{M}'$, $c_e = P_m^{on}$}}

\STATE{$n' = 0$}

\STATE{$\alpha = \min_{e} c_e \quad \vert e \text{~connecting~} S \text{~to~} m, m \notin \mathcal{M}'$}

\STATE{$w_e = w_e^{'} = \frac{1}{M^2}$ }

\STATE{$c_{tot} = 0$}

\STATE{Upon the arrival of a new user $u$}

\STATE{$n' \leftarrow n' + 1$}

\STATE{\txtblue{$\forall e$ connecting $m$-th AP to $u$, $c_e = P_{mu}$}}

\STATE{$\forall e$, keep $\lceil 2\log(n' + 1) \rceil$ independent random variables $\Gamma(e,j)$, $1 \leq j \leq
\lceil 2\log(n' + 1) \rceil$, uniformly distributed in the interval [0, 1]. Define a threshold $\gamma_e = \min_j
\Gamma(e,j)$ (\txtgreen{for randomized rounding})}

\STATE{START:}

\STATE{$\forall e$ such that $c_e \leq \frac{\alpha}{M}$, set $w_e = w_e^{'} = 1$}

\STATE{$\forall e$ such that $\frac{\alpha}{M} \leq c_e \leq \alpha$, set $c_e^{'} = \frac{c_e}{\alpha/M}$}

\IF{the maximum flow between $S$ and $u$ is at least 1 (\txtgreen{\emph{i.e.,} user is already satisfied})}

\STATE{do nothing}

\ELSE

\WHILE{the flow is less than 1 (\txtgreen{\emph{i.e.,} user is not satisfied yet})}

\STATE{Compute the minimum weight cut $\mathcal{C}$ between $S$ and $u$}

\STATE{$\forall e \in \mathcal{C}$, $w_e^{'} \leftarrow w_e^{'}(1 + \frac{1}{c_e^{'}})$ (\txtgreen{, weight
augmentation step})}

\ENDWHILE

\ENDIF

\STATE{$w_e = \max\{w_e, w_e^{'}\}$}

\STATE{\txtred{$c_{tot}^{frac} = \sum_e w_e c_e^{'}$}}

\IF{$w_e \geq \gamma_e$ \txtgreen{(Randomized Rounding)}}

\STATE{$w_e = 1$}

\ELSE

\STATE{\txtgreen{$w_e = 0$}}

\ENDIF

\STATE{$c_{tot} = \sum_e w_{e}c_e^{'}$}

\IF{$c_{tot}^{frac} > 2\alpha \log(M) + \alpha + 1$}

\STATE{$\alpha \leftarrow 2\alpha$}

\STATE{Go to START}

\ENDIF
\end{algorithmic}
\end{algorithm}

For the first phase, the algorithm \txtmag{uses the doubling technique to} guess the value of the optimal
fractional solution denoted by $\alpha$. \txtmag{To get the intuition of the doubling technique, suppose that the
online algorithm has a competitive ratio of $\Theta$ and the true cost of the optimal solution is $c^*$. We begin
with the initial guess of the optimal cost $\alpha$, and run the algorithm assuming this guess is the correct
estimate of $c^*$. If the online algorithm fails to find a feasible solution of a cost at most $\Theta \alpha$, we
double the value of $\alpha$ and continue with the algorithm. Eventually, $\alpha$ will exceed $c^*$ by at most a
factor of 2, and for this value of $\alpha$, the algorithm will compute a feasible solution (since all demands are
satisfied \txtgreen{as shown in Section~\ref{sec:PerformanceAnalysis}}) of cost at most $2\Theta c^*$. Therefore,
we can assume that the value of the optimal solution is known.}

The initial guess of $\alpha$ is set to $\alpha = \min_e c_e$. Based on the value of $\alpha$, we sort the edges
into three categories. The first category includes all edges with a cost less than $\frac{\alpha}{M}$. All edges
in the first category are assigned a weight of 1, \txtmag{paying at most an additional cost of $\sum_{e \in
\mathcal{M}'} \frac{\alpha}{M} = \alpha$}. The second category includes all edges with a cost greater than
$\alpha$. All edges in the second category are excluded from the solution. The third category includes all the
remaining edges \txtmag{where $\frac{\alpha}{M} \leq c_e \leq \alpha$}. All edges in the third category are
assigned an initial weight of $\frac{1}{M^2}$ and their costs are normalized \txtmag{by $\frac{\alpha}{M}$} to be
between $1$ and $M$. The algorithm then updates the weights of the edges in the third category until the minimum
weighted cut is greater than 1. \txtmag{During the execution of the algorithm, it may turn out that a demand
cannot be satisfied (due to the excluded edges from the second category mentioned above), or that the true value
of the optimal solution is greater than the current value of $\alpha$ (which can be verified by checking if the
cost of the fractional solution exceeds an upper bound on the cost of the algorithm, which is $2\alpha \log(M) +
\alpha + 1$, known through the competitive-ratio analysis of the algorithm (Theorem~\ref{theorem:theorem2}). In
this case, we double the value of $\alpha$, which means that more edges are available for the online algorithm to
satisfy the demands, ``forget'' about the weights assigned to the edges so far, double the value of $\alpha$, and
continue the algorithm. Although we ``forget'' about the weights when $\alpha$ is doubled, the actual weight of an
edge used to calculate the total cost of the algorithm is the maximum weight assigned to that edge so far. This
ensures that the edges previously chosen will not be unselected over the iterations.} If the total cost of the
online algorithm is less than \txtblue{$2\alpha \log(M) + \alpha + 1$}, then we are guaranteed that we achieve a
fractional solution that is within a $\log(M)$ factor of \txtmag{$\alpha$}.

After the fractional solution is obtained, the algorithm rounds the solution to an integral solution which is
within a $\log(N)$ factor of the fractional solution. \txtblue{This is done by comparing the weight $w_e$ of an
edge $e$ to the threshold $\gamma_e$ (as computed in line 11 of Algorithm 1) assigned to that edge. If the weight
of the edge is greater than the value of the threshold (line 23 of Algorithm 1), the weight of the edge is set to
1. This randomized rounding process introduces a $\mathcal{O}(\log(N))$ factor to the competitive ratio.}
Therefore, the competitive ratio of our algorithm is $\mathcal{O}(\log(M)\log(N))$.

The complexity of the online algorithm mainly comes from two parts. The first part is solving the ILP optimization
problem in line 1 of Algorithm 1. We note that this optimization problem is only solved once before the arrival of
users' requests. \txtredd{The solution to the problem can be saved for future reference when the same illumination
requirements occur again.} Moreover, this ILP problem only considers illumination requirements, it has a lower
complexity than the original problem $\mathbb{P}1$. The second part is finding the minimum weight cut for every
unsatisfied user. There are many algorithms that can find the minimum weight cut in polynomial time
\cite{WiKiCut}. Therefore, the rest of the algorithm has a polynomial time complexity.

\section{Performance Analysis}\label{sec:PerformanceAnalysis}
\txtred{In this section, we prove that the competitive ratio of the online algorithm is
$\mathcal{O}(\log(M)\log(N))$ with respect to Problem $\mathbb{P}4$}. Moreover, we prove that the best competitive
ratio achieved by any online algorithm under our settings is $\Omega(\log(M))$. \txtblue{The adversary model used
in the proof is an oblivious one.} We note that the optimization problem \eqref{eqn:basic_formulation1} is solved
\txtred{using CPLEX \cite{cplex}} at the beginning of the online algorithm. Therefore, the following proof is for
lines 8-28 in Algorithm~\ref{alg:alg1}. \txtmag{Since $\alpha$ is the guess of the optimal solution, which is
assumed to be known through the doubling technique, $\alpha = \sum_{e} c_{e}^{'}w_{e}^{*}$, where $w_{e}^{*}$
denotes the weight assigned to edge $e$ by the optimal solution.} All logarithms are to the base 2.

\begin{comment}
\txtred{We start the proof by showing that when the algorithm terminates, $\alpha$ is within a constant factor of
the optimal solution. We then show that the total cost incurred by the online algorithm is a
$\mathcal{O}(\log(M)\log(N))$ competitive.

\begin{lemma}\label{lemma:lemma1}
When the algorithm terminates, $\alpha$ is within a constant factor of the optimal solution.
\end{lemma}

\begin{proof}
At the beginning of the algorithm, $\alpha = \min c_e$, which is a lower bound on the optimal solution. Moreover,
the value of $\alpha$ and the value of the optimal solution cannot exceed $\sum_e c_e \leq \max (c_e) \vert E
\vert$. Since $\alpha$ is doubled with every step, the number of doubling steps is constant and bounded by
$\log(\frac{\max c_e}{\min c_e}\vert E \vert$. Therefore, when the algorithm, $\alpha$ is within a constant factor
of the optimal solution
\end{proof}

Based on the results of Lemma~\ref{lemma:lemma1}, we use $\alpha$ to denote the value of the optimal solution in
the following theorem. Therefore, $\alpha = \sum_e (c^{'}_e w^*_e)$.}
\end{comment}

\begin{theorem}\label{theorem:theorem2}
For a fixed $\alpha$, the online algorithm produces an integral solution that is:
\begin{itemize}
\item $\mathcal{O}(\log(M)\log(N))$ competitive.
\item The solution is feasible with probability $1 - \frac{1}{N}$.
\end{itemize}
\end{theorem}

\begin{proof}
\begin{proof}
We first show that the fractional solution of the online algorithm is within a $\mathcal{O}(\log(M))$ factor of
the optimal offline solution. Then, we show that the integral solution of the online algorithm is within a
$\mathcal{O}(\log(N))$ factor of the fractional solution. Therefore, the integral solution is within an
$\mathcal{O}(\log(M)\log(N))$ factor of the optimal offline solution.

We prove that the fractional solution of the online algorithm is within $\mathcal{O}(\log(M))$ of the optimal
solution by proving the following two claims:
\begin{itemize}
\item The number of weight augmentation steps of the online algorithm is \txtblue{$\mathcal{O}(\alpha\log(M))$}.
\item In each weight augmentation step, the increment of \txtblue{the total cost of the solution returned by the
    online algorithm} does not exceed 1.
\end{itemize}

To prove the first claim, consider the following potential function:
\begin{displaymath}
\beta = \sum_e c_{e}^{'}w_{e}^{*}\log(w_{e}^{'})
\end{displaymath}
where $w_e^*$ is the weight of edge $e$ computed by the optimal solution, \txtblue{$c_{e}^{'}$ is the normalized
version of the cost $c_e$ of edge $e$ as computed in line 14 of Algorithm~\ref{alg:alg1}} and $w_e^{'}$ is the
weight of edge $e$ updated by the online algorithm. We show that if the potential function is updated after each
weight augmentation step, then we have the following properties of the potential function:
\begin{itemize}
\item Initially, $\beta = -2\alpha\log(M)$.
\item When the algorithm terminates, $\beta \leq \alpha$.
\item In each weight augmentation step, $\Delta\beta \geq 1$.

\end{itemize}
These properties guarantee that the number of weight augmentation steps performed by the online algorithm is at
most $\alpha + 2\alpha\log(M) \txtgreen{+ 1}= \mathcal{O}\log(M)$. To prove the first property, note that
initially, $w_{e}^{'} = \frac{1}{M^2}$ and $\beta = \sum_e c_{e}^{'}w_{e}^{*}\log(\frac{1}{M^2}) =
-2\alpha\log(M)$. To prove the second property, the weight of each link is at most 2 when the algorithm
terminates, and $\beta \leq \sum_e c_{e}^{'}w_{e}^{*}\log(2) = \alpha$. To prove the third property, note that in
each augmentation step, we have:
\begin{align}
\Delta\beta &= \sum_{e \in \mathcal{C}}c_{e}^{'}w_{e}^{*}\log(w_{e}^{'}(1 + \frac{1}{c_{e}^{'}})) - \sum_{e \in \mathcal{C}}c_{e}^{'}w_{e}^{*}\log(w_{e}^{'})\nonumber\\
&= \sum_{e \in \mathcal{C}}c_{e}^{'}w_{e}^{*}\log(1 + \frac{1}{c_{e}^{'}})
\geq \sum_{e \in \mathcal{C}}w_{e}^{*} \geq 1 \nonumber
\end{align}
where the last inequality follows from the fact that the total weight assigned by the optimal solution \txtblue{to
the edges in the cut} is at least 1. \txtblue{Therefore, the number of weight augmentation steps performed by the
online algorithm is at most $2\alpha \log(M) + \alpha + 1 = \mathcal{O}(\alpha\log(M))$.}

To prove the second claim, consider an iteration in which a new user arrives and the weight augmentation step is
performed. The total weight of the links in $\mathcal{C}$ is less than 1, and the weight of each link $e \in
\mathcal{C}$ will increase by $w_{e}^{'}c_{e}^{'}$. Therefore, the total increase of \txtblue{the cost of the
solution returned by the online algorithm} in a single weight augmentation step is:
\begin{displaymath}
\sum_{e \in \mathcal{C}}\frac{w_{e}^{'}}{c_{e}^{'}}c_{e}^{'} = \sum_{e \in \mathcal{C}}w_{e}^{'} \leq 1
\end{displaymath}
From the first claim, the total number of weight augmentation steps is \txtblue{$\mathcal{O}(\alpha\log(M))$}.
Therefore, the fractional solution of the online algorithm is within a $\log(M)$ factor of the optimal solution.

\begin{figure*}
\centering
\includegraphics[scale = 0.35]{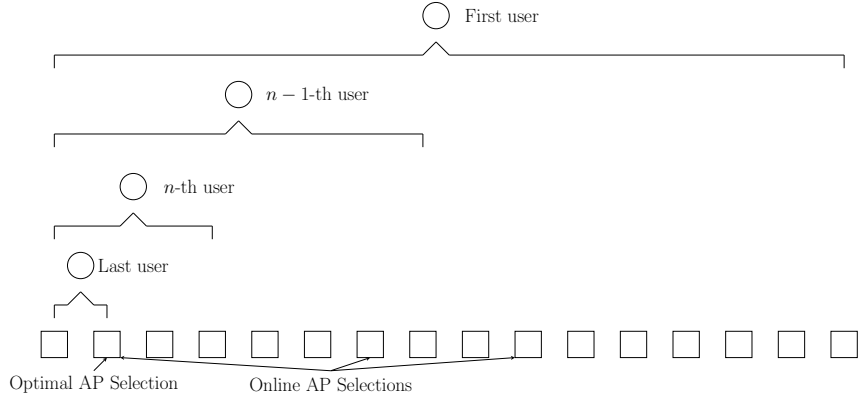}
\caption{The system used in the proof of Theorem \ref{theorem:lower}.}
\label{fig:lower}
\end{figure*}

The proof that the integral solution of the online algorithm is within $\mathcal{O}(\log(N))$ of the fractional
solution, and is feasible with probability $1 - \frac{1}{N}$ is explained as follows.

For each $j, 1 \leq j \leq 2\log(N+1)$, the probability that $\Gamma(e,j) \leq \txtgreen{w_{e}}$ is exactly
$\txtgreen{w_{e}}$. Let \txtblue{$Z$} denote the set of edges whose weights are set to 1 after the rounding
process. Note that the probability that an edge $e \in \txtblue{Z}$ is the probability that there exists a $j$
such that $\Gamma(e,j) \leq \txtgreen{w_{e}}$. Let $Y(e,j)$ denote the indicator of the event that $\Gamma(e,j)
\leq \txtgreen{w_{e}}$ (i.e., $Y(e,j) = 1$ if $\Gamma(e,j) \leq \txtgreen{w_{e}}$, and 0 otherwise). Since
$\Gamma(e,j)$ is chosen uniformly at random in the range [0, 1], the probability and the expectation of the
indicator $Y(e,j)$ is at most $\txtgreen{w_{e}}$. Therefore, the expected cost of the integral solution is
\begin{align}
 \mathbb{E} \bigl[\sum_{e \in \txtblue{Z}}c_e \bigr] &\leq \sum_{e \in E}\sum_{j=1}^{2\log(N+1)}c_e \mathbb{E}[Y(e,j)]\nonumber\\
&\leq \sum_{e \in E}\sum_{j=1}^{2\log(N+1)}c_e \txtgreen{w_{e}}= 2\log(N+1) \sum_{e \in E}c_e \txtgreen{w_{e}} \nonumber
\end{align}

Therefore, the expected cost of the integral solution is at most $2\log(N+1)$ times the cost of the fractional
solution.

To prove that the solution is feasible with probability $1 - \frac{1}{N}$, pick any user $n$. The probability that
user $n$ is not served is the probability that all of the edges connecting the virtual source $S$ to user $n$ has
$\Gamma(e,j) > \txtgreen{w_{e}}, \forall j, 1\leq j \leq 2\log(N+1)$. This probability is:
\begin{align}
\prod_{j}\prod_{e}(1 - Y(e,j)) &= \prod_{j}\prod_{e}(1 - \txtgreen{w_{e}}) \nonumber\\
&\leq \prod_{j}\exp(-\sum_{e}\txtgreen{w_{e}}) \nonumber\\
&\leq \prod_{j}exp(-1) = \exp(-2\log(N+1)) \nonumber\\
&\leq \frac{1}{N^2} \nonumber
\end{align}
The first inequality follows since $(1 - \txtgreen{w_{e}}) \leq \exp(-\txtgreen{w_{e}})$. \txtgreen{The second
inequality follows since user $n$ is satisfied in the fractional solution and thus $\sum_e w_e \geq 1$. Using the
union bound, the probability that there exists an unserved user in the integral solution is at most
$N\frac{1}{N^2} = \frac{1}{N}$.}
\end{proof}
\end{proof}

Note that the process of doubling $\alpha$ results in an additional factor of 2 of the competitive ratio, and
``forgetting'' the weights when $\alpha$ is doubled results in an additional factor of 2, to a total additional
factor of 4. Nevertheless, the competitive ratio of the online algorithm is $\mathcal{O}(\log(M)\log(N))$.

\begin{comment}
\txtred{To get the intuition behind the doubling technique, let us say that the online algorithm works in
iterations, where an iteration begins whenever the value of $\alpha$ is doubled. Now consider the last two
iterations of the algorithm $I - 1$ and $I$. At the end of iteration $I - 1$, $c_{tot}^{frac} > 2\alpha \log(M) +
\alpha + 1$, which means that our current guess of the optimal solution $\alpha$ is less than $OPT$, since in this
case the cost of the online solution is not $\mathcal{O}(\log (M))$-competitive. Therefore, we double $\alpha$ and
start the next iteration. In the worst case, the value of $\alpha$ before doubling is such that $\vert OPT -
\alpha \vert \leq \epsilon$, where $\epsilon > 0$ is a very small number. When $\alpha$ is doubled for the last
iteration, $\alpha$ is within a factor of 2 of $OPT$ when the algorithm terminates.}
\end{comment}
Now we prove a lower bound on the competitive ratio achieved by any online algorithm by proving the following
theorem:

\begin{theorem}\label{theorem:lower}
The best competitive ratio achieved by any online algorithm is $\Omega(\log(M))$.
\end{theorem}

\begin{proof}
\txtgreen{To prove this theorem, we provide an example network such that any online algorithm implemented over
this network cannot achieve a competitive ratio better than $\Omega(\log(M))$.} Consider a system with $M$ APs and
$\log(M)$ users. The first user is connected to all APs, the second user is connected to the first $\frac{M}{2}$
APs, the $n$-th user is connected to the first $\frac{M}{2^{n-1}}$. The system has the following properties:
\begin{itemize}
\item The power consumption required to turn on an AP is the same for all APs and is equal to $f$.
\item The power consumption required for data transmission between a user and an AP is equal to 0 if the user is
    connected to the AP, and is otherwise equal to infinity.
\end{itemize}
The system used in the proof is shown in Fig.~\ref{fig:lower}.

In this system, users arrive sequentially and in-order starting from the first user up to the last user.
Sequentially means that the online algorithm has to satisfy a user before considering the next one. Consider the
$n$-th user. The probability that the $n$-th user is connected to an AP, which was turned on to satisfy the
($n-1$)-th user is, $\frac{1}{2}$. This is because the AP that was turned on due to the ($n-1$)-th user is either
connected to both the ($n-1$)-th and $n$-th users, in which case the $n$-th user is satisfied without turning on
any additional APs, or the turned on AP is connected to the ($n-1$)-th user but is not connected to the $n$-th
user, in which case the $n$-th user is satisfied by turning on an additional AP. Therefore, the total power
consumption of the online algorithm is $\sum_{n=1}^{\log(M)}\frac{f}{2} = \frac{f}{2}\log(M)$. The optimal
algorithm will turn on a single AP connected to the last user, and so the total power consumption of the optimal
algorithm is \txtblue{$f$}. Therefore, the best competitive ratio achieved by any online algorithm is
$\Omega(\log(M))$.
\end{proof}

\section{Numerical Results}\label{sim}
In this section, we present the simulation results for the following schemes:

\begin{itemize}
\item Hybrid: our optimal hybrid scheme.
\item VLC: in this scheme, communications is performed using the VLC APs only and the same optimization problem
    considered for the hybrid scheme is performed but only on the VLC APs.
\item WiFi: in this scheme, communications is performed using the WiFi APs only and the same optimization
    problem considered for the hybrid scheme is performed but only on the WiFi APs.
\item Online: our proposed online algorithm.
\end{itemize}

Note that the first three schemes are offline and represented by NP-complete problems. The simulations are
conducted using Matlab R2013b. CPLEX 12.6.1 \cite{cplex2014v12} is utilized to run the offline optimization
problems.

\subsection{Simulation settings}
We consider the first floor of the Electrical and Computer Engineering Department building at New Jersey Institute
of Technology. The entire storey consists of 16 external rooms, 4 internal rooms, annular corridors and stairways
located in corners. Each external room has a lateral window while the internal rooms do not. In the simulations,
we assume that user terminals (UTs) are uniformly distributed at random in rooms. There are 4 WiFi APs deployed in
corners on the fourth floor. In each room, one light source equipped with 4 VLC APs is mounted on the ceiling. The
beamangle (i.e. formed by the central luminous flux and the central vertical line) of each VLC AP is
pre-configured in order to position the intersection of the center luminous flux and the horizontal UTs plane at
the center of each square region. The benefit of this new light source configuration is manifested in
\cite{shao2015joint}. These system model parameters are summarized in Table~\ref{table_system_model}.

\begin{table}
\centering
\caption{System Model Parameters}
\begin{tabular}{c c}
\hline
\hline
floor size&18.0 m $\times$ 18.0 m $\times$ 3.0 m\\
room size&3.0 m $\times$ 3.0 m $\times$ 3.0 m\\
number of internal rooms&4\\
number of external rooms&16\\
height of desk&0.85 m\\
number of VLC APs&80\\
number of WiFi APs&4\\
\hline
\hline
\end{tabular}
\label{table_system_model}
\end{table}

For a VLC AP, we set $P_{m}^{on}=15$ watts. The wallplug efficiency factors $\eta_{m}^{DC}$ and $\eta_{m}^{AC}$
denote the ratios of emitted optical power to injected electrical power when the $m$-th VLC AP is only sending DC
signals and AC signals, respectively. $\eta_{m}^{DC}$ for all the VLC APs is set to 0.1
\cite{bright2014efficient}, while $\eta_{m}^{AC}$ is varied to represent different driver circuitry. The channel
bandwidth of each VLC AP is 100 MHz. In each room, 4 VLC APs transmit data over 4 different channels, while other
4 VLC APs in a different room can reuse those 4 channels due to the obstacles posed by walls. When multiple UTs
are connected to the same VLC AP, they access the same spectrum resource by time division multiplexing (TDM). We
assume that, in any fraction of one time slot, a VLC AP is either sending AC signals or DC signals. Thus, the
maximum power consumption of a VLC AP $P_{m}^{max}=P_{m}^{on}\times\eta_{m}^{DC}/\eta_{m}^{AC}$. To calculate
$P_{mn}$ for the $n$-th UT connected to the $m$-th VLC AP, we start by calculating the VLC link capacity $C_{mn}$
based on Shannon-Hartley theorem \cite{cover2012elements}, assuming the VLC AP is always sending AC signals and
the AC signal strength is twice that of the average DC level. Then, the additive power consumed by the $n$-th UT
is \txtgreen{$\frac{R_{n}}{C_{mn}}\times P_{m}^{on}\times(\frac{\eta_{m}^{DC}}{\eta_{m}^{AC}}-1)$}, where $R_{n}$
is the required data rate of the $n$-th UT. The semi-angle at half power of each VLC AP is set to 30$^\circ$. The
constant Gaussian noise, dominated by thermal noise, is calculated from the parameters in
\cite{komine2004fundamental} and set to be 4.7$\times$10$^{-14}$ A$^{2}$. The receiver parameters (i.e., FOV of
receiver, detector area of a photodiode, gain of optical filter, refractive index of lens and O/E conversion
efficiency) are the same as those in \cite{komine2004fundamental}. The required illuminance level is above 300 lux
and the LED luminosity efficacy is 150 lumen per watt of electricity \cite{Zyga2010white}. These VLC parameters
are summarized in Table~\ref{table_VLC}.

\begin{table}
\centering
\caption{VLC Parameters}
\begin{tabular}{c c}
\hline
\hline
$P_{m}^{on}$ of VLC AP&15 W\\
semi-angle at half power&30$^\circ$\\
VLC channel bandwidth&100 MHz\\
constant Gaussian noise&4.7$\times$10$^{-14}$ A$^{2}$\\
detector area of photodiode&1.0 cm$^{2}$\\
O/E conversion efficiency&0.54 A/W\\
gain of optical filter&1.0\\
refractive index of lens&1.5\\
FOV of receiver&90 deg.\\
LED luminosity efficacy&150 lm/W\\
DC efficiency factor&0.1\\
\hline
\hline
\end{tabular}
\label{table_VLC}
\end{table}

For a WiFi AP, we set $P_{m}^{on}=10$ watts \cite{WiFiEnergy} and the maximum power consumption of a WiFi AP
$P_{m}^{max}=4$ watts for the NETGEAR model WNDR3400v3. Since WiFi APs are deployed on the fourth floor, the WiFi
signal strength attenuation due to floors' obstruction is set to -30 dB \cite{house2010rf}. A typical WiFi noise
level $N_{0}$ is -90 dBm \cite{kessler2011diagnosing}. When multiple UTs are connected to the same WiFi AP, they
access the network simultaneously by frequency division multiplexing (FDM). The WiFi bandwidth $W$ allocated to
each UT is set to 2 MHz. To calculate $P_{mn}$ for the $n$-th UT connected to the $m$-th WiFi AP, we start by
calculating the received power to satisfy the required data rate $R_{n}$ at the $n$-th user
$P_{rx}=N_{0}(2^{\frac{R_{n}}{W}}-1)$. Based on $P_{rx}$, we calculate the transmitted RF power $P_{tx}$ using
Friis equation \cite{rappaport1996wireless}. When the gains of Tx and Rx antennas are set to 1,
$P_{tx}=P_{rx}(\frac{4\pi r_{mn}}{\lambda})^{2}$, where $\lambda=0.125$ meters according to the 2.4 GHz carrier
frequency. To convert the transmitted RF power to electrical power consumption of the WiFi AP, we divide $P_{tx}$
by an efficiency factor $\eta_{WiFi}$ obtained from \cite{ebert2002measurement}. According to Fig.~4(a) in
\cite{ebert2002measurement}, if the RF power level is changed from 1 to 50 mW, the increase in electrical power
consumption is about 500 mW. Therefore, we set $\eta_{WiFi}$ to be 0.1. These WiFi parameters are summarized in
Table~\ref{table_WiFi}.

\begin{table}
\centering
\caption{WiFi Parameters}
\begin{tabular}{c c}
\hline
\hline
$P_{m}^{on}$ of WiFi AP&10 W\\
WiFi channel bandwidth&2 MHz\\
Carrier Frequency&2.4 GHz\\
Maximum power consumption of WiFi AP&14 W\\
WiFi noise level&-90 dBm\\
Attenuation due to floors&-30 dB\\
Efficiency factor&0.1\\
\hline
\hline
\end{tabular}
\label{table_WiFi}
\end{table}

To evaluate the ambient light level in each room, we utilize the concept of daylight factor (DF) introduced in
\cite{li2006average}. In particular, DF$=(E_{i}/E_{o})\times100\%$, where $E_{i}$ is illuminance due to daylight
at a point on the indoors working plane and $E_{o}$ is simultaneous outdoor illuminance on a horizontal plane from
an unobstructed hemisphere of overcast sky. In external rooms, DF increases when the evaluated point get closer to
the lateral window. While in internal room, DF~$=0$ for the entire UTs plane. According to
\cite{littlefair1988measurements}, the luminous efficacy of sunlight $\rho=93$ lm/W, where the watt represents
optical power. The variation of solar radiation $R_{sun}$ [W/m$^{2}$] is collected from \cite{Soda}. Given these
parameters, the indoor ambient light level at position $w$ can be estimated by
$\mathcal{I}_{ambient}^{w}=~$DF$^{w}\times\rho\times R_{sun}\times0.01$.

\subsection{Simulation Results}
We measure the power consumption of all schemes (i.e. Hybrid, VLC, WiFi and Online) when we change the throughput
requirement per UT, the number of UTs, and the hour of the day. For different throughput requirements and number
of UTs, we simulate the power consumption for day ($R_{sun}=110$ W/m$^{2}$) and night ($R_{sun}=0$ W/m$^{2}$). The
results are averaged over 100 runs and shown in Fig.~\ref{fig_throughput_night} to Fig.~\ref{fig_hours}. In
different figures, we vary $\eta_{m}^{AC}$ from 0.06 to 0.09, in order to evaluate the performance of the four
schemes as the efficiency of the VLC driver circuitry improves. The power consumption we measure here is after
subtracting the minimum power consumption required for illumination.

\txtgreen{First, from all the figures, we note that the power consumption of the online scheme is at 2-4 times the
power consumption of the hybrid scheme, which is within the $\mathcal{O}(\log(N)\log(M))$ factor proved in
Section~\ref{sec:PerformanceAnalysis}.}

In Fig.~\ref{fig_throughput_night} and Fig.~\ref{fig_throughput_day}, 100 UTs are distributed uniformly at random.
At night (Fig.~\ref{fig_throughput_night}), as all the VLC APs are turned on for illumination, the power
consumption for communication of VLC scheme \txtgreen{and online scheme} is much lower than that of WiFi scheme.
Note that the sudden increases of WiFi performance curve (i.e. green curve) are due to the turning-on of
additional WiFi APs. The total additive power of VLC scheme is even lower than the power consumption of turning on
a WiFi AP, and thus the hybrid scheme performs the same as VLC. As the wallplug efficiency factor for AC signals
$\eta_{m}^{AC}$ increases, the power consumption of VLC and hybrid schemes become negligible, and the performance
of online scheme is highly enhanced. During the day (Fig.~\ref{fig_throughput_day}), if communications is not
needed, half of the VLC APs are turned on for illumination. When the number of UTs is large, it may force most of
the VLC APs to be turned on for communications and produce unnecessary illumination. The results in
Fig.~\ref{fig_throughput_day} validate the analysis, manifested as the higher power consumption of VLC scheme
compared to that of WiFi. It is worth noticing that, at 2.5 Mbps, an extra WiFi AP needs to be turned on, which
leads to the superiority of the hybrid scheme. Since some UTs are located in the square regions such that the
corresponding VLC APs have already been turned on for illumination, those UTs will be connected to VLC APs in the
hybrid scheme instead of being connected to an extra WiFi AP. In contrast to the results at night, the performance
of online scheme during daytime does not change much as the $\eta_{m}^{AC}$ increases. \txtgreen{We also note from
the figures that the the online scheme consumes half the power compared to WiFi scheme at high throughputs during
the night, and consumes 7 to 17 times less power than the VLC scheme during the day.}

In Fig.~\ref{fig_numofUT_night} and Fig.~\ref{fig_numofUT_day}, the throughput requirement per UT is 6 Mbps. The
results are similar to those shown in Fig.~\ref{fig_throughput_night} and Fig.~\ref{fig_throughput_day}. The only
noticeable point is that in Fig.~\ref{fig_numofUT_day}, as the number of UTs increases or in other words the
uniformity of UTs increases, the probability of a UT located in a square region such that the corresponding VLC AP
has not been turned on for illumination increases. Therefore, the power consumption of the VLC scheme increases
fast, as the number of UTs increases. \txtgreen{We also note from the figures that the power consumption of WiFi
can reach 4 times the power consumption of the online scheme during the night, and the power consumption of the
VLC scheme can reach 7 times the power consumption of the online scheme during the day.}

In Fig.~\ref{fig_hours}, we set the number of UTs to 100 and throughput requirement per UT to 6 Mbps, to evaluate
the performance under the crowded environment. It can be seen that, from 6am to 7pm, the WiFi scheme outperforms
the VLC scheme due to the large power penalty of turning on VLC APs, and also the hybrid scheme outperforms both
WiFi and VLC schemes due to the fact that some UTs are located under the turned-on VLC APs and the total additive
power caused by those UTs are less than the power consumption of turning on a WiFi AP. From 7pm to 6am, the power
consumption of the VLC scheme is negligible as compared to the WiFi scheme, and thus the performance of the hybrid
scheme is the same as that of VLC. The power-saving of the hybrid scheme over WiFi and VLC schemes could be over
90\%. When $\eta_{m}^{AC}=0.09$, the online scheme saves around 75\% of the power consumption of the WiFi scheme
from 6am to 7pm. While from 7pm to 6am, the online scheme consumes slightly more power than the WiFi scheme.
However, this additional power consumption is acceptable since the online scheme does not have the knowledge of
the future arrivals of UTs while the WiFi scheme does.

\begin{figure*}
    \centering
    \begin{subfigure}[b]{0.238\textwidth}
        \includegraphics[width=\textwidth]{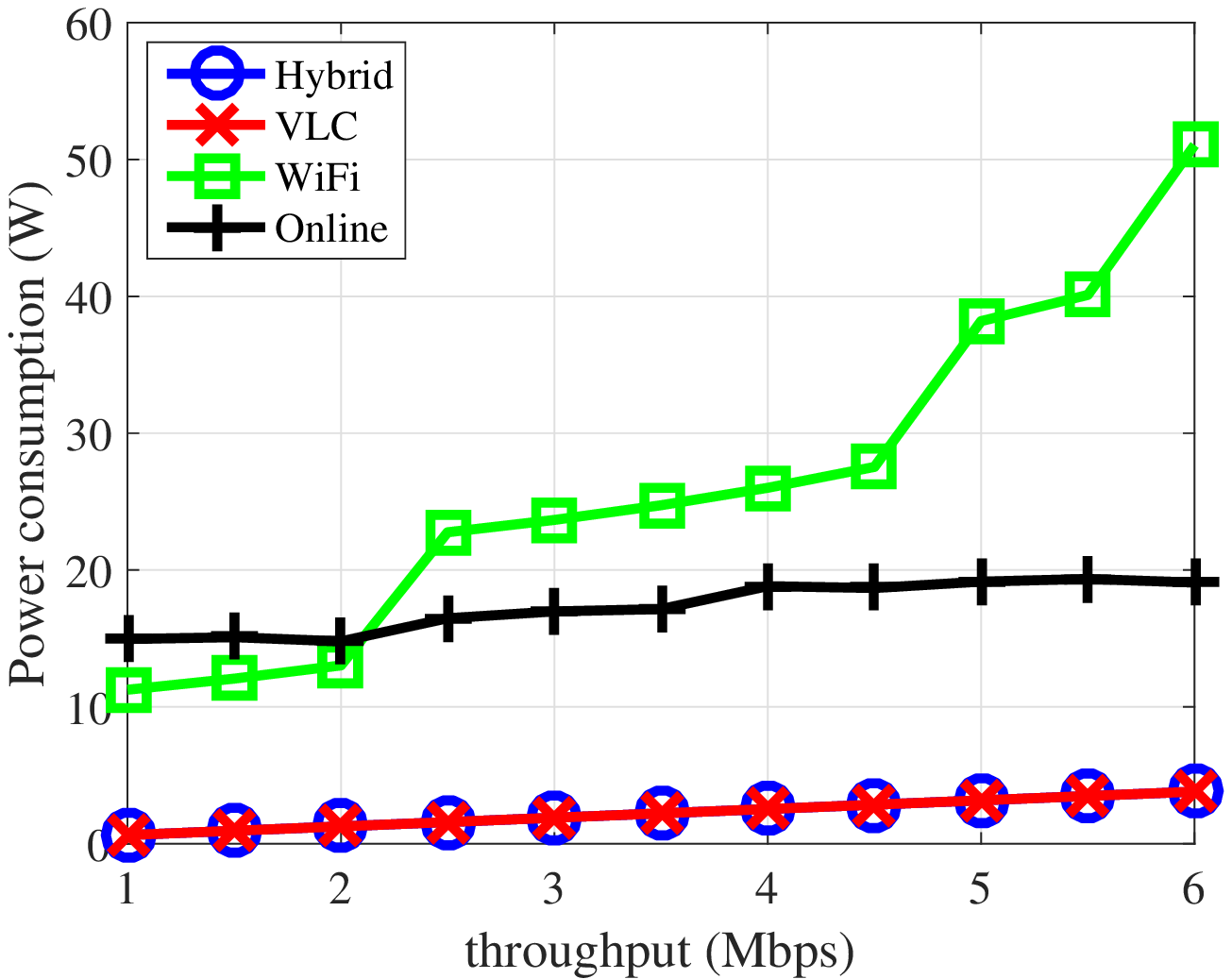}
        \caption{$\eta_{m}^{AC}=0.06$}
        \label{fig_throughput_night006}
    \end{subfigure}
    ~ %add desired spacing between images, e. g. ~, \quad, \qquad, \hfill etc.
      %(or a blank line to force the subfigure onto a new line)
    \begin{subfigure}[b]{0.238\textwidth}
        \includegraphics[width=\textwidth]{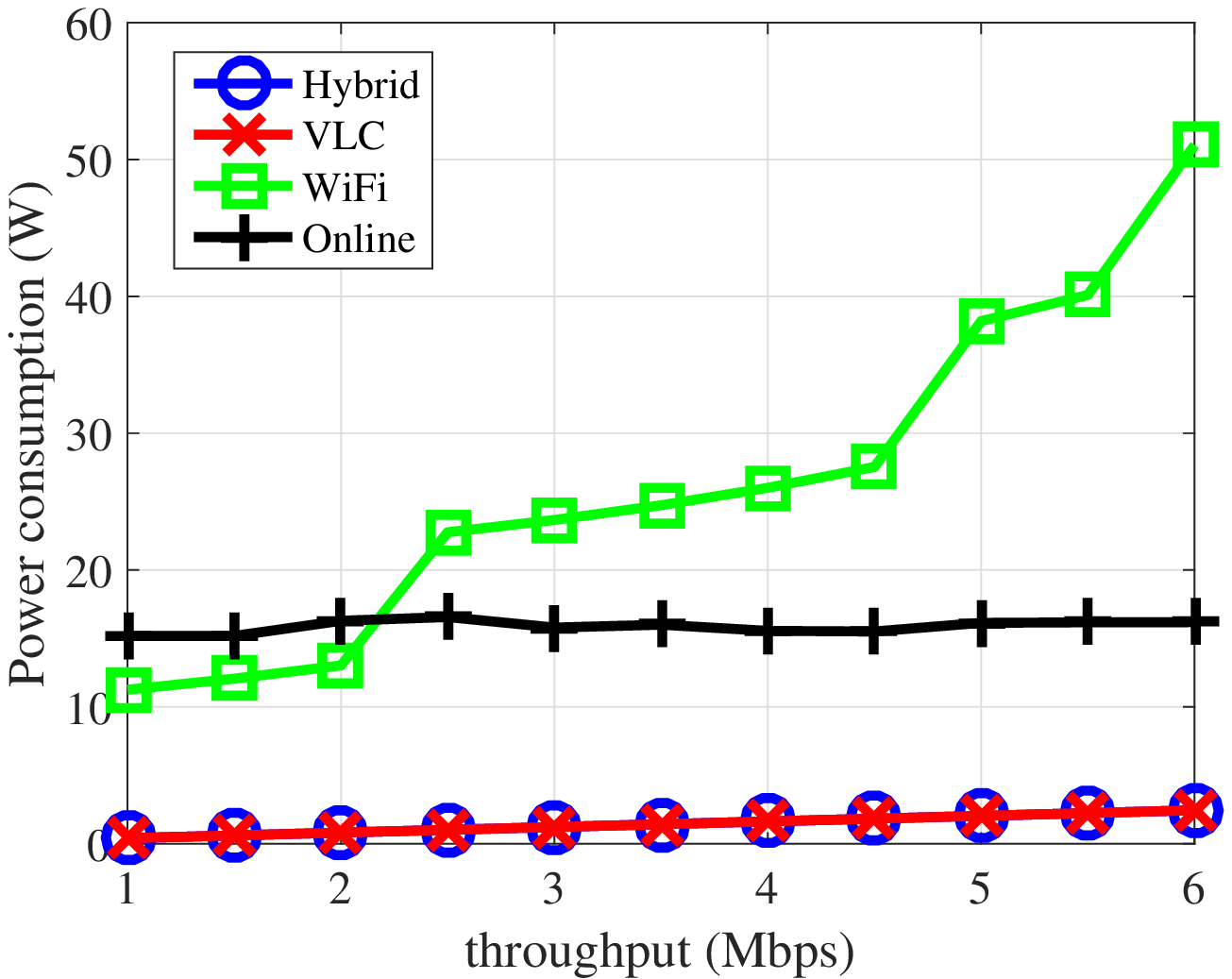}
        \caption{$\eta_{m}^{AC}=0.07$}
        \label{fig_throughput_night007}
    \end{subfigure}
    ~ %add desired spacing between images, e. g. ~, \quad, \qquad, \hfill etc.
    %(or a blank line to force the subfigure onto a new line)
    \begin{subfigure}[b]{0.238\textwidth}
        \includegraphics[width=\textwidth]{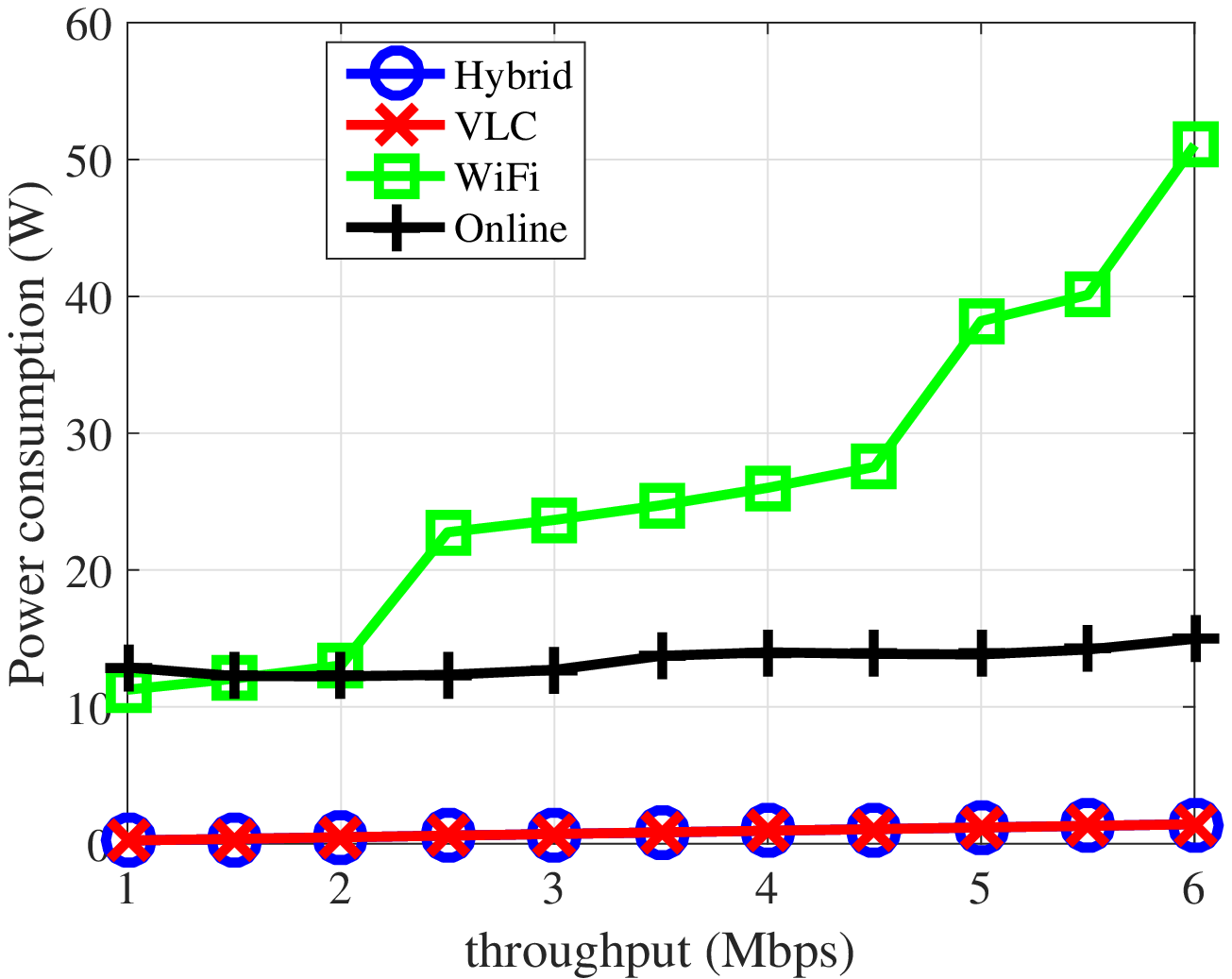}
        \caption{$\eta_{m}^{AC}=0.08$}
        \label{fig_throughput_night008}
    \end{subfigure}
    \begin{subfigure}[b]{0.238\textwidth}
        \includegraphics[width=\textwidth]{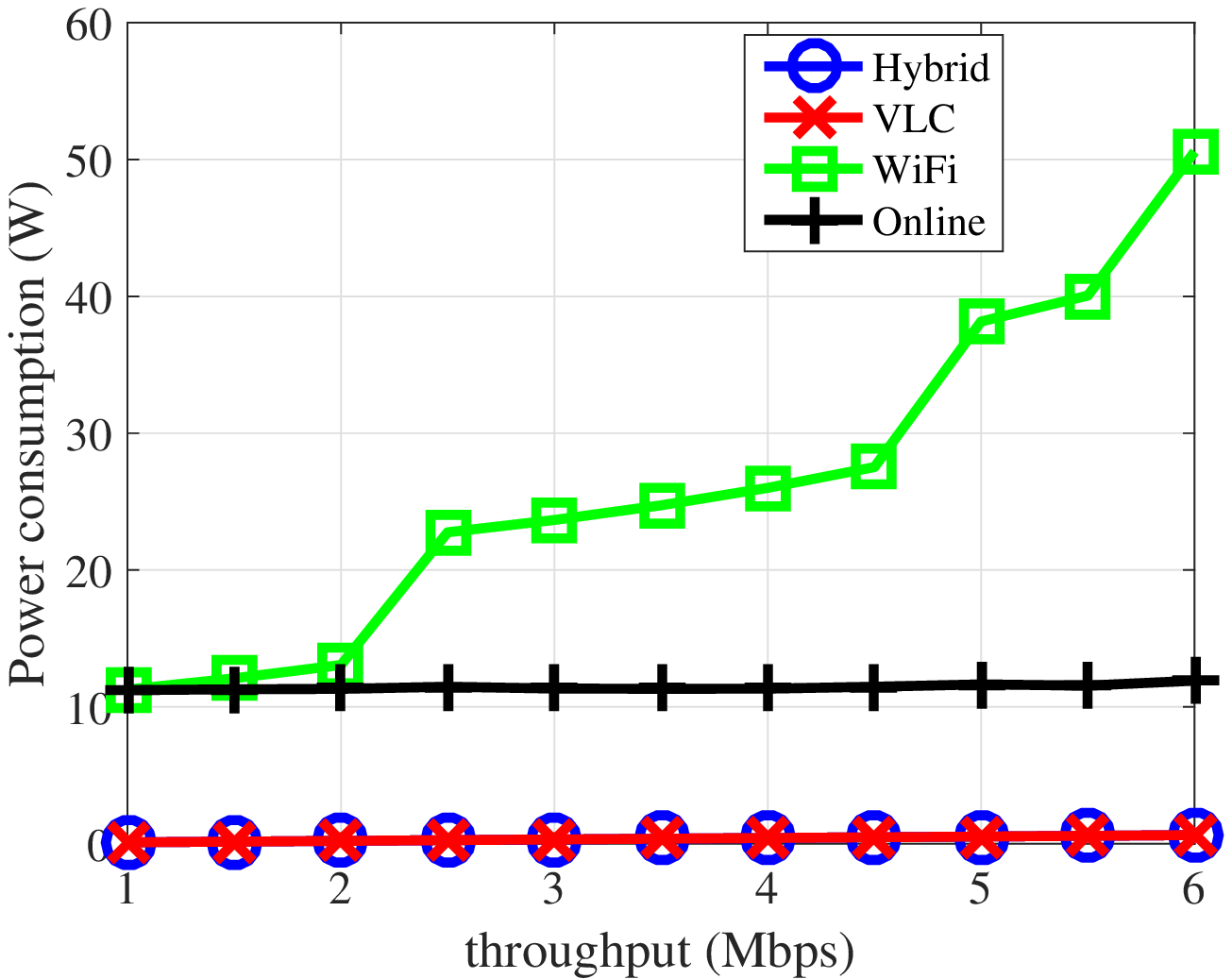}
        \caption{$\eta_{m}^{AC}=0.09$}
        \label{fig_throughput_night009}
    \end{subfigure}
    \caption{Power consumption in terms of throughput requirement per user at night}\label{fig_throughput_night}
\end{figure*}

\begin{figure*}
    \centering
    \begin{subfigure}[b]{0.238\textwidth}
        \includegraphics[width=\textwidth]{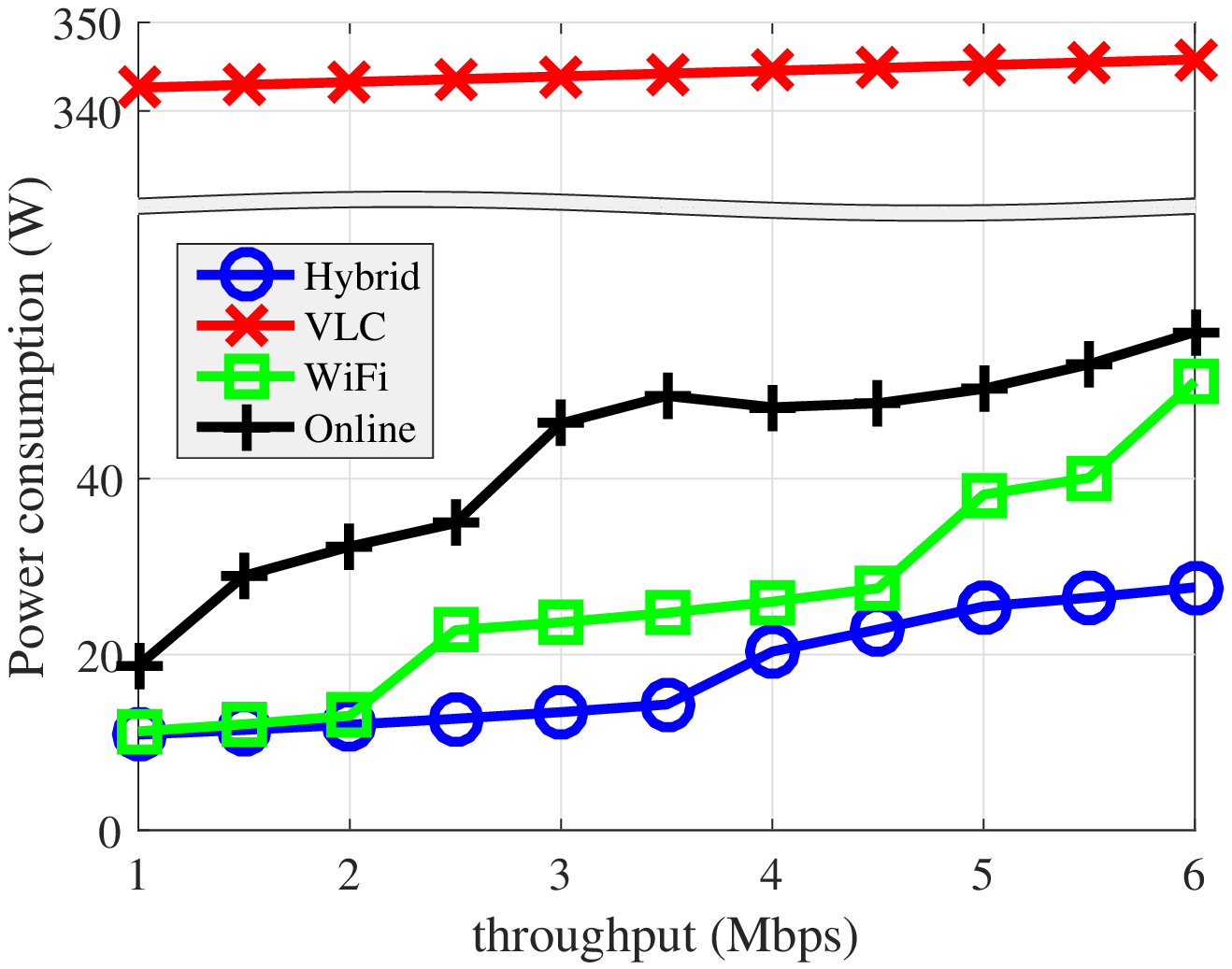}
        \caption{$\eta_{m}^{AC}=0.06$}
        \label{fig_throughput_day006}
    \end{subfigure}
    ~ %add desired spacing between images, e. g. ~, \quad, \qquad, \hfill etc.
      %(or a blank line to force the subfigure onto a new line)
    \begin{subfigure}[b]{0.238\textwidth}
        \includegraphics[width=\textwidth]{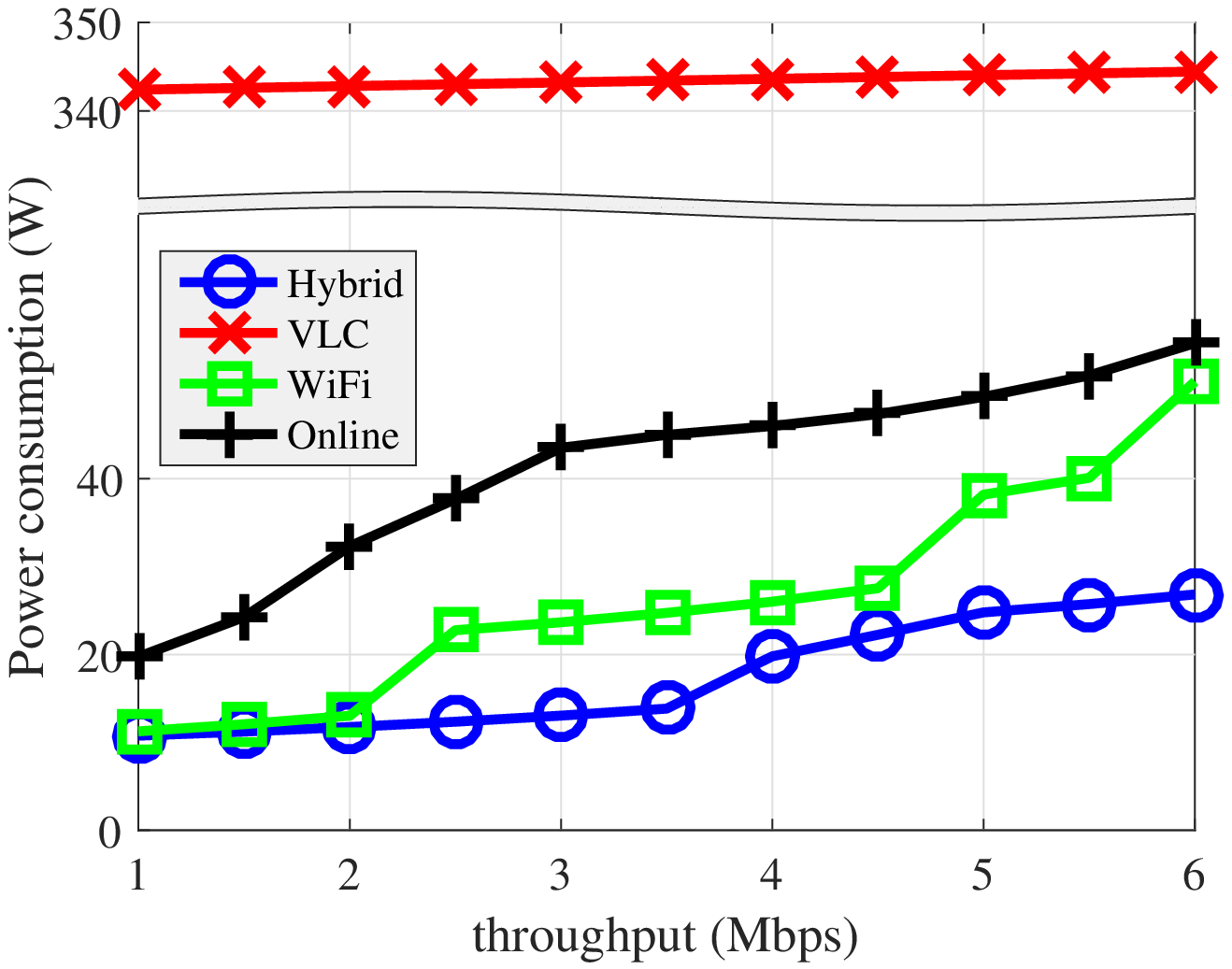}
        \caption{$\eta_{m}^{AC}=0.07$}
        \label{fig_throughput_day007}
    \end{subfigure}
    ~ %add desired spacing between images, e. g. ~, \quad, \qquad, \hfill etc.
    %(or a blank line to force the subfigure onto a new line)
    \begin{subfigure}[b]{0.238\textwidth}
        \includegraphics[width=\textwidth]{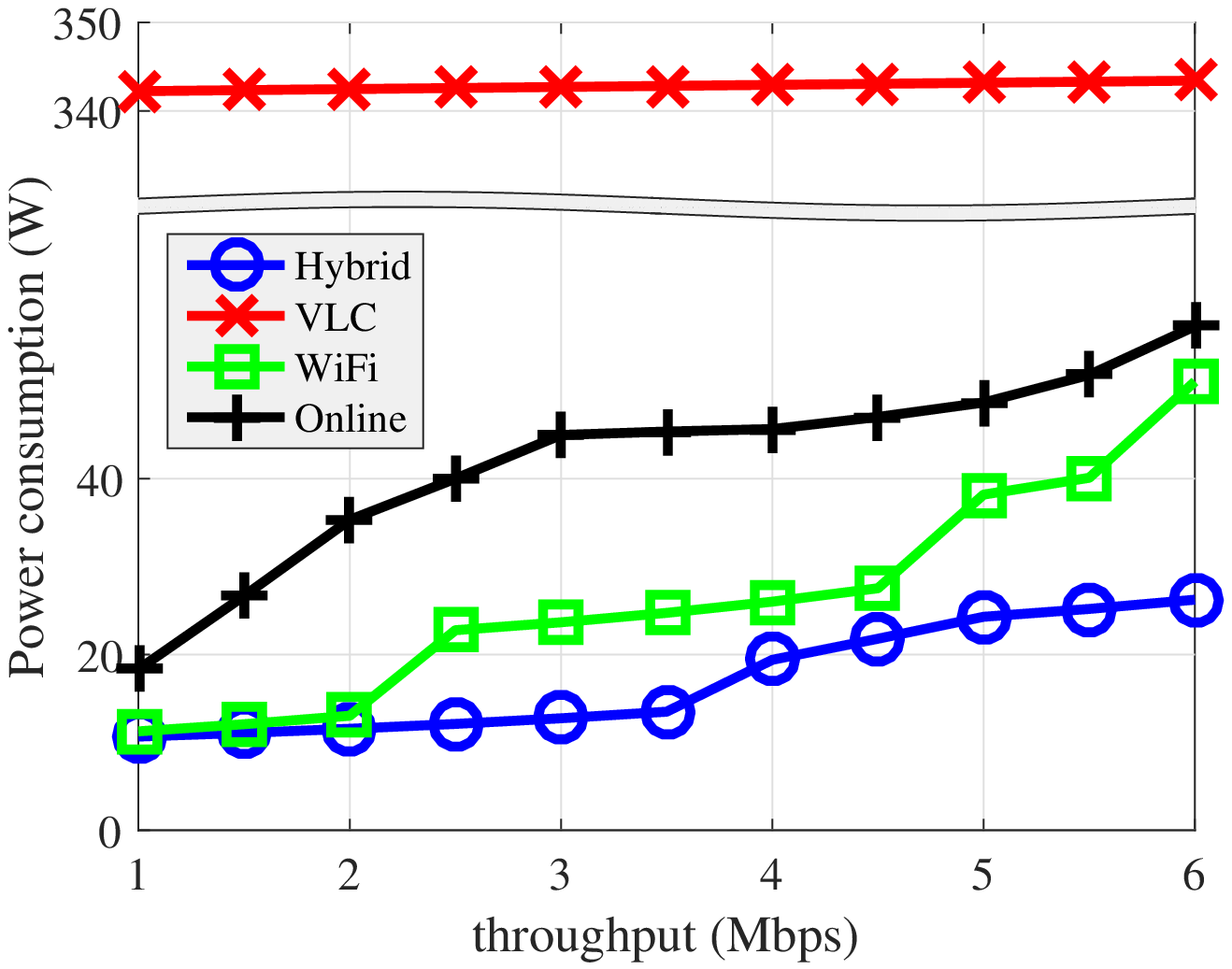}
        \caption{$\eta_{m}^{AC}=0.08$}
        \label{fig_throughput_day008}
    \end{subfigure}
    \begin{subfigure}[b]{0.238\textwidth}
        \includegraphics[width=\textwidth]{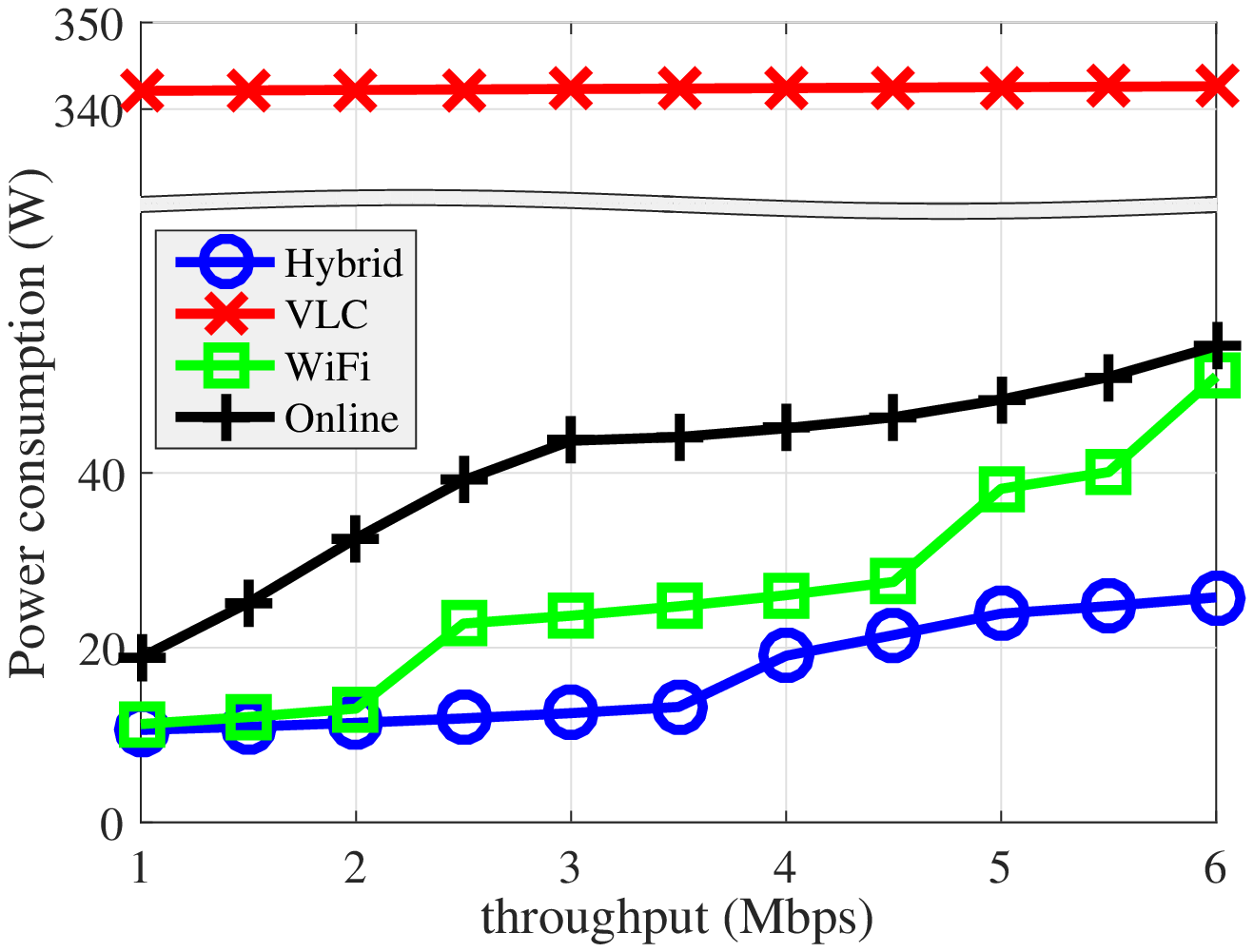}
        \caption{$\eta_{m}^{AC}=0.09$}
        \label{fig_throughput_day009}
    \end{subfigure}
    \caption{Power consumption in terms of throughput requirement per user at day}\label{fig_throughput_day}
\end{figure*}

\begin{figure*}
    \centering
    \begin{subfigure}[b]{0.238\textwidth}
        \includegraphics[width=\textwidth]{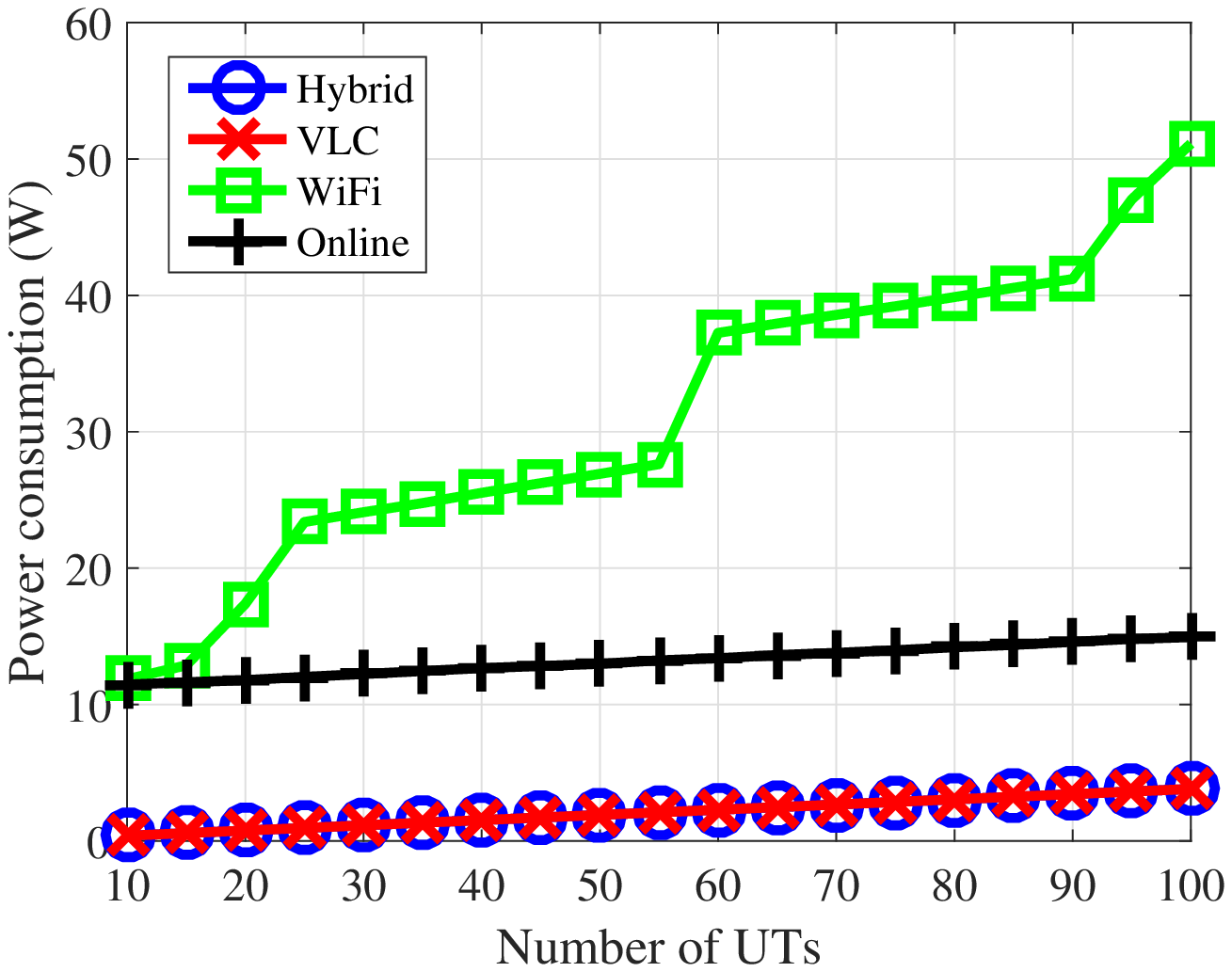}
        \caption{$\eta_{m}^{AC}=0.06$}
        \label{fig_numofUT_night006}
    \end{subfigure}
    ~ %add desired spacing between images, e. g. ~, \quad, \qquad, \hfill etc.
      %(or a blank line to force the subfigure onto a new line)
    \begin{subfigure}[b]{0.238\textwidth}
        \includegraphics[width=\textwidth]{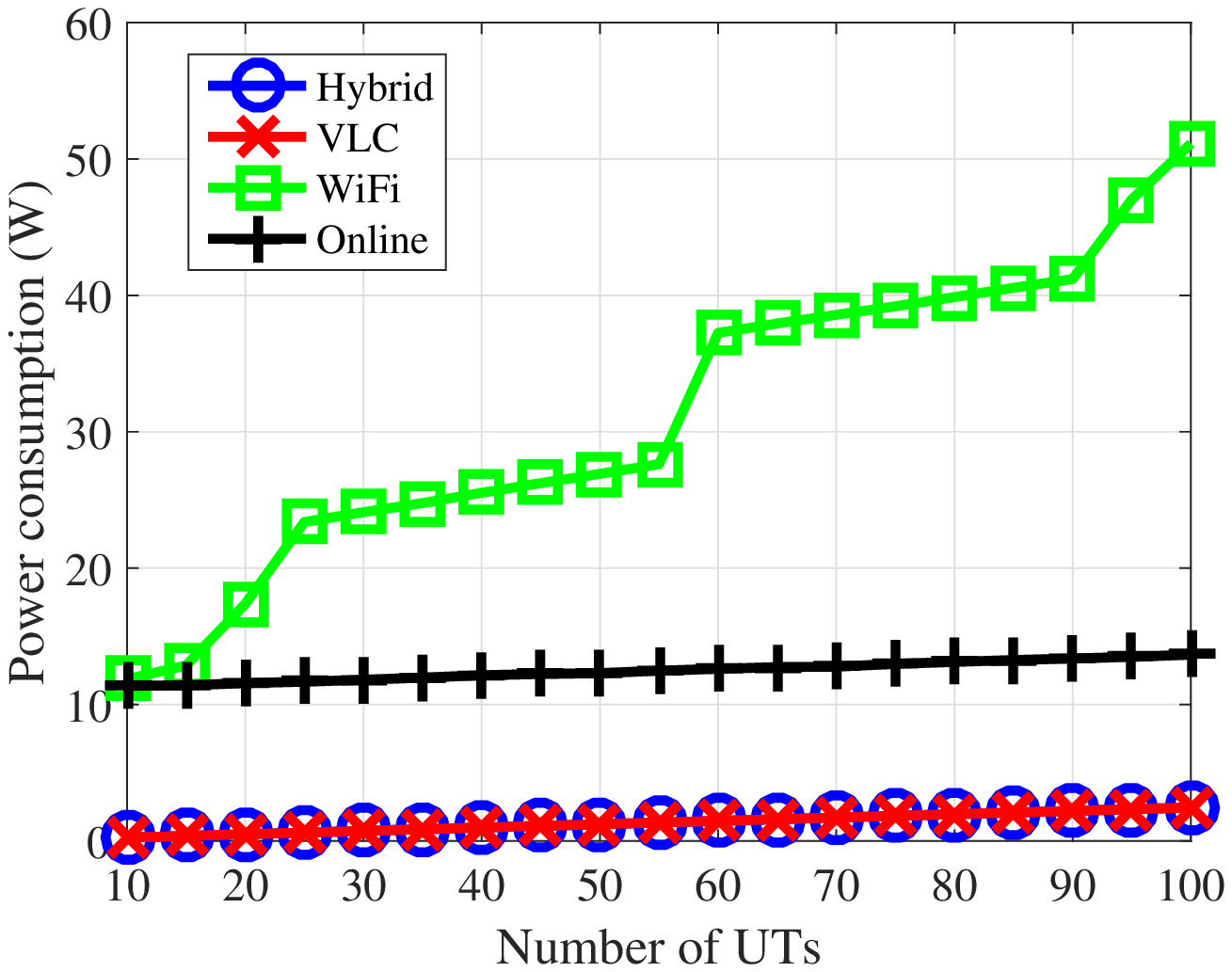}
        \caption{$\eta_{m}^{AC}=0.07$}
        \label{fig_numofUT_night007}
    \end{subfigure}
    ~ %add desired spacing between images, e. g. ~, \quad, \qquad, \hfill etc.
    %(or a blank line to force the subfigure onto a new line)
    \begin{subfigure}[b]{0.238\textwidth}
        \includegraphics[width=\textwidth]{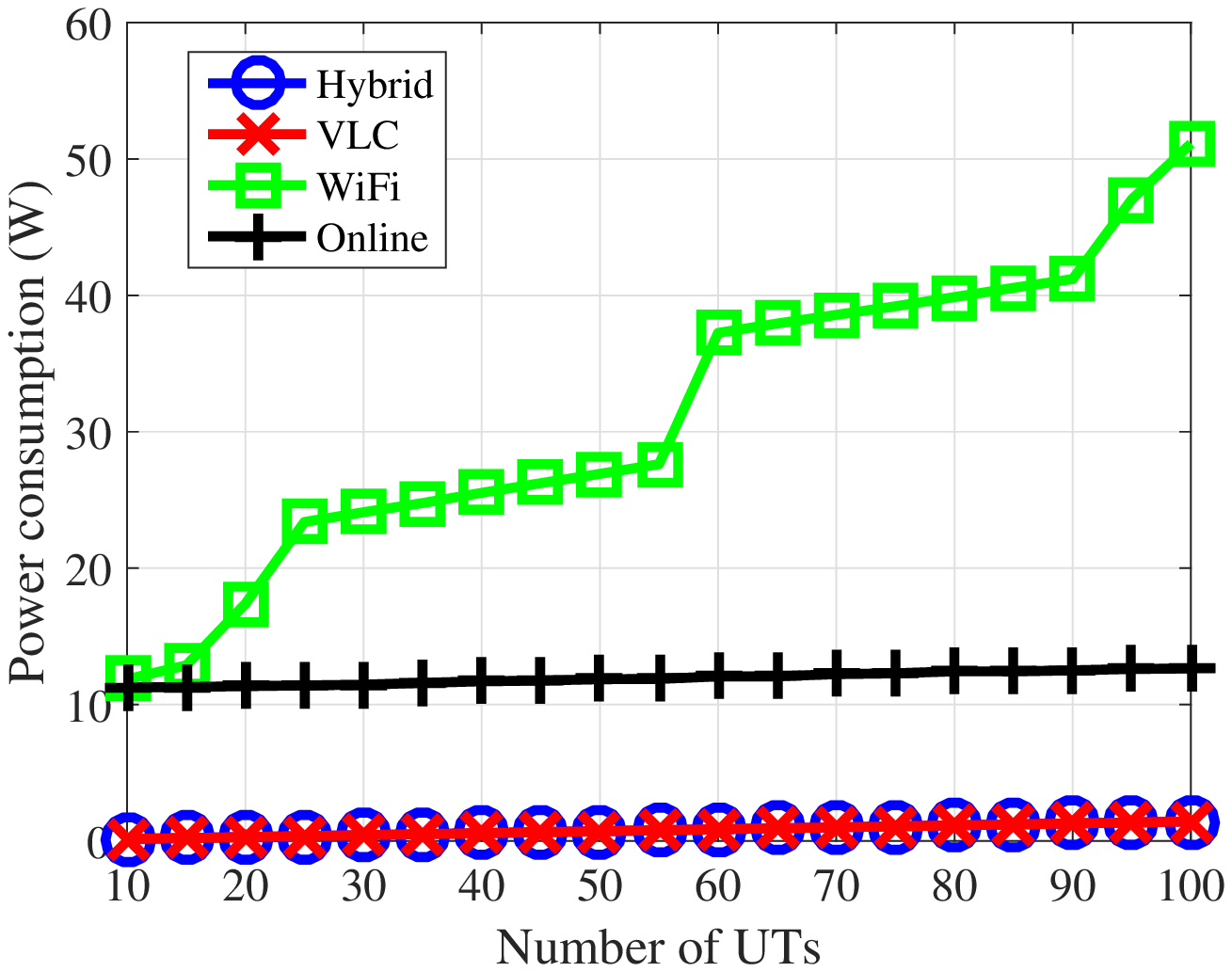}
        \caption{$\eta_{m}^{AC}=0.08$}
        \label{fig_numofUT_night008}
    \end{subfigure}
    \begin{subfigure}[b]{0.238\textwidth}
        \includegraphics[width=\textwidth]{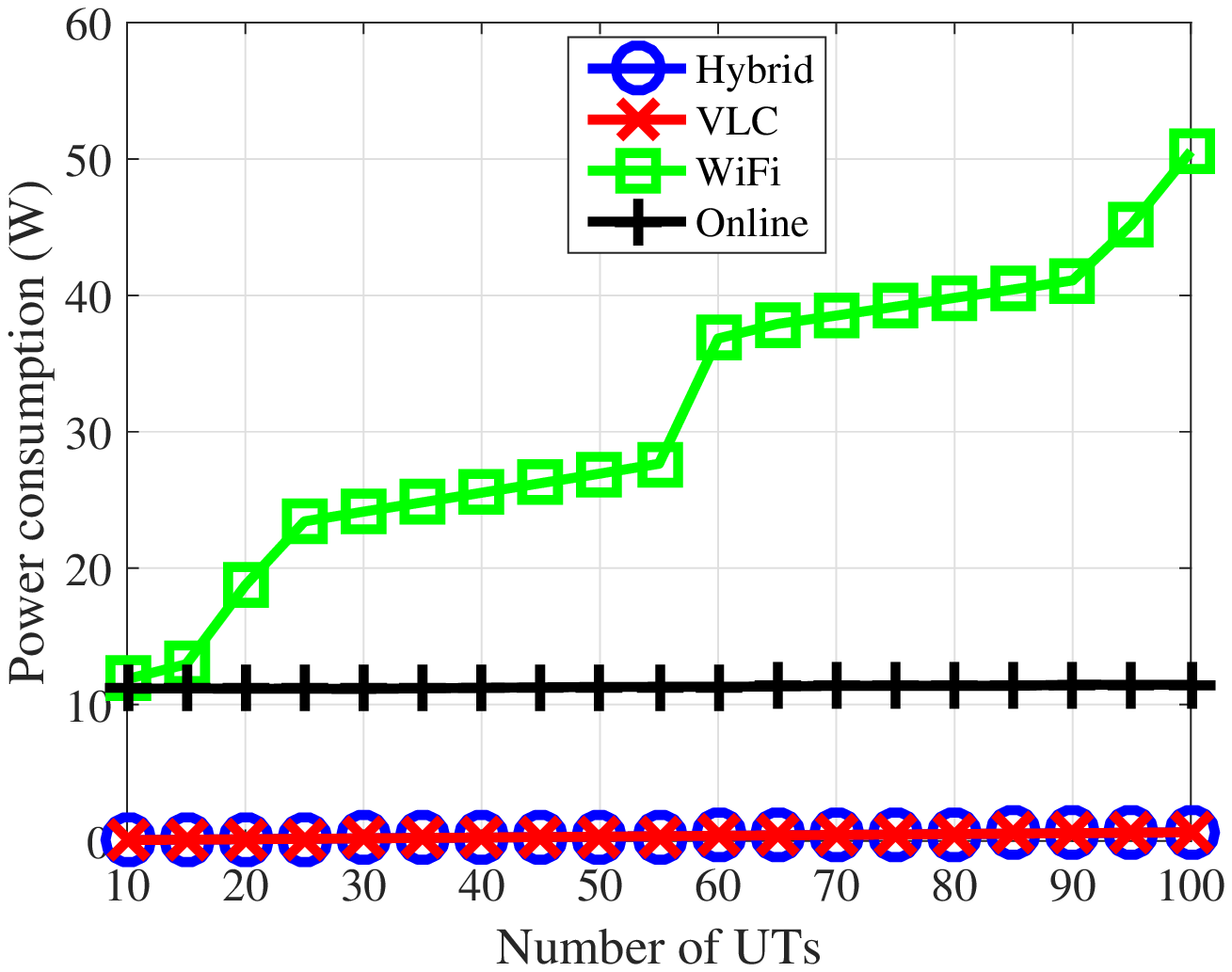}
        \caption{$\eta_{m}^{AC}=0.09$}
        \label{fig_numofUT_night009}
    \end{subfigure}
    \caption{Power consumption in terms of number of UTs at night}\label{fig_numofUT_night}
\end{figure*}

\begin{figure*}
    \centering
    \begin{subfigure}[b]{0.238\textwidth}
        \includegraphics[width=\textwidth]{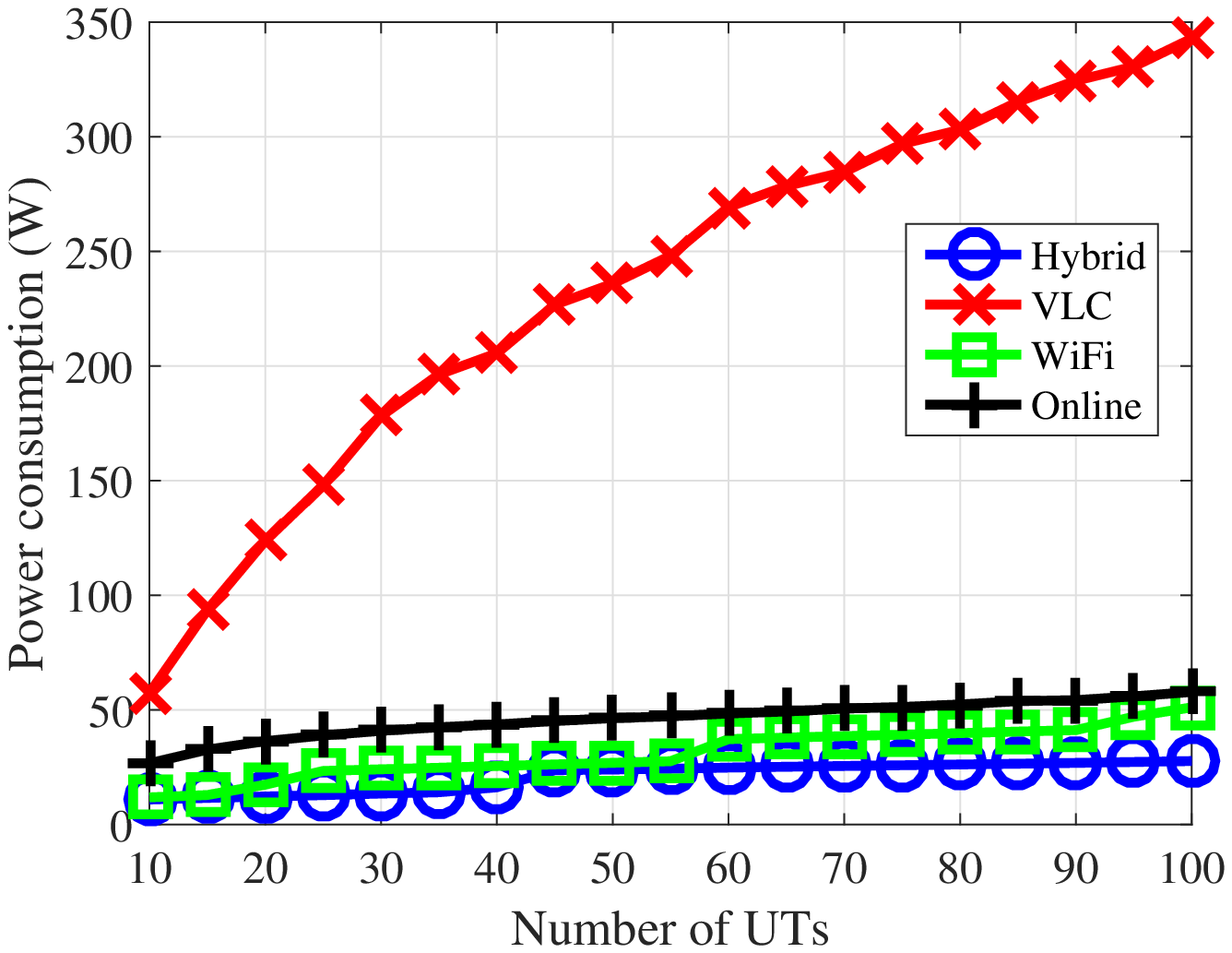}
        \caption{$\eta_{m}^{AC}=0.06$}
        \label{fig_numofUT_day006}
    \end{subfigure}
    ~ %add desired spacing between images, e. g. ~, \quad, \qquad, \hfill etc.
      %(or a blank line to force the subfigure onto a new line)
    \begin{subfigure}[b]{0.238\textwidth}
        \includegraphics[width=\textwidth]{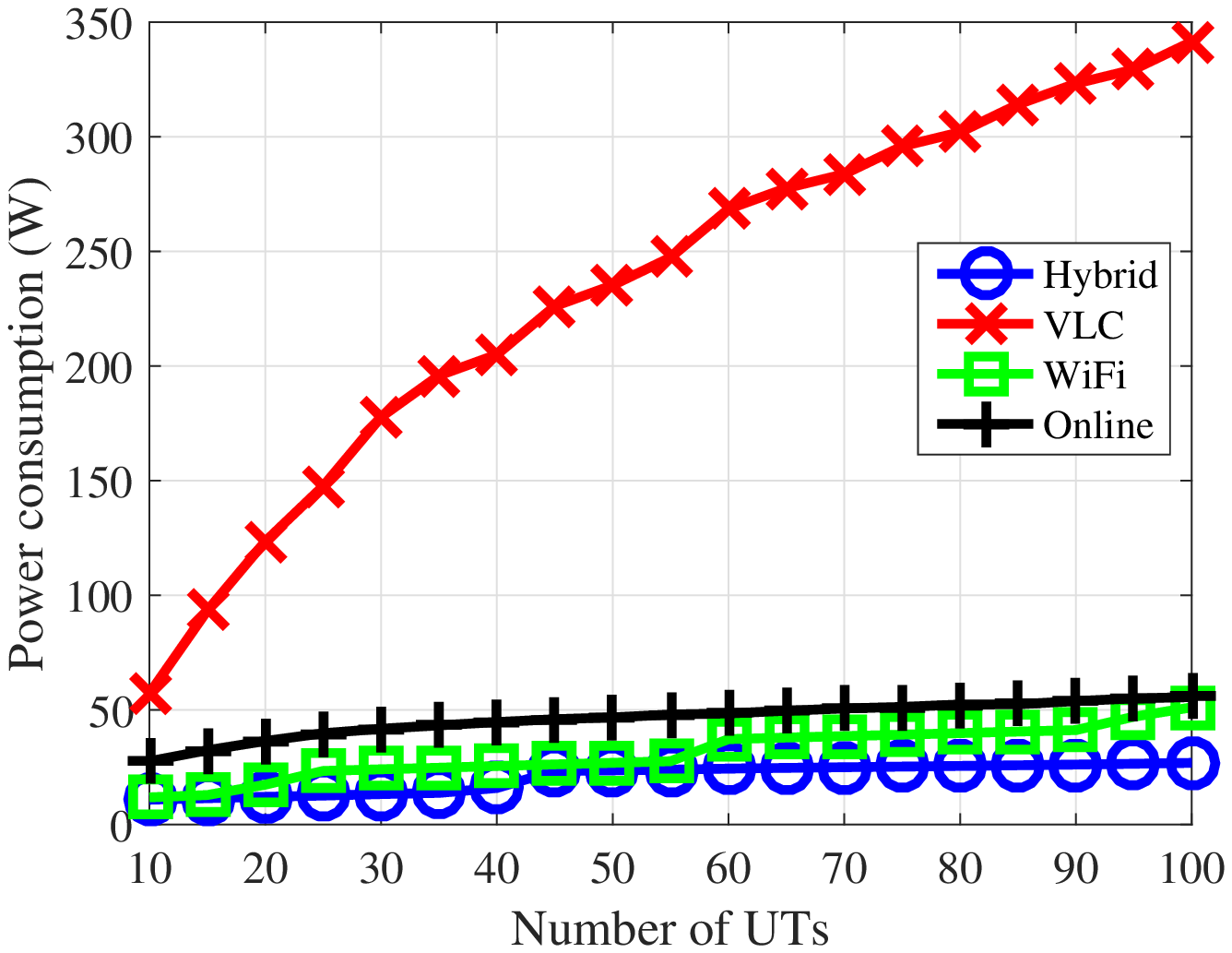}
        \caption{$\eta_{m}^{AC}=0.07$}
        \label{fig_numofUT_day007}
    \end{subfigure}
    ~ %add desired spacing between images, e. g. ~, \quad, \qquad, \hfill etc.
    %(or a blank line to force the subfigure onto a new line)
    \begin{subfigure}[b]{0.238\textwidth}
        \includegraphics[width=\textwidth]{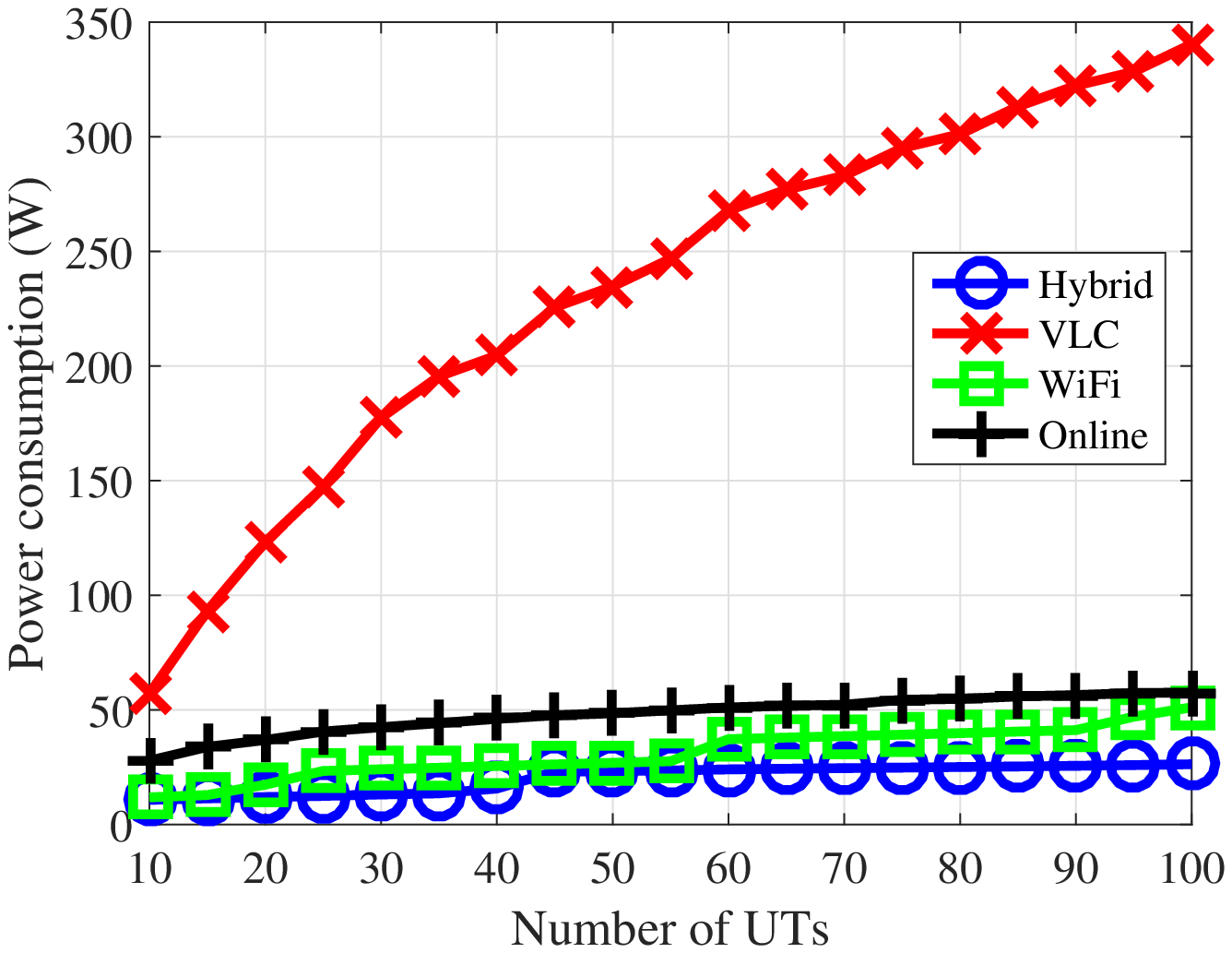}
        \caption{$\eta_{m}^{AC}=0.08$}
        \label{fig_numofUT_day008}
    \end{subfigure}
    \begin{subfigure}[b]{0.238\textwidth}
        \includegraphics[width=\textwidth]{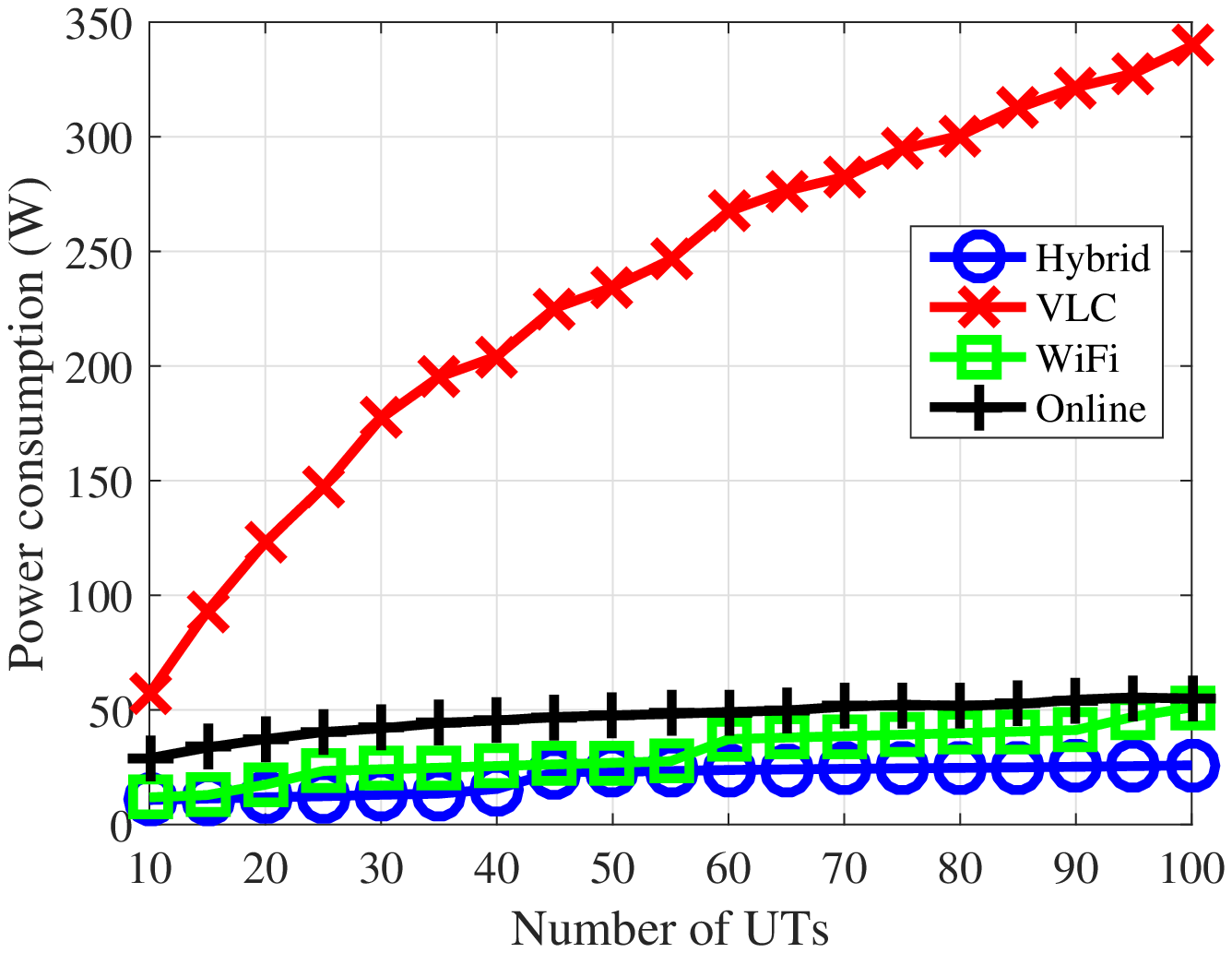}
        \caption{$\eta_{m}^{AC}=0.09$}
        \label{fig_numofUT_day009}
    \end{subfigure}
    \caption{Power consumption in terms of number of UTs at day}\label{fig_numofUT_day}
\end{figure*}

\begin{figure*}
    \centering
    \begin{subfigure}[b]{0.238\textwidth}
        \includegraphics[width=\textwidth]{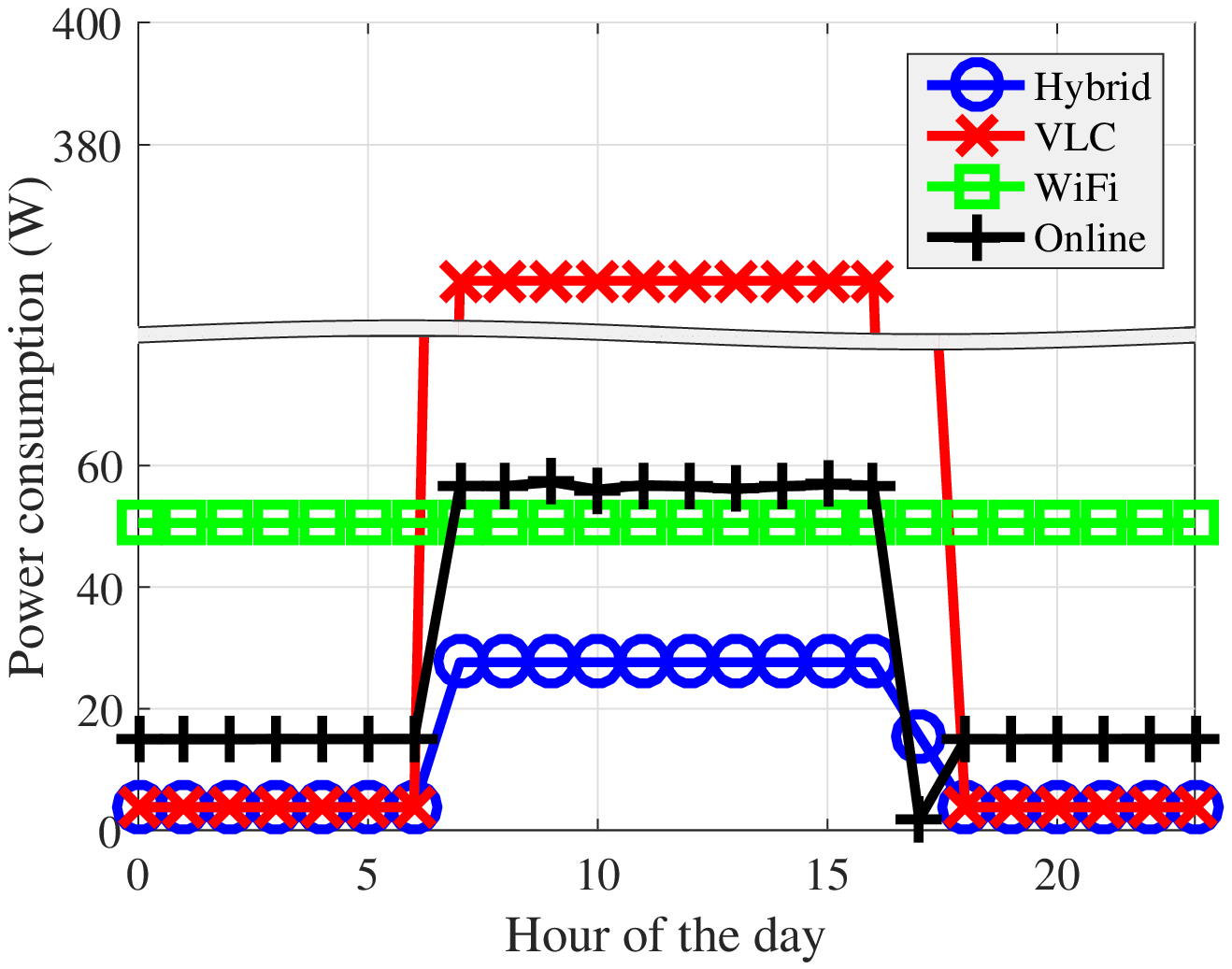}
        \caption{$\eta_{m}^{AC}=0.06$}
        \label{fig_hours006}
    \end{subfigure}
    ~ %add desired spacing between images, e. g. ~, \quad, \qquad, \hfill etc.
      %(or a blank line to force the subfigure onto a new line)
    \begin{subfigure}[b]{0.238\textwidth}
        \includegraphics[width=\textwidth]{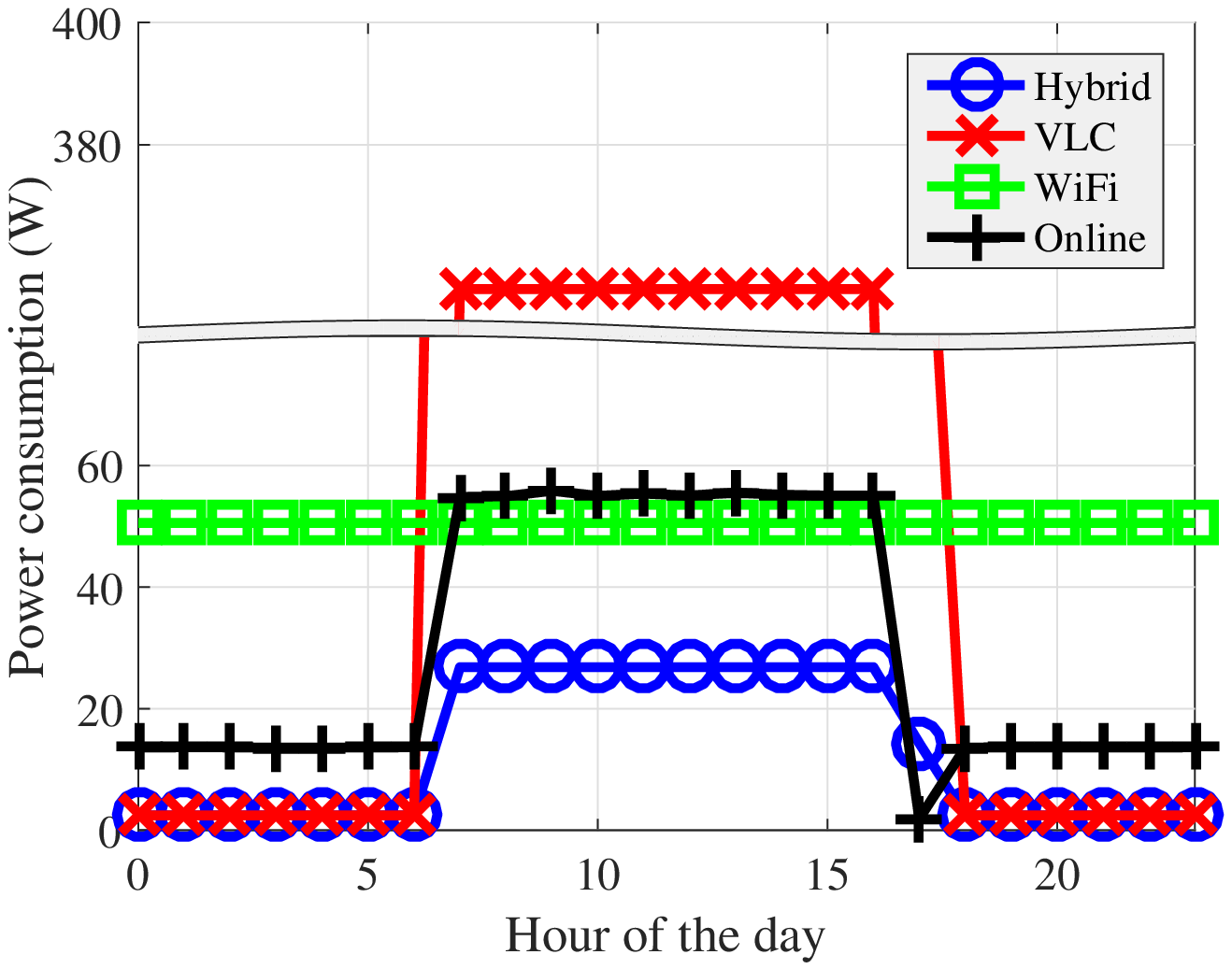}
        \caption{$\eta_{m}^{AC}=0.07$}
        \label{fig_hours007}
    \end{subfigure}
    ~ %add desired spacing between images, e. g. ~, \quad, \qquad, \hfill etc.
    %(or a blank line to force the subfigure onto a new line)
    \begin{subfigure}[b]{0.238\textwidth}
        \includegraphics[width=\textwidth]{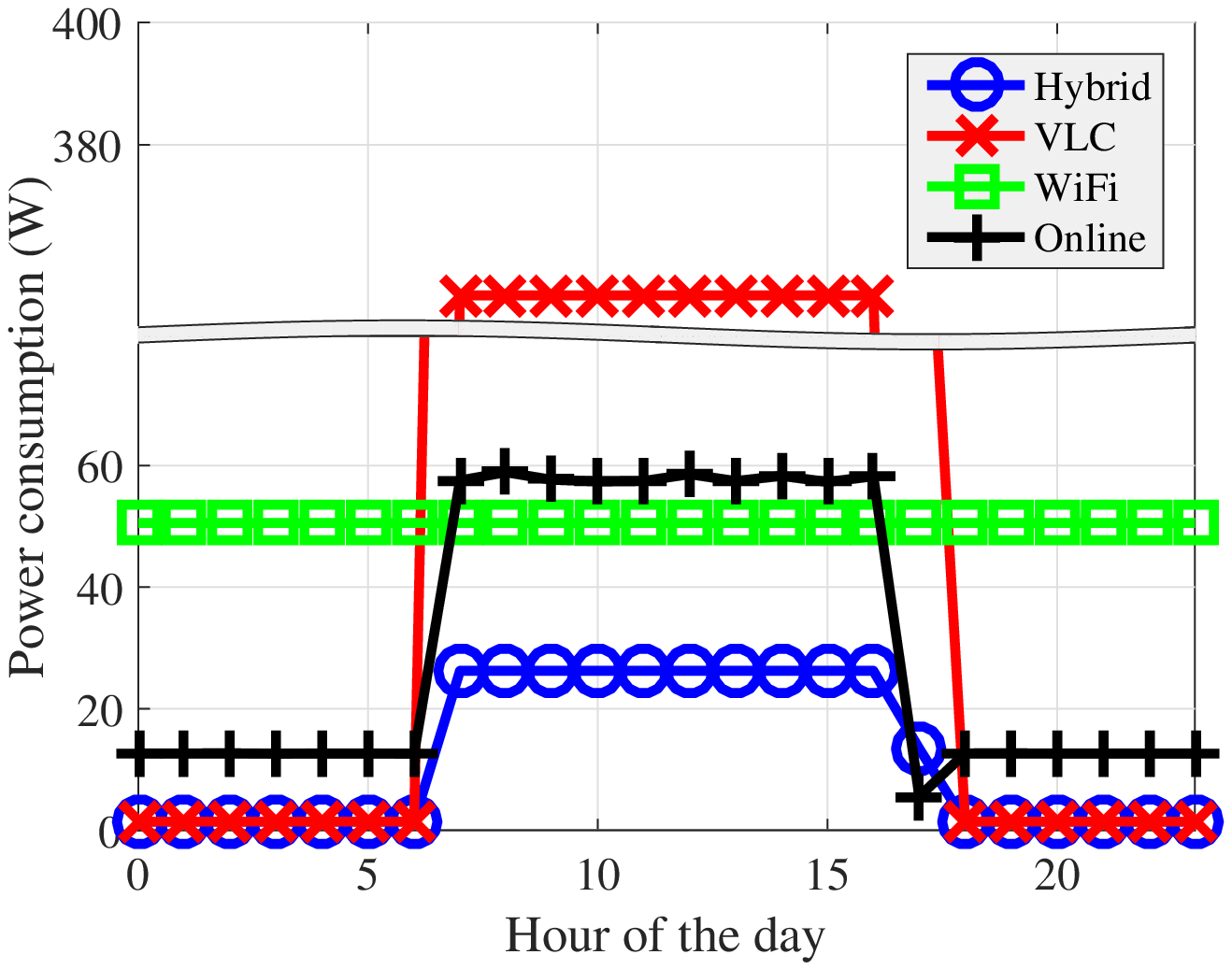}
        \caption{$\eta_{m}^{AC}=0.08$}
        \label{fig_hours008}
    \end{subfigure}
    \begin{subfigure}[b]{0.238\textwidth}
        \includegraphics[width=\textwidth]{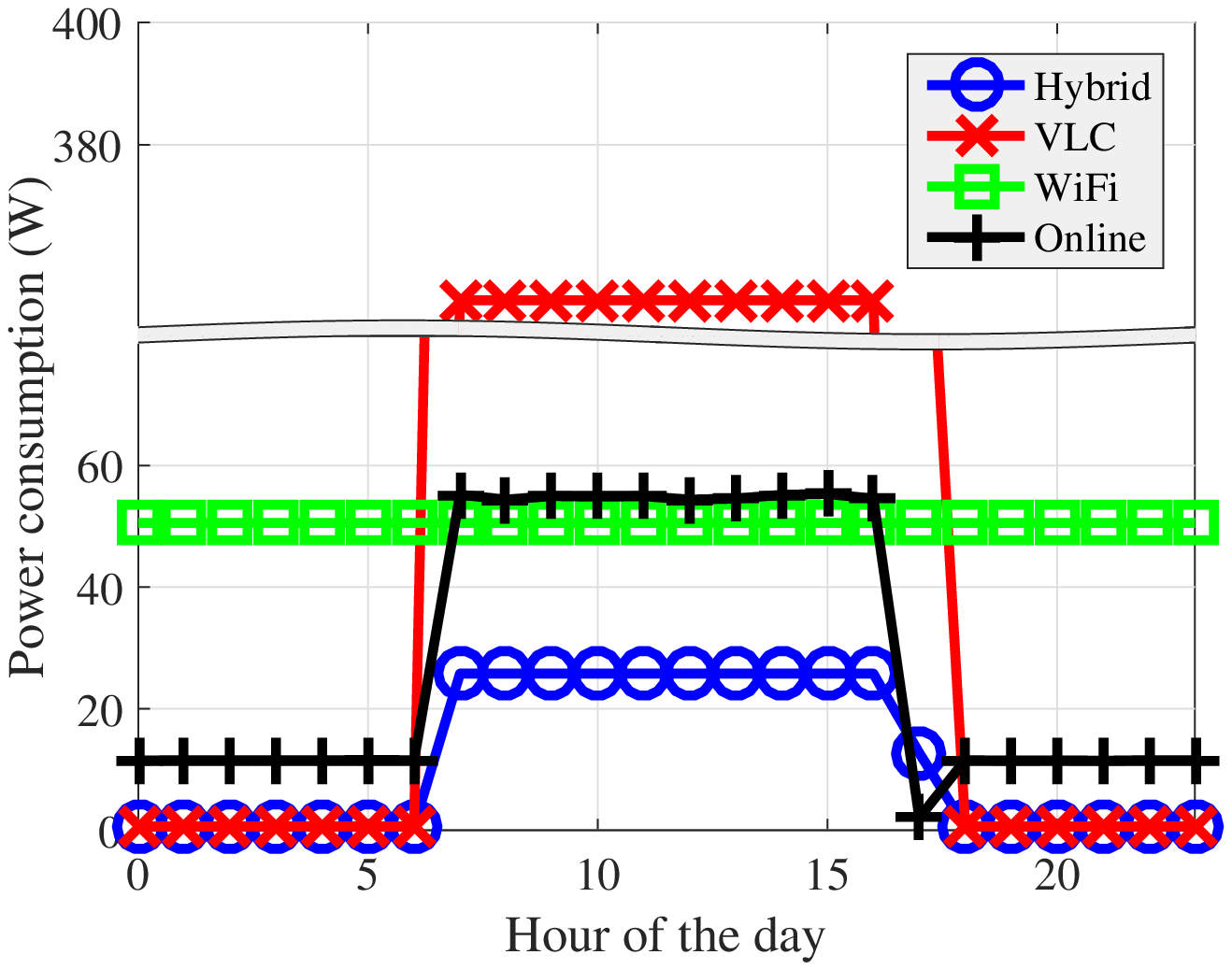}
        \caption{$\eta_{m}^{AC}=0.09$}
        \label{fig_hours009}
    \end{subfigure}
    \caption{Power consumption in terms of hour of the day}\label{fig_hours}
\end{figure*}

\txtredd{In Fig.~\ref{fig_online_throughput_night} to~\ref{fig_online_hours}, we present simulation results of two
modified versions of our proposed online algorithm. The first version, named ``Online WiFi'', is our online
algorithm when only the WiFi APs are considered. The second version, named ``Online VLC'', is our online algorithm
when only the VLC APs are considered. In Fig.~\ref{fig_online_throughput_night} and
Fig.~\ref{fig_online_throughput_day}, 100 UTs are distributed uniformly at random. At night
(Fig.~\ref{fig_online_throughput_night}), the power consumption of our proposed algorithm and the Online VLC
scheme are lower than that of the Online WiFi scheme, as all the VLC APs are turned on for illumination. The
sudden increase in power consumption for the Online WiFi scheme is due to turning on additional WiFi APs. During
the day (Fig.~\ref{fig_online_throughput_day}), the VLC APs may be forced to be turned on when the number of UTs
is large, resulting in unnecessary illumination and higher power consumption for the Online VLC scheme compared to
our proposed algorithm and the Online WiFi scheme. We also note that the performance of the proposed algorithm is
better than that of the Online WiFi scheme, since some of the UTs are located in the square region (\emph{i.e.,}
where there are no windows) and the VLC APs are already turned on for illumination. Our proposed algorithm takes
advantage of those turned on VLC APs to provide communication at a lower power consumption.

In Fig.~\ref{fig_online_numofUT_night} and Fig.~\ref{fig_online_numofUT_day}, the throughput requirement per UT is
set to 6 Mbps. During the night (Fig.~\ref{fig_online_numofUT_night}), where all the VLC APs are turned on for
illumination, our proposed algorithm and the Online VLC scheme perform better than the Online WiFi scheme. We also
note that the power consumption of the Online VLC scheme is similar to the power consumption of the optimal VLC
scheme shown in Fig. 9. This shows that our proposed algorithm can be used to determine the assignment of UTs to
APs at a much lower complexity while achieveing similar power consumptions.

In Fig.~\ref{fig_online_hours}, the number of UTs is set to 100 and the throughput requirement per UT is set to 6
Mbps. From the figure, we note that, from 6am to 7pm, the Online WiFi scheme consumes less power than the Online
VLC scheme due to the large power consumption of turning on VLC APs. We also see that our proposed hybrid
algorithm performs better than both the Online WiFi and the Online VLC schemes during the day. This is due to the
fact that our proposed algorithm can take advantage of the already turned on VLC APs for illumination in the
square region and the WiFi APs for the external rooms. From 7pm to 6am, since all VLC APs are turned on for
illumination, the power consumption of the Online VLC scheme and our proposed algorithm is less than the power
consumption of the Online WiFi.}

\begin{figure*}
    \centering
    \begin{subfigure}[b]{0.20\textwidth}
        \includegraphics[width=\textwidth]{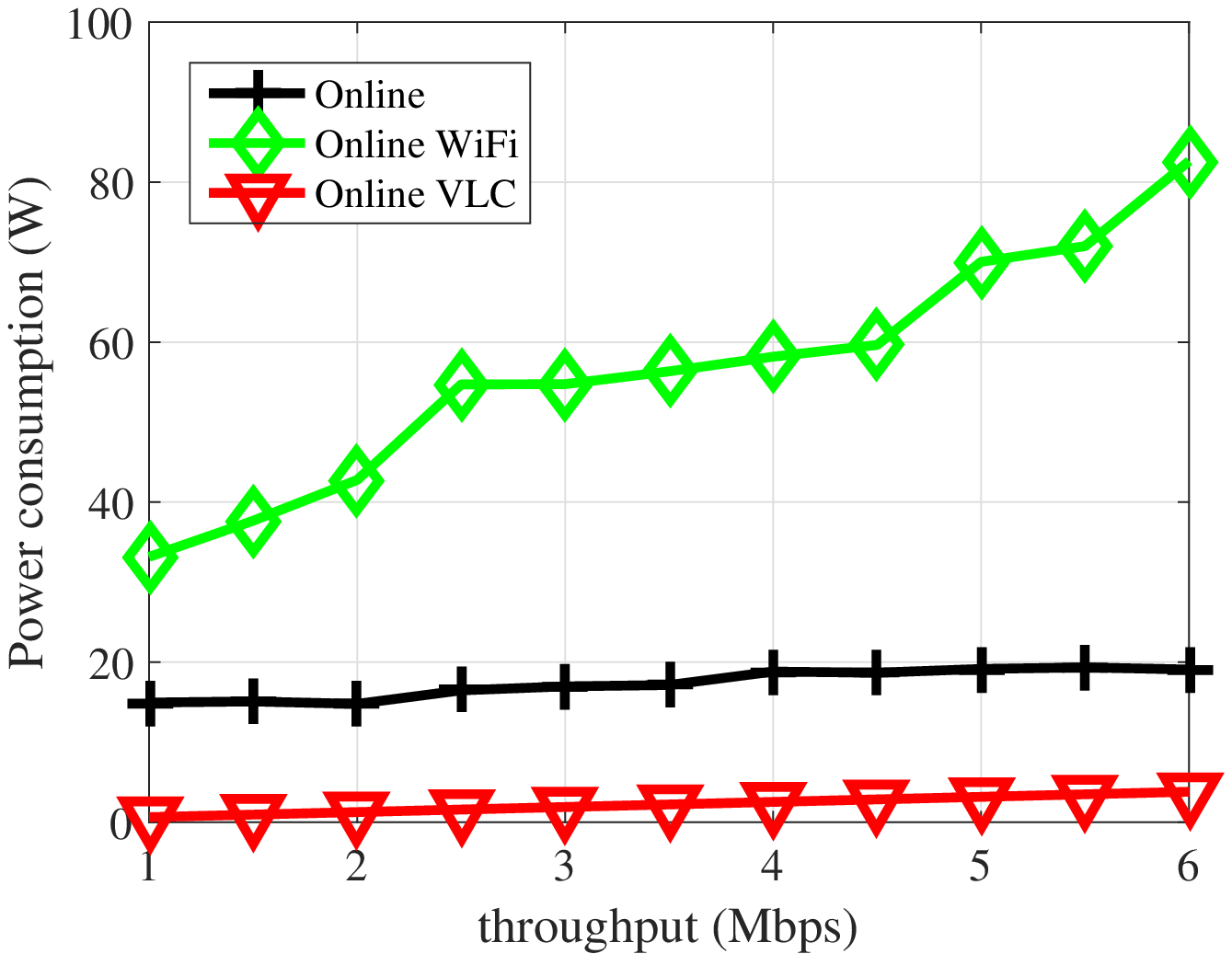}
        \caption{$\eta_{m}^{AC}=0.06$}
        \label{fig_online_throughput_night006}
    \end{subfigure}
    ~ %add desired spacing between images, e. g. ~, \quad, \qquad, \hfill etc.
      %(or a blank line to force the subfigure onto a new line)
    \begin{subfigure}[b]{0.20\textwidth}
        \includegraphics[width=\textwidth]{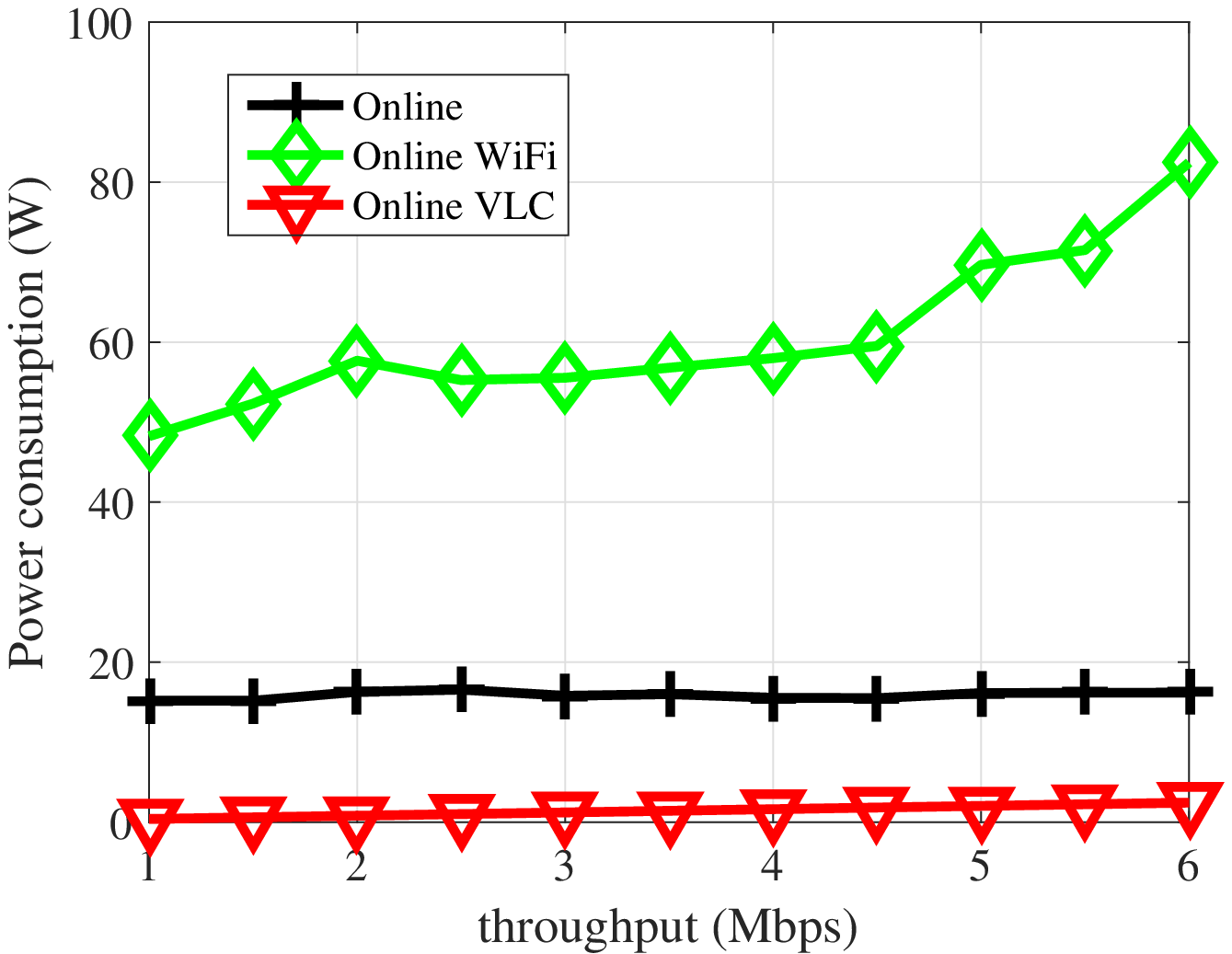}
        \caption{$\eta_{m}^{AC}=0.07$}
        \label{fig_online_throughput_night007}
    \end{subfigure}
    ~ %add desired spacing between images, e. g. ~, \quad, \qquad, \hfill etc.
    %(or a blank line to force the subfigure onto a new line)
    \begin{subfigure}[b]{0.20\textwidth}
        \includegraphics[width=\textwidth]{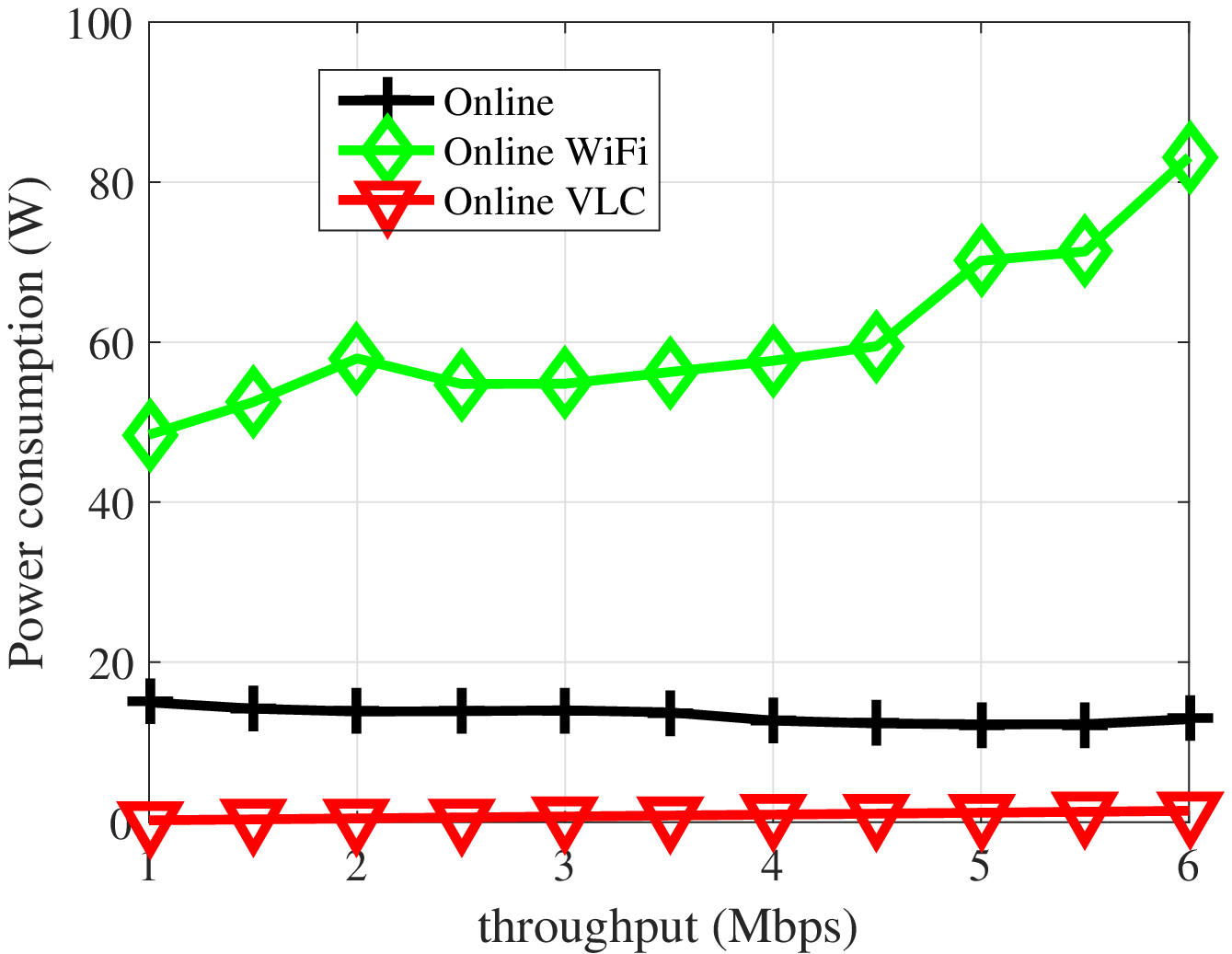}
        \caption{$\eta_{m}^{AC}=0.08$}
        \label{fig_online_throughput_night008}
    \end{subfigure}
    \begin{subfigure}[b]{0.20\textwidth}
        \includegraphics[width=\textwidth]{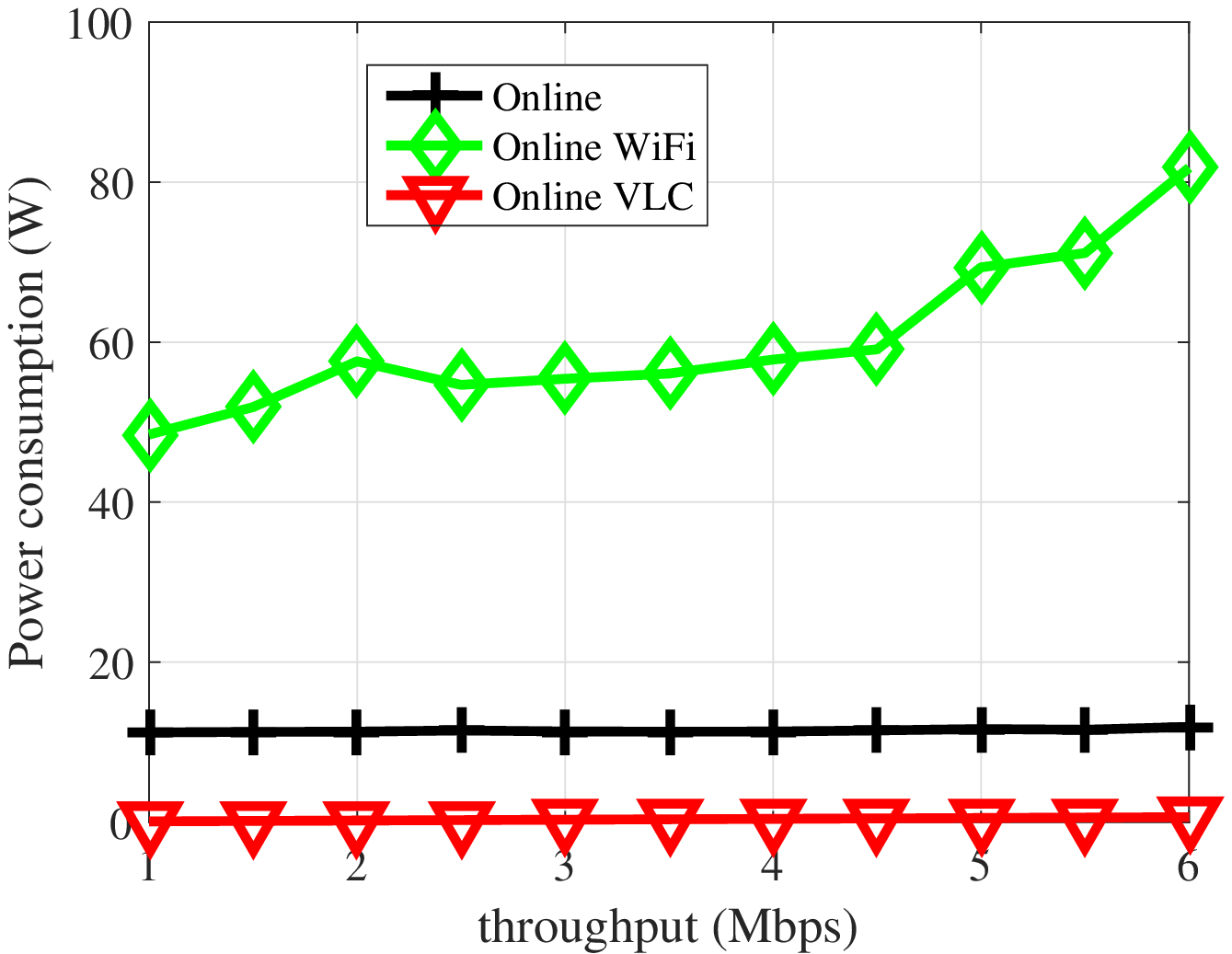}
        \caption{$\eta_{m}^{AC}=0.09$}
        \label{fig_online_throughput_night009}
    \end{subfigure}
    \caption{Power consumption in terms of throughput requirement per user at night}\label{fig_online_throughput_night}

    \begin{subfigure}[b]{0.20\textwidth}
        \includegraphics[width=\textwidth]{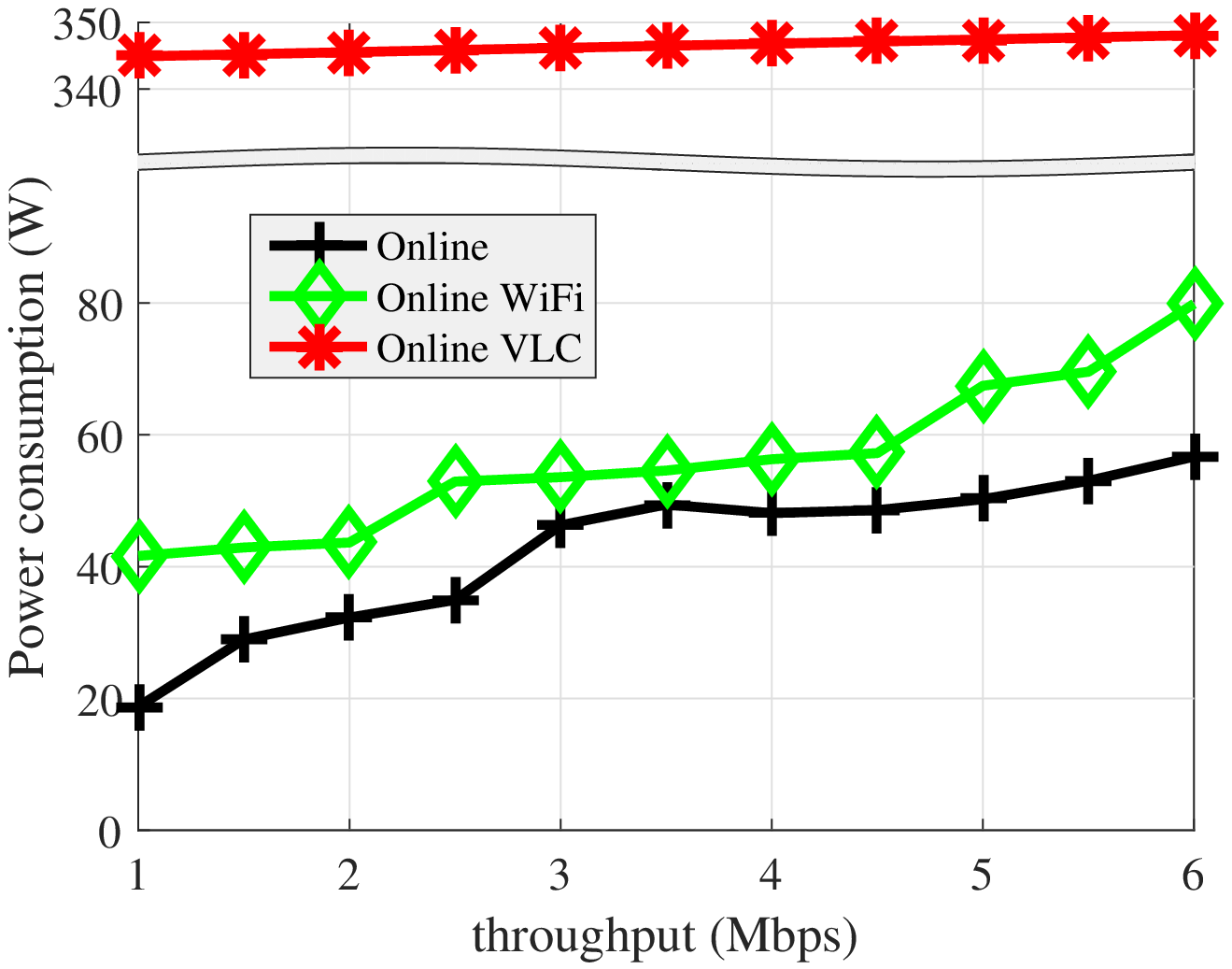}
        \caption{$\eta_{m}^{AC}=0.06$}
        \label{fig_online_throughput_day006}
    \end{subfigure}
    ~ %add desired spacing between images, e. g. ~, \quad, \qquad, \hfill etc.
      %(or a blank line to force the subfigure onto a new line)
    \begin{subfigure}[b]{0.20\textwidth}
        \includegraphics[width=\textwidth]{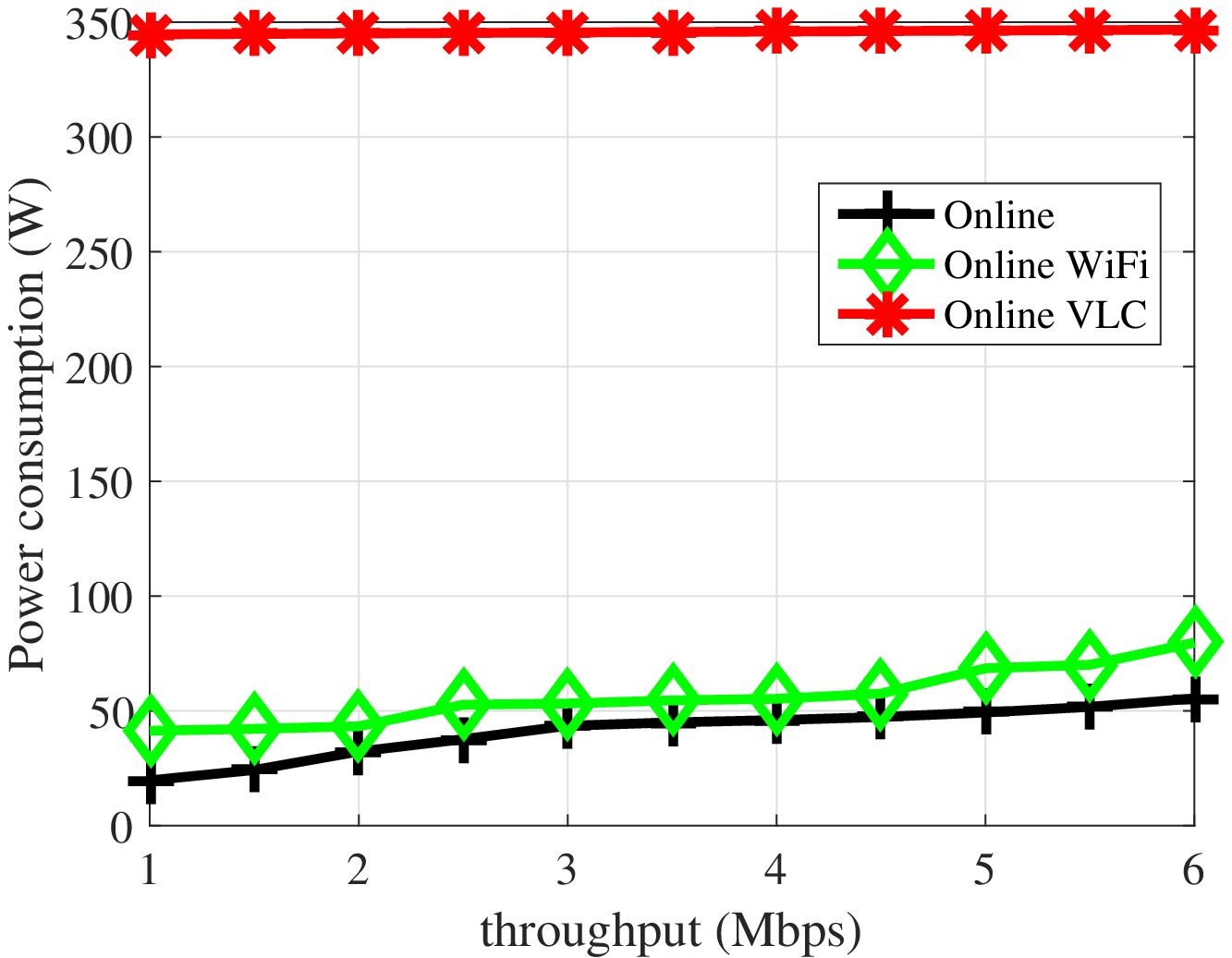}
        \caption{$\eta_{m}^{AC}=0.07$}
        \label{fig_online_throughput_day007}
    \end{subfigure}
    ~ %add desired spacing between images, e. g. ~, \quad, \qquad, \hfill etc.
    %(or a blank line to force the subfigure onto a new line)
    \begin{subfigure}[b]{0.20\textwidth}
        \includegraphics[width=\textwidth]{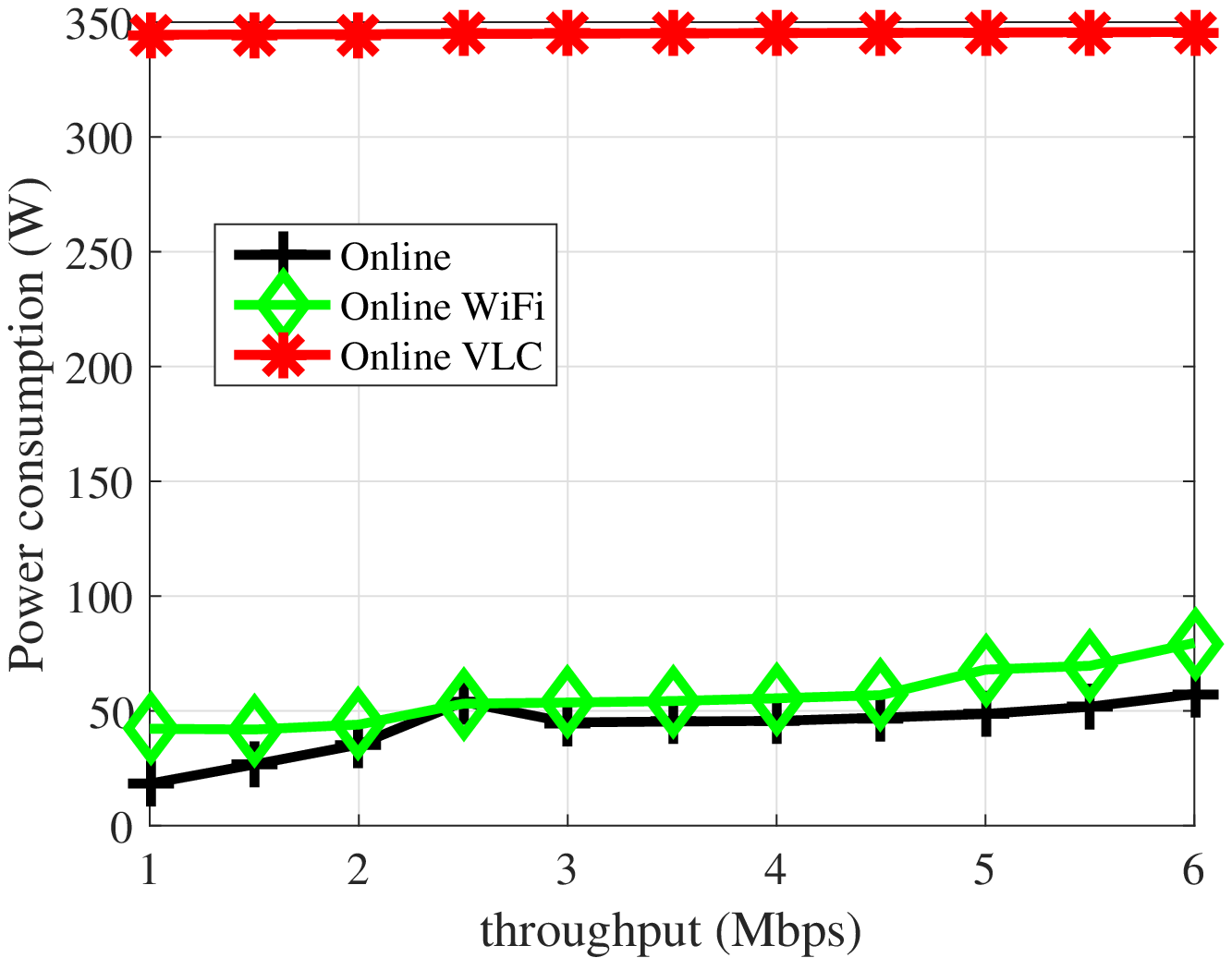}
        \caption{$\eta_{m}^{AC}=0.08$}
        \label{fig_online_throughput_day008}
    \end{subfigure}
    \begin{subfigure}[b]{0.20\textwidth}
        \includegraphics[width=\textwidth]{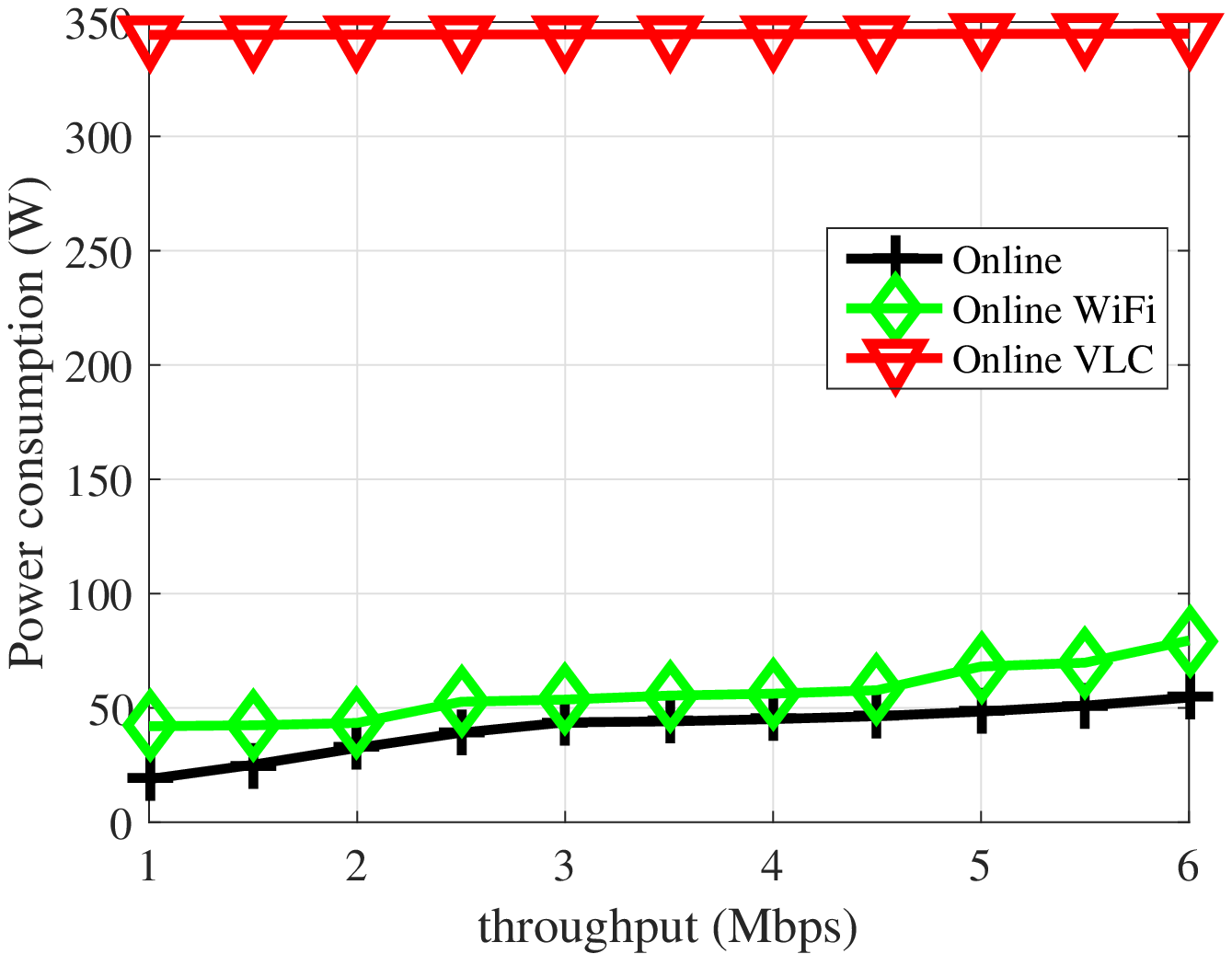}
        \caption{$\eta_{m}^{AC}=0.09$}
        \label{fig_online_throughput_day009}
    \end{subfigure}
    \caption{Power consumption in terms of throughput requirement per user at day}\label{fig_online_throughput_day}

    \begin{subfigure}[b]{0.20\textwidth}
        \includegraphics[width=\textwidth]{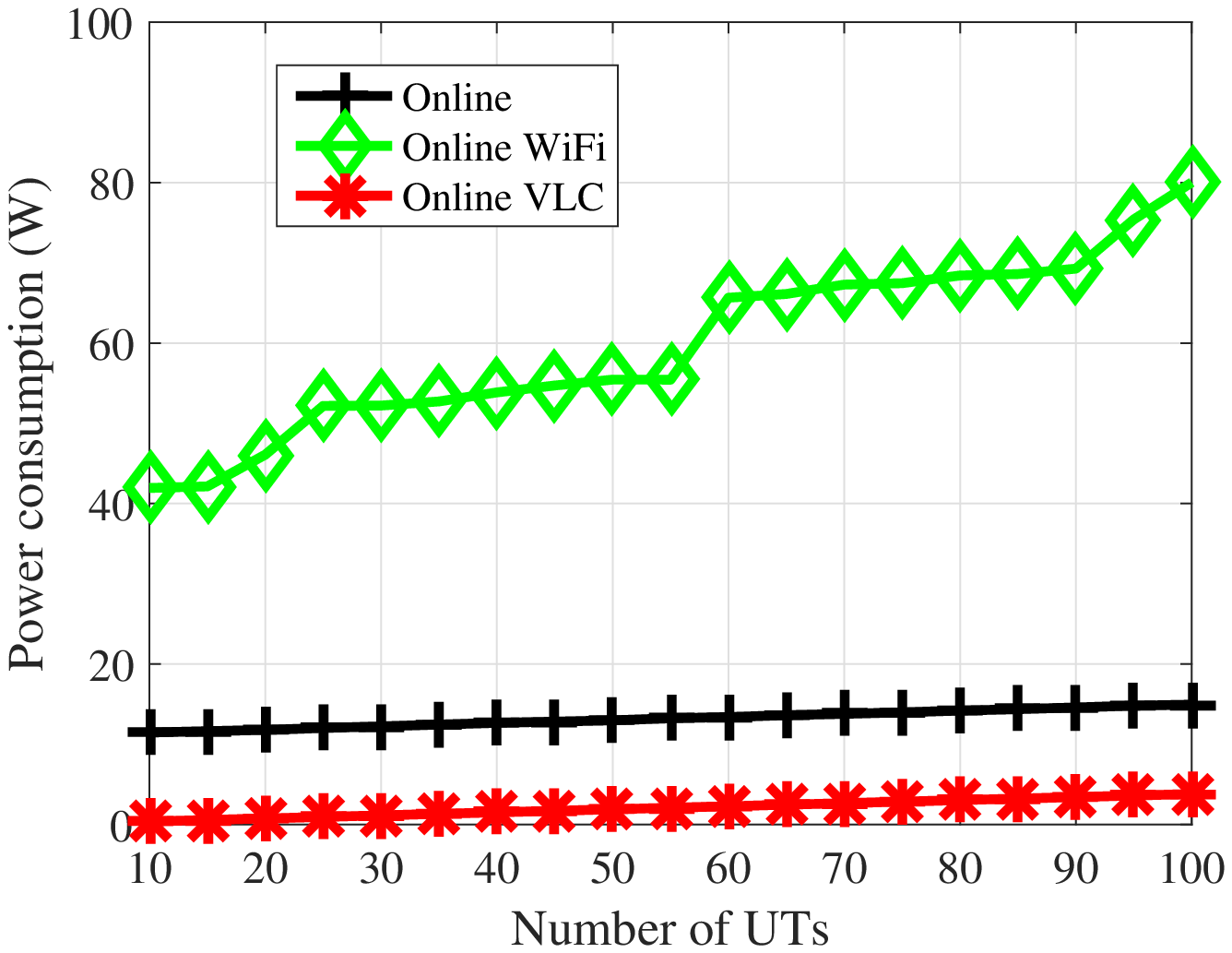}
        \caption{$\eta_{m}^{AC}=0.06$}
        \label{fig_online_numofUT_night006}
    \end{subfigure}
    ~ %add desired spacing between images, e. g. ~, \quad, \qquad, \hfill etc.
      %(or a blank line to force the subfigure onto a new line)
    \begin{subfigure}[b]{0.20\textwidth}
        \includegraphics[width=\textwidth]{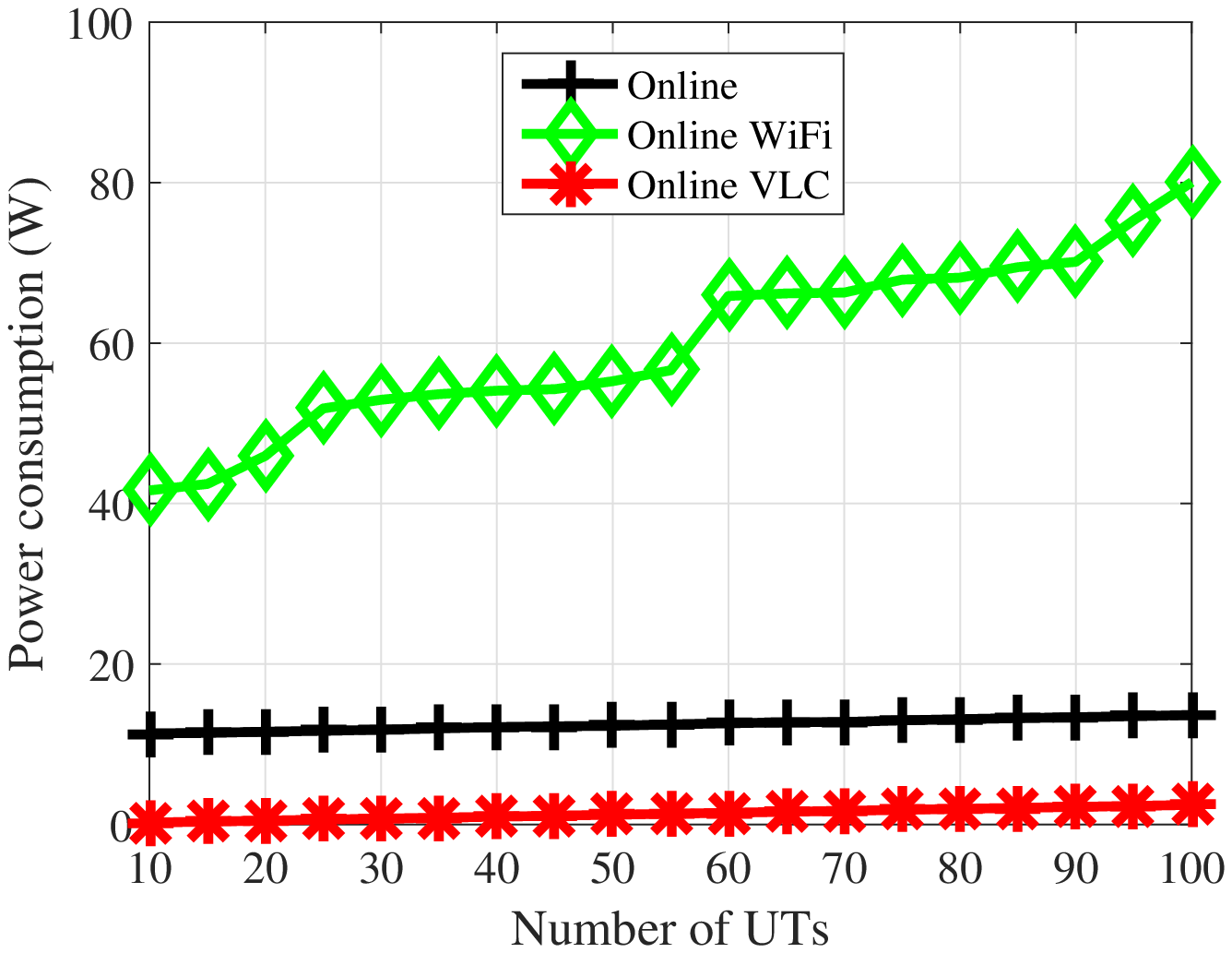}
        \caption{$\eta_{m}^{AC}=0.07$}
        \label{fig_online_numofUT_night007}
    \end{subfigure}
    ~ %add desired spacing between images, e. g. ~, \quad, \qquad, \hfill etc.
    %(or a blank line to force the subfigure onto a new line)
    \begin{subfigure}[b]{0.20\textwidth}
        \includegraphics[width=\textwidth]{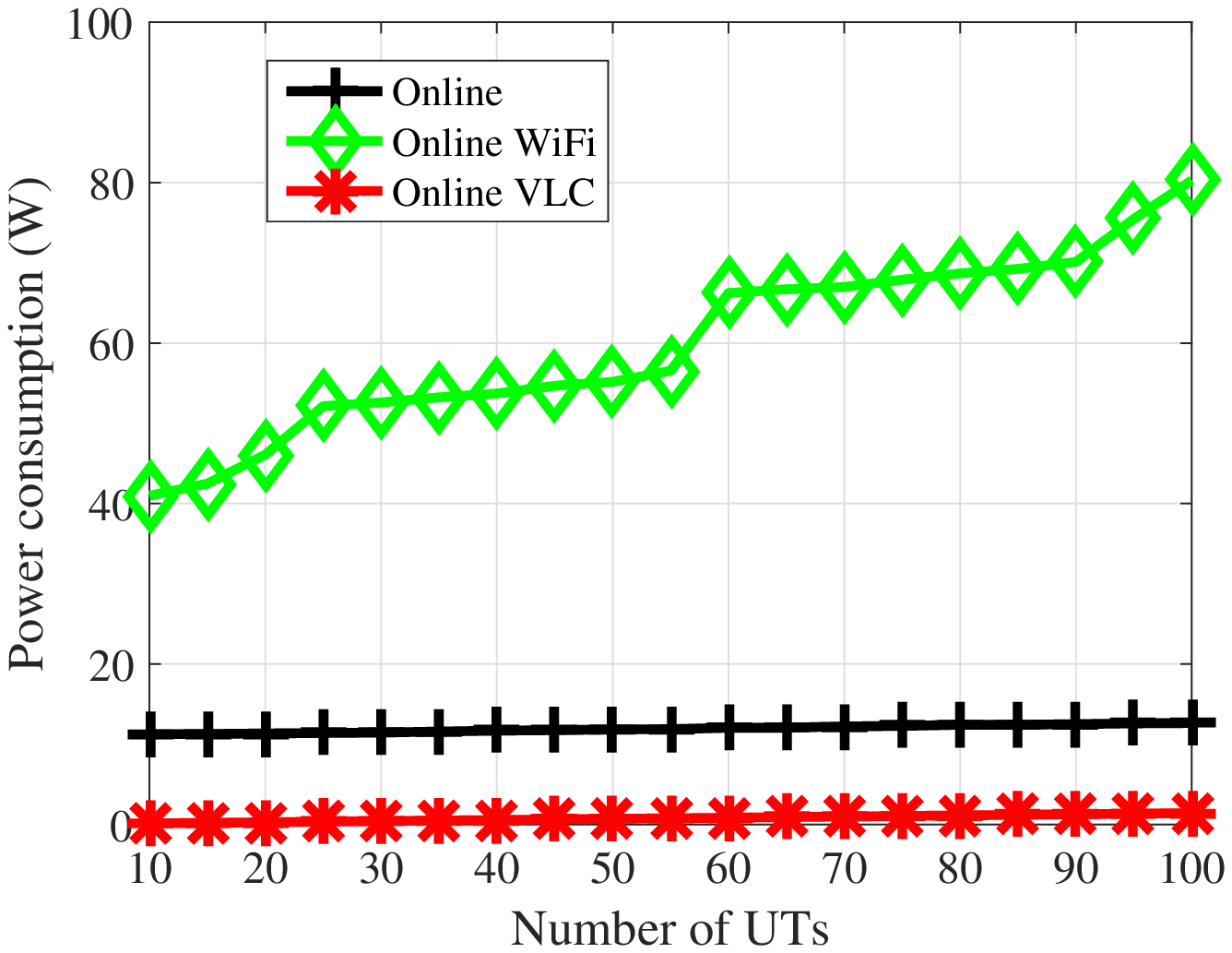}
        \caption{$\eta_{m}^{AC}=0.08$}
        \label{fig_online_numofUT_night008}
    \end{subfigure}
    \begin{subfigure}[b]{0.20\textwidth}
        \includegraphics[width=\textwidth]{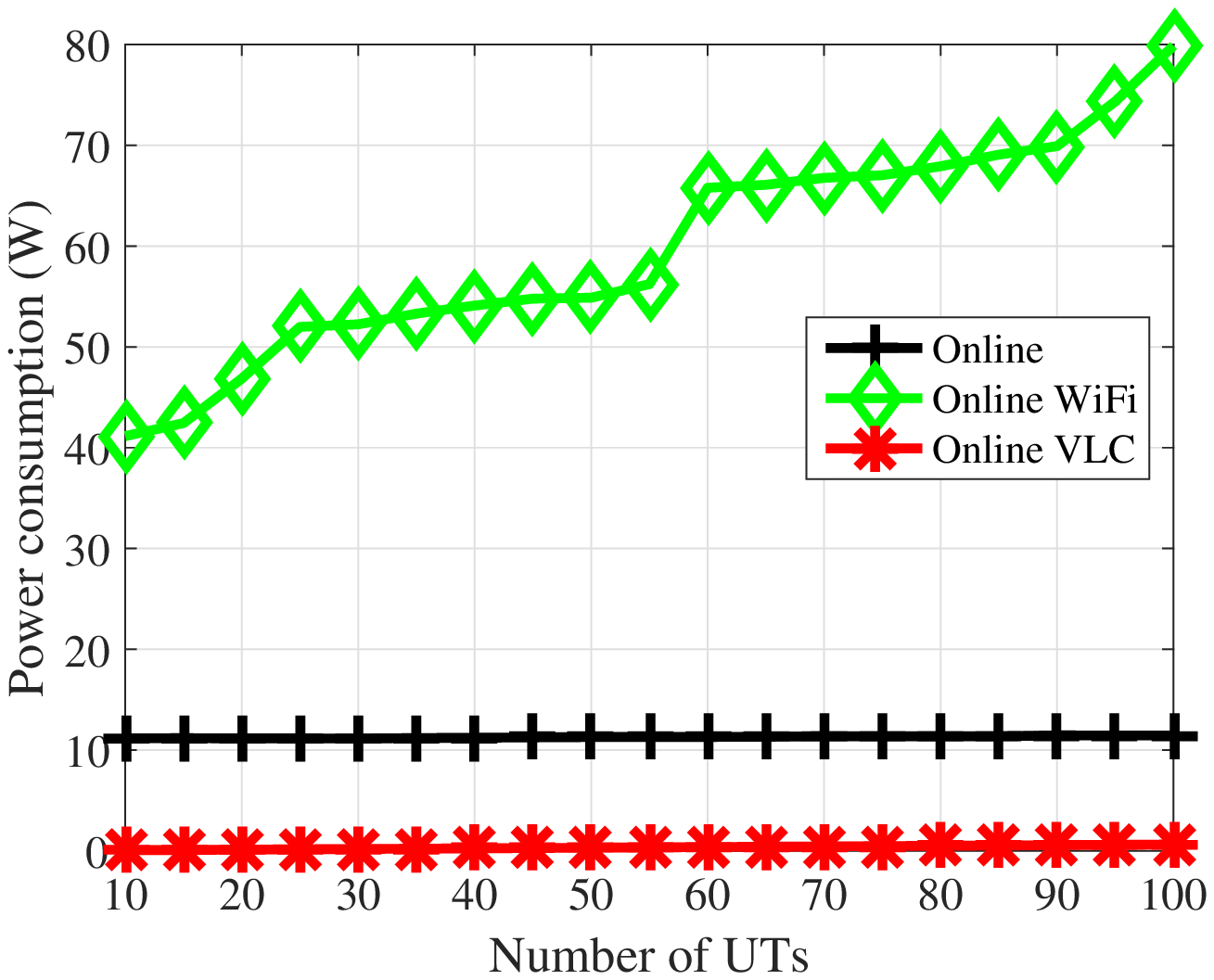}
        \caption{$\eta_{m}^{AC}=0.09$}
        \label{fig_online_numofUT_night009}
    \end{subfigure}
    \caption{Power consumption in terms of number of UTs at night}\label{fig_online_numofUT_night}

    \begin{subfigure}[b]{0.20\textwidth}
        \includegraphics[width=\textwidth]{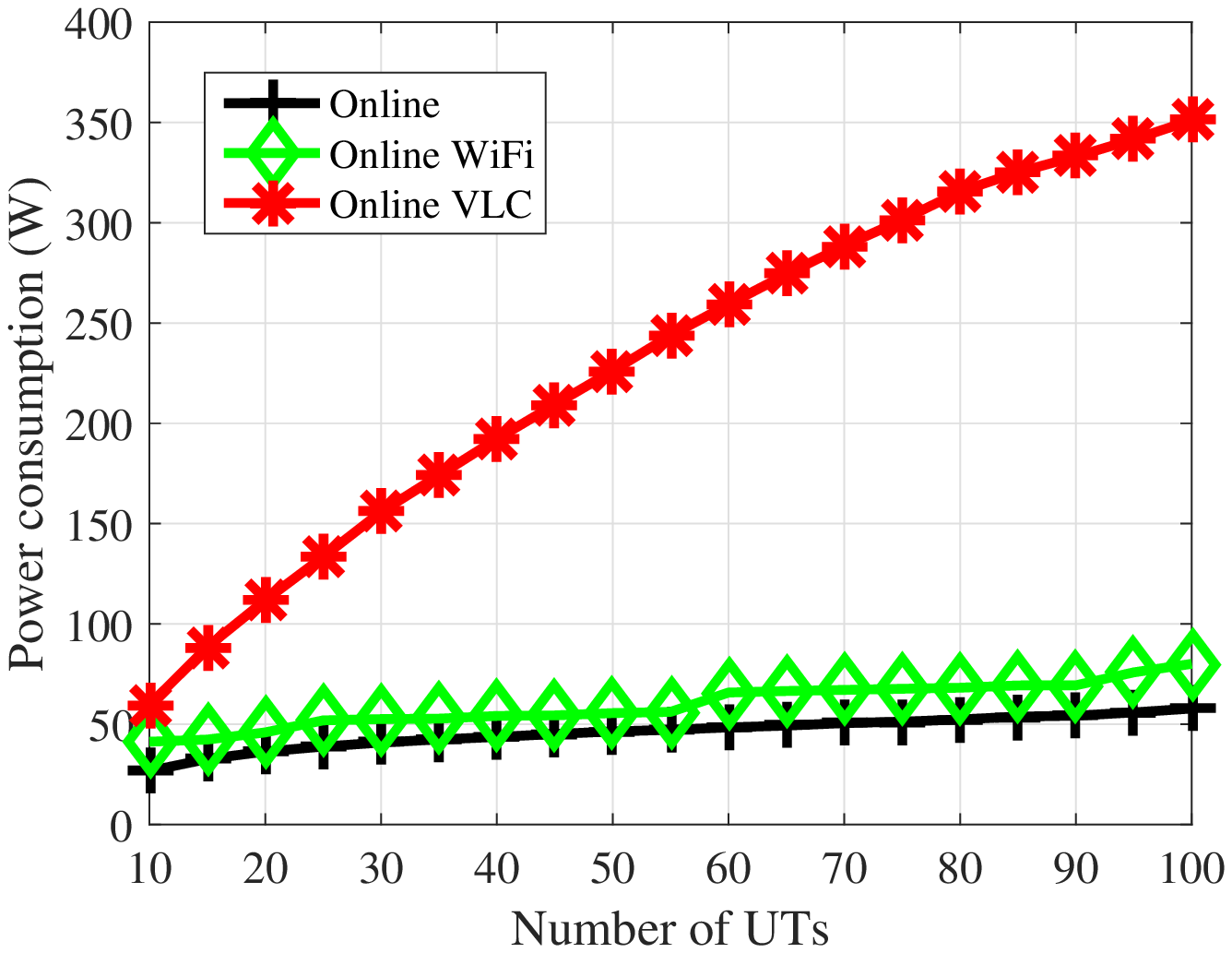}
        \caption{$\eta_{m}^{AC}=0.06$}
        \label{fig_online_numofUT_day006}
    \end{subfigure}
    ~ %add desired spacing between images, e. g. ~, \quad, \qquad, \hfill etc.
      %(or a blank line to force the subfigure onto a new line)
    \begin{subfigure}[b]{0.20\textwidth}
        \includegraphics[width=\textwidth]{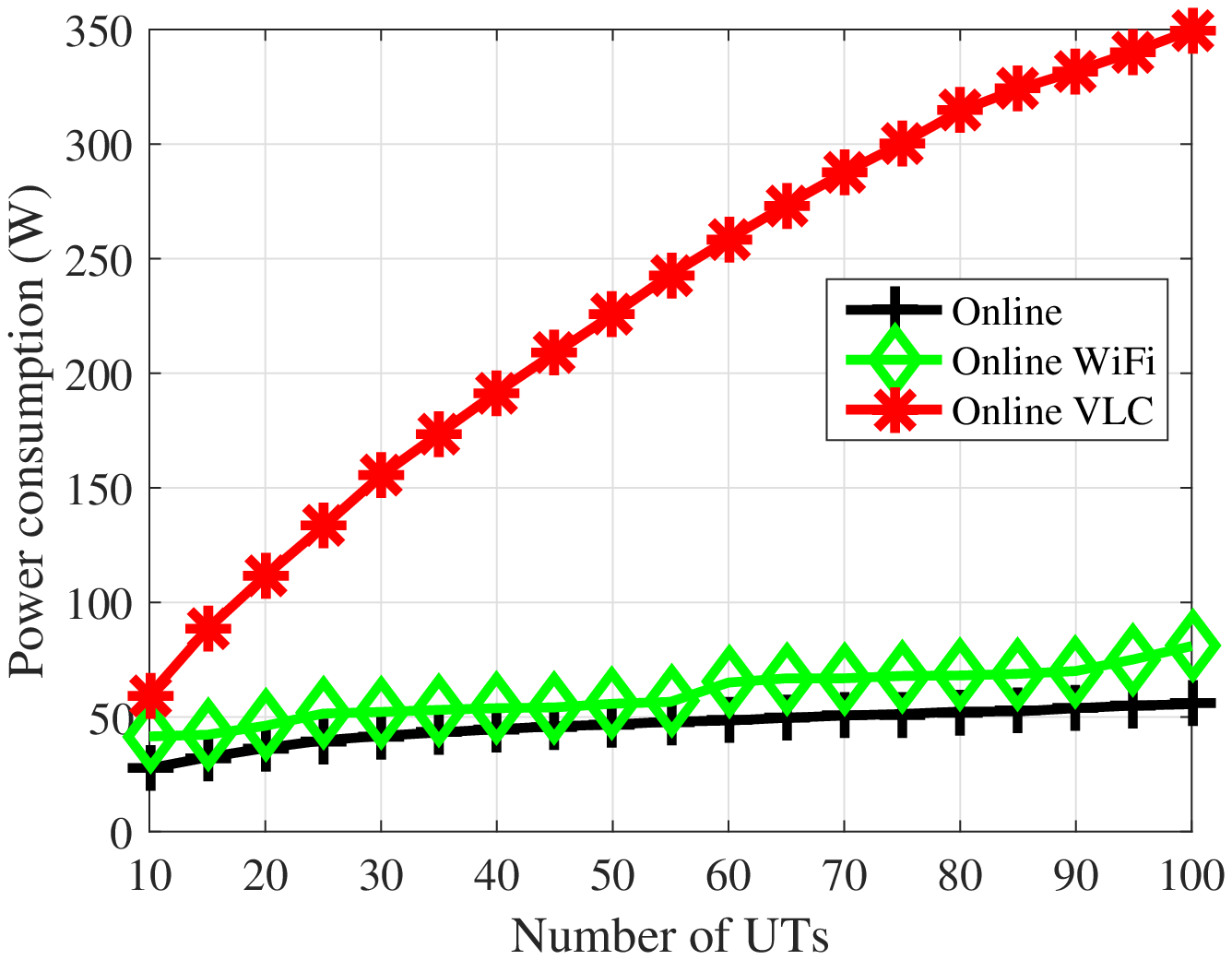}
        \caption{$\eta_{m}^{AC}=0.07$}
        \label{fig_online_numofUT_day007}
    \end{subfigure}
    ~ %add desired spacing between images, e. g. ~, \quad, \qquad, \hfill etc.
    %(or a blank line to force the subfigure onto a new line)
    \begin{subfigure}[b]{0.20\textwidth}
        \includegraphics[width=\textwidth]{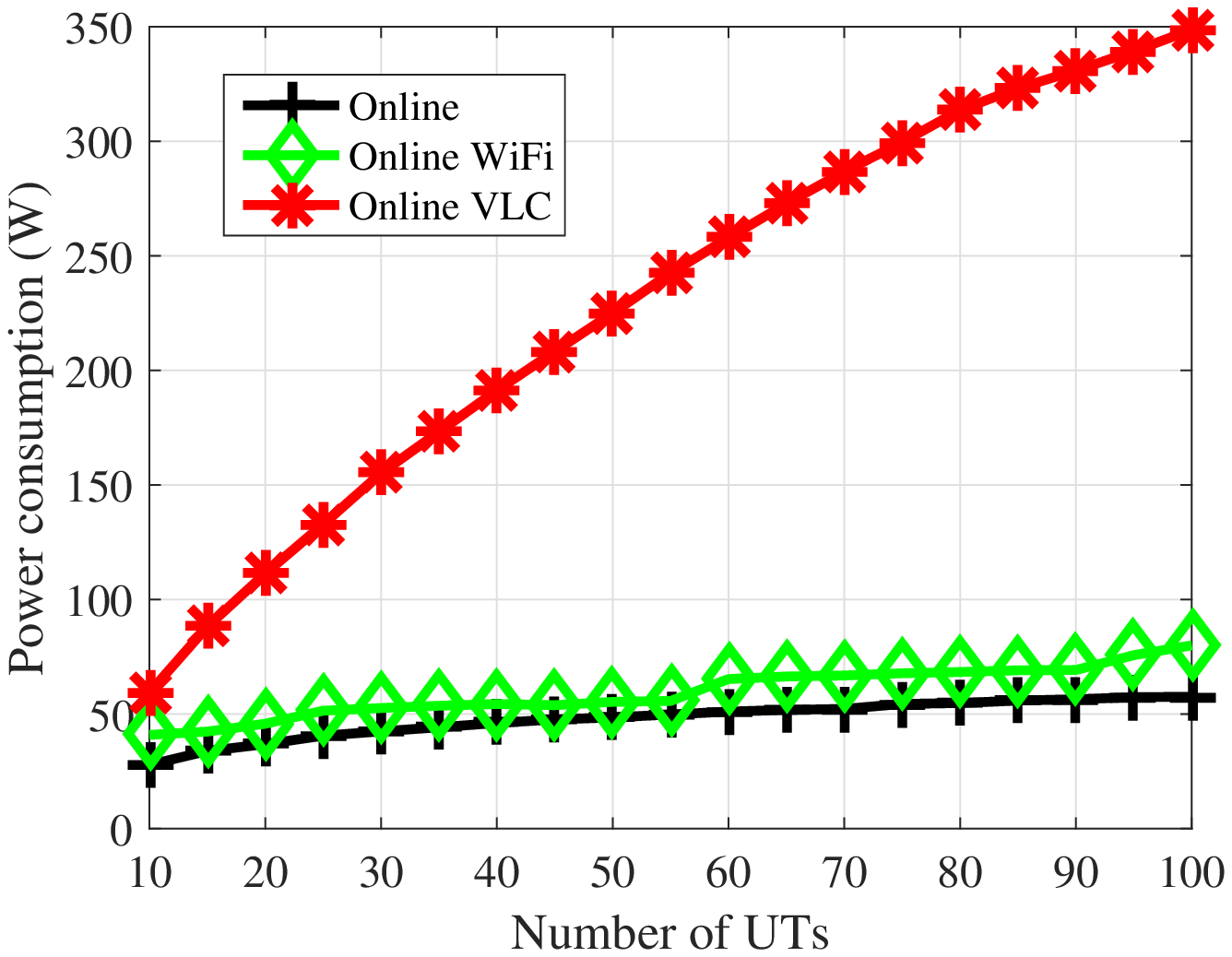}
        \caption{$\eta_{m}^{AC}=0.08$}
        \label{fig_online_numofUT_day008}
    \end{subfigure}
    \begin{subfigure}[b]{0.20\textwidth}
        \includegraphics[width=\textwidth]{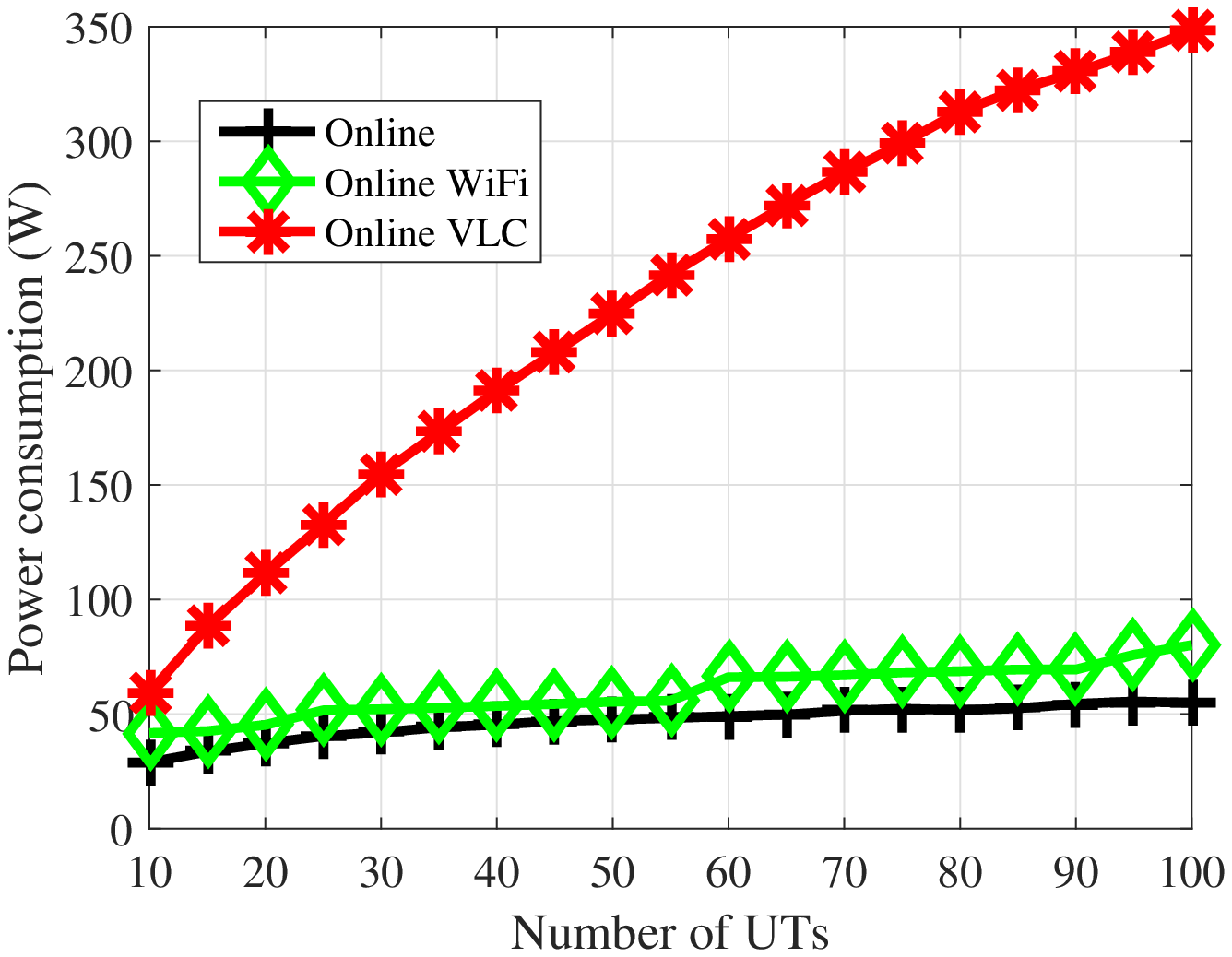}
        \caption{$\eta_{m}^{AC}=0.09$}
        \label{fig_online_numofUT_day009}
    \end{subfigure}
    \caption{Power consumption in terms of number of UTs at day}\label{fig_online_numofUT_day}

    \begin{subfigure}[b]{0.20\textwidth}
        \includegraphics[width=\textwidth]{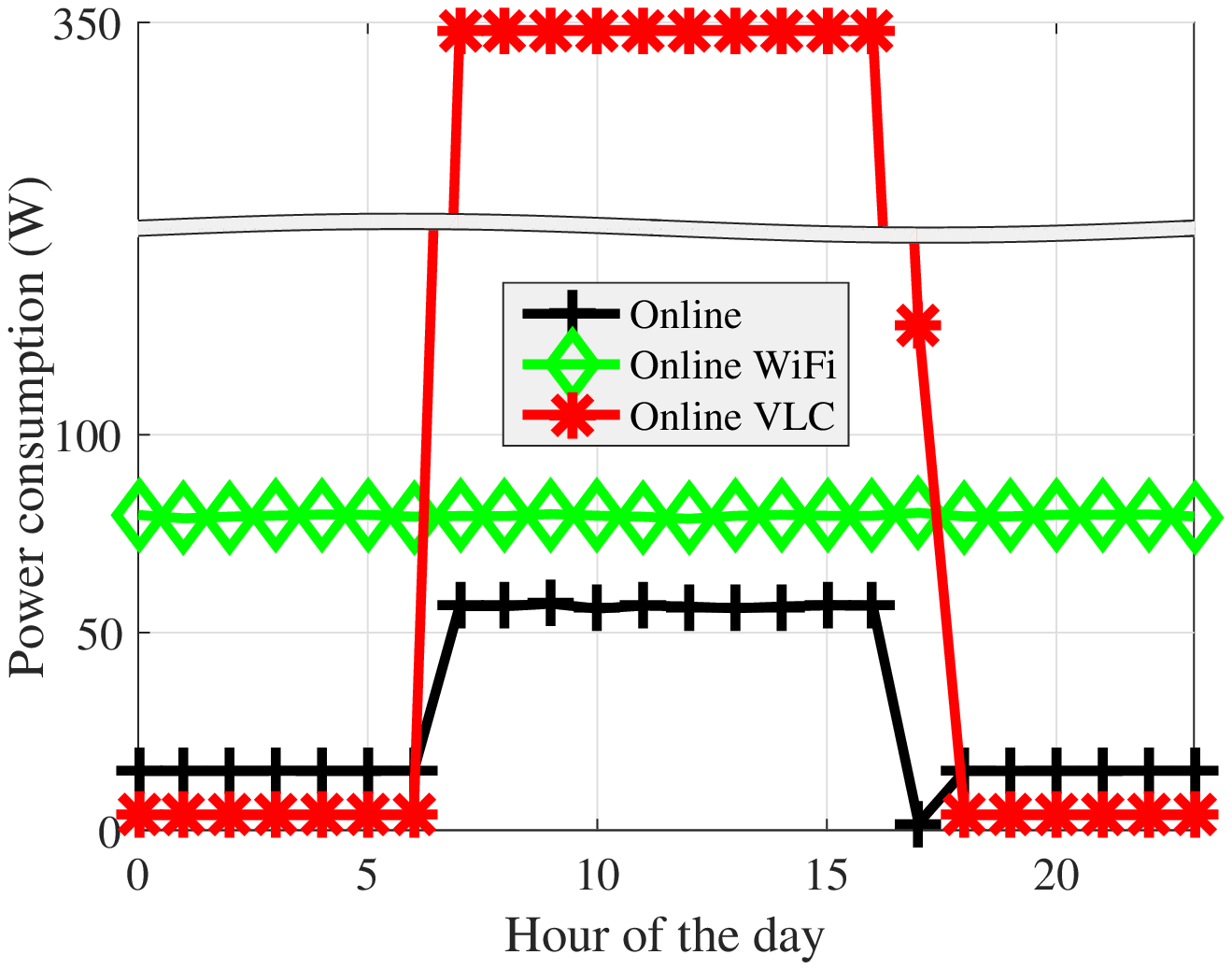}
        \caption{$\eta_{m}^{AC}=0.06$}
        \label{fig_online_hours006}
    \end{subfigure}
    ~ %add desired spacing between images, e. g. ~, \quad, \qquad, \hfill etc.
      %(or a blank line to force the subfigure onto a new line)
    \begin{subfigure}[b]{0.20\textwidth}
        \includegraphics[width=\textwidth]{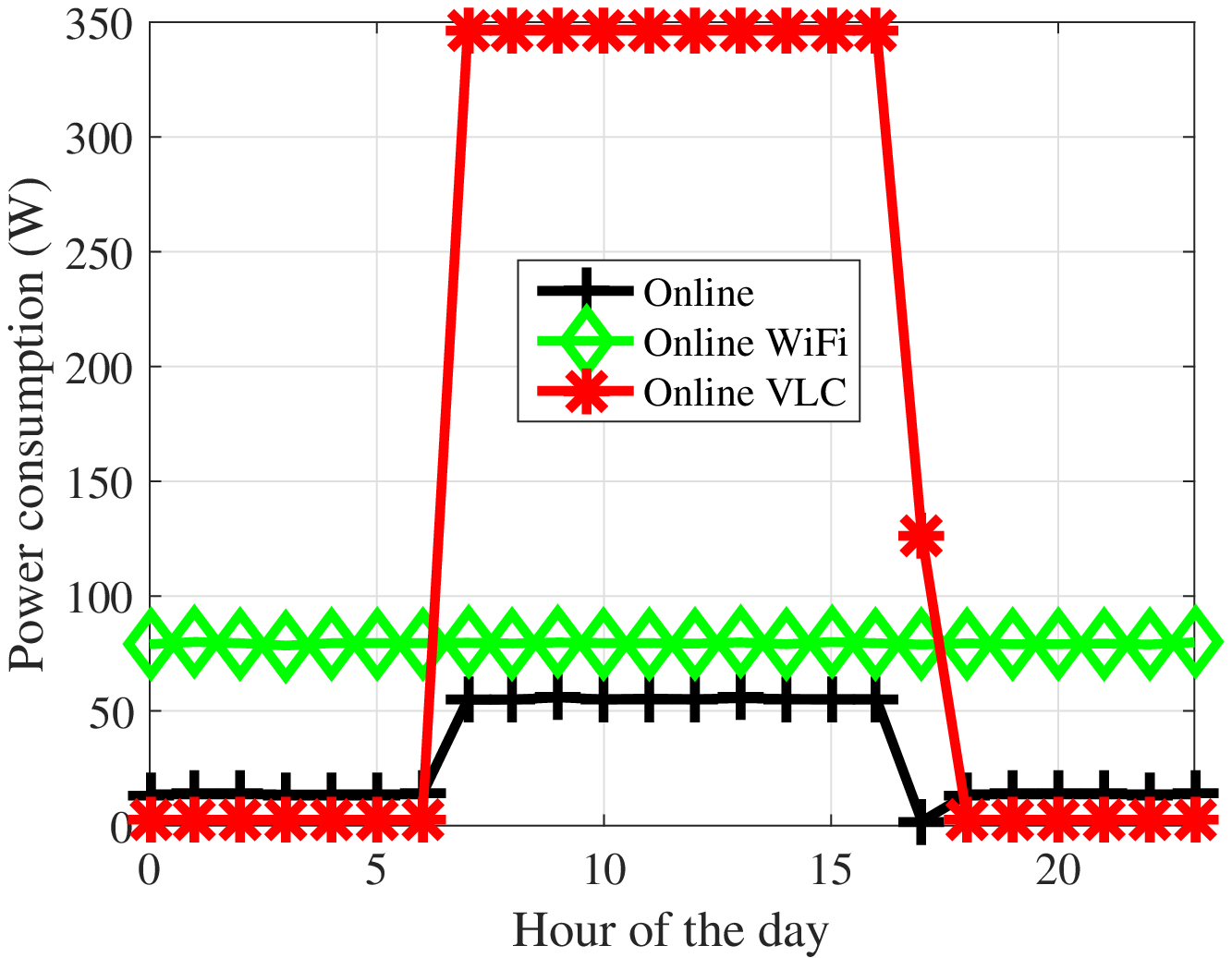}
        \caption{$\eta_{m}^{AC}=0.07$}
        \label{fig_online_hours007}
    \end{subfigure}
    ~ %add desired spacing between images, e. g. ~, \quad, \qquad, \hfill etc.
    %(or a blank line to force the subfigure onto a new line)
    \begin{subfigure}[b]{0.20\textwidth}
        \includegraphics[width=\textwidth]{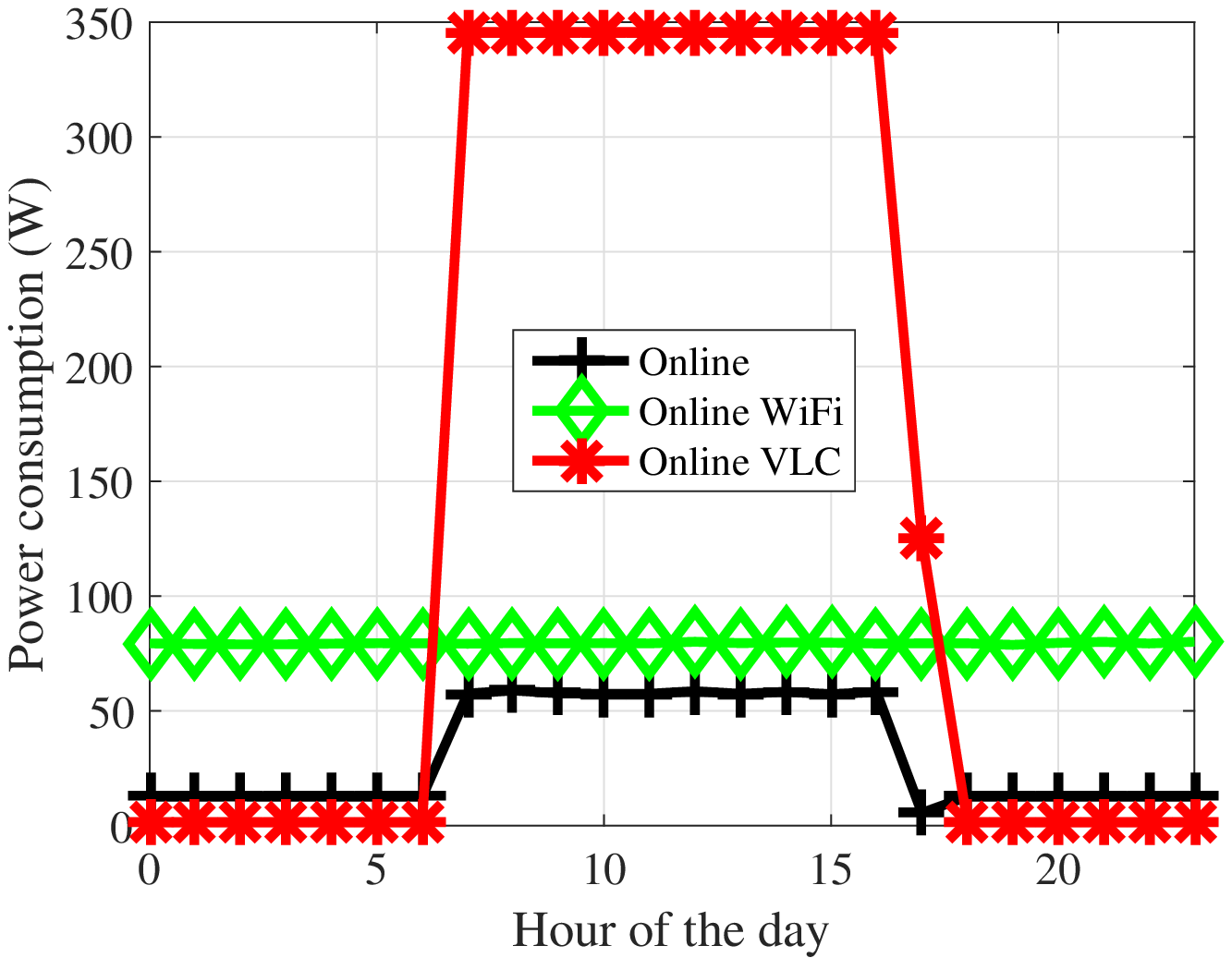}
        \caption{$\eta_{m}^{AC}=0.08$}
        \label{fig_online_hours008}
    \end{subfigure}
    \begin{subfigure}[b]{0.20\textwidth}
        \includegraphics[width=\textwidth]{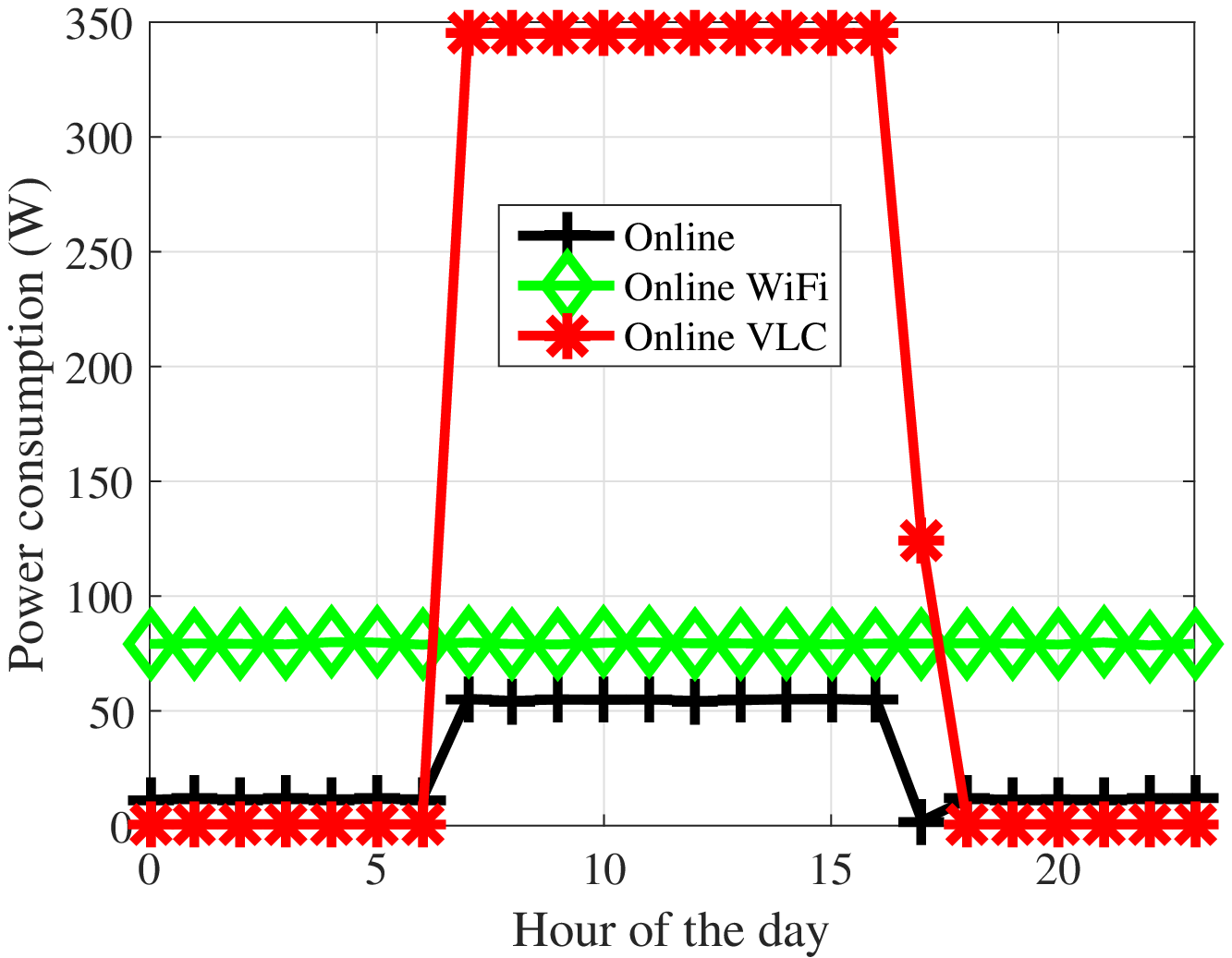}
        \caption{$\eta_{m}^{AC}=0.09$}
        \label{fig_online_hours009}
    \end{subfigure}
    \caption{Power consumption in terms of hour of the day}\label{fig_online_hours}
\end{figure*}

\section{Conclusion}\label{conc}
In this paper, we have tackled the problem of reducing the power consumption of wireless indoor access networks.
Unlike previous works, we utilize VLC, which can jointly provide communications and illumination, and thus greatly
reduces the power consumption. However, in a sunny day or with a large number of users, the RF access methods will
be more energy efficient than VLC. Our proposed system is a hybrid one comprising both VLC and WiFi access
methods. We formulate the problem of minimizing the \txtgreen{power} consumption of the system while satisfying
the requests of the users and achieving the desired illumination level. The problem is NP-complete. To alleviate
the complexity, we design an online algorithm for the problem with good competitive ratio. Our simulation results
show that the proposed hybrid system has the potential to reduce the \txtgreen{power} consumption by more than
75\% over the individual WiFi and VLC systems.

\begin{comment}
\section{acknowledgment}
This research was supported in part by NSF grant ECCS 1331018.
\end{comment}

\bibliographystyle{IEEEtran}
\bibliography{Abdallah}

\ifCLASSOPTIONcaptionsoff
  \newpage
\fi

\end{document}